\documentclass[a4paper,USenglish,cleveref,numberwithinsect]{lipics-v2021}
\pdfoutput=1 
\hideLIPIcs

\usepackage{
microtype,mathrsfs,amsmath,amssymb,amsthm,xcolor,amsfonts,latexsym,mathtools,graphicx,xspace,float,mathdots,caption,ellipsis,mleftright,dirtytalk,tikz,tikz-cd,silence}
\usepackage{textcomp}
\usepackage{lmodern}

\newtheorem{problem}[theorem]{Problem}\crefname{problem}{Problem}{Problems}
\newtheorem{sephypothesis}[theorem]{Separation Hypothesis}\crefname{sephypothesis}{Separation Hypothesis}{Separation Hypotheses}
\newtheorem{numhypothesis}[theorem]{Number-Theoretic Hypothesis}\crefname{numhypothesis}{Number-Theoretic Hypothesis}{Number-Theoretic Hypotheses}
\newtheorem{openproblem}[theorem]{Conjecture}\crefname{openproblem}{Conjecture}{Conjectures}

\DeclareMathOperator{\rk}{rk}

\DeclareMathOperator{\poly}{poly}
\DeclareMathOperator{\rad}{rad}
\DeclareMathOperator{\sep}{sep}

\DeclareMathOperator{\ii}{i}
\DeclareMathOperator{\Log}{Log}
\DeclareMathOperator{\Exp}{Exp}
\DeclareMathOperator{\dist}{dist}
\newcommand{\distLD}{\delta_{\log}}

\DeclareMathOperator{\Lie}{Lie}
\DeclareMathOperator{\diag}{diag}
\DeclareMathOperator{\Crit}{Crit}

\DeclareMathOperator{\ran}{im}
\DeclareMathOperator{\inter}{int}

\newcommand{\QQ}{\mathbb{Q}}
\newcommand{\CC}{\mathbb{C}}
\newcommand{\ZZ}{\mathbb{Z}}
\newcommand{\RR}{\mathbb{R}}
\newcommand{\NN}{\mathbb{Z}_{\geq 0}}
\newcommand{\OO}{\ensuremath{\mathcal{O}}}
\newcommand{\KK}{\ensuremath{\mathcal{C}}}

\newcommand{\NP}{\ensuremath{\mathsf{NP}}\xspace}

\newcommand{\eps}{\varepsilon}

\newcommand{\abc}{\texorpdfstring{$abc$}{abc}}

\newcommand{\Xquot}{X/\hspace{-1ex}\sim}
\DeclarePairedDelimiter\abs{\lvert}{\rvert}
\DeclarePairedDelimiter\norm{\lVert}{\rVert}

\numberwithin{equation}{section}

\allowdisplaybreaks[4]
\title{Complexity of Robust Orbit Problems for Torus Actions and the \abc-conjecture}

\author{Peter B\"{u}rgisser}{Institute of Mathematics, Technische Universit\"{a}t Berlin, Germany.}{pbuerg@math.tu-berlin.de}{}{Supported by the ERC under the European Union's Horizon 2020 research and innovation programme (grant agreement no.~787840).}
\author{Mahmut Levent Do\u{g}an}{Institute of Mathematics, Technische Universit\"{a}t Berlin, Germany.}{dogan@math.tu-berlin.de}{}{Supported by the ERC under the European Union's Horizon 2020 research and innovation programme (grant agreement no.~787840).}

\author{Visu Makam}{Radix Trading, Amsterdam.}{visu@umich.edu}{}{Partially supported by NSF Grant CCF-1900460 and the University of Melbourne.}

\author{Michael Walter}{Faculty of Computer Science, Ruhr-Universit\"at Bochum, Germany.}{michael.walter@rub.de}{}{Supported by the European Union (ERC, SYMOPTIC, 101040907), by the Deutsche Forschungsgemeinschaft (DFG, German Research Foundation) under Germany's Excellence Strategy - EXC\ 2092\ CASA - 390781972, by the BMBF (QuBRA, 13N16135), and by the Dutch Research Council (NWO grant OCENW.KLEIN.267).}

\author{Avi Wigderson}{School of Mathematics, Institute for Advanced Study, Princeton.}{avi@ias.edu}{}{Supported by the NSF Grant CCF-1900460.}

\authorrunning{P. B\"{u}rgisser, M. L. Do\u{g}an, V. Makam, M. Walter, and A. Wigderson} 

\Copyright{Peter B\"{u}rgisser, Mahmut Levent Do\u{g}an, Visu Makam, Michael Walter, and Avi Wigderson} 


\category{} 

\relatedversion{} 

\acknowledgements{The authors thank Mat\'{i}as Bender, Alperen A. Erg\"{u}r, Jonathan Leake, and Philipp Reichenbach for productive discussions.}

\nolinenumbers

\begin{document}

\maketitle

\begin{abstract}
When a group acts on a set, it naturally partitions it into orbits, giving rise to \emph{orbit problems}.
These are natural algorithmic problems, as symmetries are central in numerous questions and structures in physics, mathematics, computer science, optimization, and more.
Accordingly, it is of high interest to understand their computational complexity.
Recently,~\cite{torus} gave the first polynomial-time algorithms for orbit problems of \emph{torus actions}, that is, actions of commutative continuous groups on Euclidean space.
In this work, motivated by theoretical and practical applications, we study the computational complexity of \textit{robust} generalizations of these orbit problems, which amount to \emph{approximating} the distance of orbits in $\CC^n$ up to a factor~$\gamma \ge 1$.
In particular, this allows deciding whether two inputs are approximately in the same orbit or far from being so.
On the one hand, we prove the NP-hardness of this problem for~$\gamma = n^{\Omega(1/\log\log n)}$ by reducing the closest vector problem for lattices to it.
On the other hand, we describe algorithms for solving this problem for an approximation factor~$\gamma = \exp(\poly(n))$.
Our algorithms combine tools from invariant theory and algorithmic lattice theory, and they also provide group elements witnessing the proximity of the given orbits (in contrast to the algebraic algorithms of prior work).
We prove that they run in polynomial time if and only if a version of the famous number-theoretic $abc$-conjecture holds -- establishing
a new and surprising connection between computational complexity and number theory.
\end{abstract}

\section{Introduction}
\label{sec:intro}

Computational invariant theory was a central topic of 19th century algebra,
see the historical account in~\cite{kung-rota:84}.
In the second half of the 20th century, structural progress was made through Mumford's invention
of geometric invariant theory ~\cite{MFK:94} and computational progress came through the theory of Gr\"{o}bner bases, see~\cite{Sturmfels}.
More recently, it was realized that algorithmic questions in invariant and representation theory
are deeply connected with the core complexity questions of P vs.~NP  and VP vs.~VNP.

On the one hand,
Mulmuley and Sohoni's Geometric Complexity Theory (GCT)~\cite{mulmuley2001geometric}
highlights the inherent symmetries of complete problems of these complexity classes.
This was the starting point of several specific invariant theoretic and
representation theoretic attacks on the VP vs.~VNP questions, which has led
to many new questions, techniques, and much faster algorithms:
for example, see~\cite{GCTV,FS,BCMW:17,MW19}.

The other connection is through
the work of Impagliazzo and Kabanets~\cite{KI04},
which uses Valiant's completeness theory for VP and VNP
to construct efficient deterministic algorithms
for the basic PIT (Polynomial Identity Testing) problem.
This problem, again thanks to Valiant's completeness result, has natural symmetries
which resemble basic problems of invariant theory.
Major progress was recently achieved in this direction in~\cite{Gur04,GGOW16,IQS,IQS2,DM-poly}.
This work was followed by~\cite{DM-oc, hamada-hirai, ivanyos-qiao,GGOW:20,BGOWW,BFGOWW,BLNW:20,BFGOWW2};
we refer to~\cite{BFGOWW2} for a description of the state-of-art.

In addition to these fundamental connections to computational complexity, orbit problems appear in many other areas,
for example in physics, since the symmetries are often given by the actions of Lie groups,
see~\cite{audin2012torus, marsdenweinstein, symplectic-in-physics}.
Recently, a connection of algebraic statistics to orbit problems was discovered~\cite{AKRS:21b,amendola2021toric}:
the problem of finding maximum likelihood estimates in certain statistical models and the `flip flop' algorithm from
statistics is precisely related to the invariant theory and the scaling algorithms of~\cite{GGOW16,BGOWW}; see also~\cite{franks2021near,derksen2021maximum,derksen2022maximum}.
In particular, the setting of log-linear models in~\cite{amendola2021toric} relates precisely to the setting of this work.

\subsection{Background and high-level summary of results}\label{subsec:orbit problems}
In this paper, we study actions of the commutative group~$T\coloneqq(\CC^\times)^d$ and its subgroup~$K\coloneqq(S^1)^d$
on vector spaces~$V \coloneqq \CC^n$.%
\footnote{Here, $\CC^\times$ denotes the multiplicative group of the nonzero complex numbers and $S^1\coloneqq\{ z\in\CC\mid |z|=1 \}$.}
The group~$T$ is called an \emph{algebraic torus} of rank~$d$, and the subgroup~$K$ is known as a \emph{compact torus}.
The action of either group partitions the vector space~$V$ into \emph{orbits}.
That is, for a vector $v\in V$, we consider the set of all vectors reachable from $v$ by applying group elements:
\[
  \OO_v\coloneqq \{t\cdot v\mid t\in T\}, \qquad \KK_v \coloneqq \{ k\cdot v\mid k\in K \} .
\]
The sets $\OO_v$ and $\KK_v$ are called the \textit{$T$-orbit} and \textit{$K$-orbit} of~$v$, respectively.
As~$K$ is compact, so is~$\KK_v$.
However, the $T$-orbit~$\OO_v$ is in general not closed and hence it is natural to also consider its \emph{orbit closure}~$\overline{\OO_v}$.%
\footnote{We take the closure with respect to either the Euclidean or the Zariski topology, as they coincide here~\cite[I\ \S10]{mumford:red}.}
See \cref{fig:example-actions} for two illustrative examples.
While torus actions are very well-understood from a structural perspective, they give rise to interesting and challenging computational problems:

\begin{problem}\label{prob:orbitproblems}
  Given a torus action on a vector space~$V$ and $v,w\in V$:
  \begin{enumerate}
    \item\label{it:orbitproblems:eq} \textbf{Orbit equality:} Decide if $\OO_v=\OO_w$.
    \item\label{it:orbitproblems:oci} \textbf{Orbit closure intersection:} Decide if $\overline{\OO_v}\cap\overline{\OO_w}\ne\varnothing$.
    \item\label{it:orbitproblems:occ} \textbf{Orbit closure containment:} Decide if $\OO_w \subseteq \overline{\OO_v}$.
    \item\label{it:orbitproblems:compact eq} \textbf{Compact orbit equality:} Decide if $\KK_v=\KK_w$.
  \end{enumerate}
\end{problem}

These problems capture and relate to a broad class of ``isomorphism'', ``classification'', or ``transformation'' problems.
Their computational complexity was determined only very recently:
\cite{torus} found that all four problems can be decided in polynomial time.
This is perhaps surprising, because for actions of noncommutative groups these problems are believed to differ
in their computational complexity and, e.g., \cref{prob:orbitproblems}~(\ref{it:orbitproblems:occ})
is known to be hard for some actions~\cite{BILPS:20}.

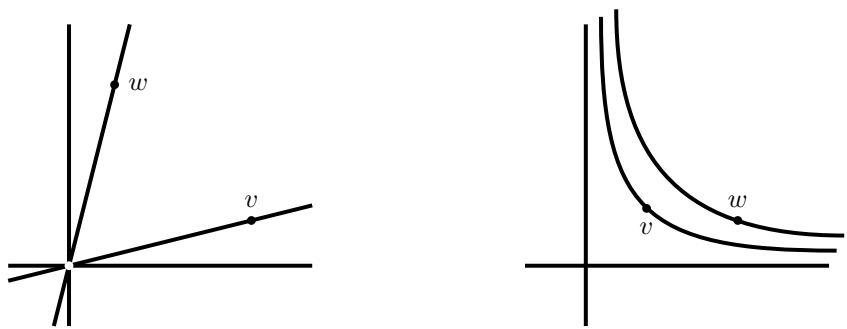
\begin{figure}[t]
    \centering
    \begin{tikzpicture}[scale=0.4]
    \pgfdeclarelayer{nodelayer}
    \pgfdeclarelayer{edgelayer}
    \pgfsetlayers{main,nodelayer,edgelayer}
	\begin{pgfonlayer}{nodelayer}
		\node  (1) at (17, 0) {};
		\node  (3) at (25, 0) {};
		\node  (4) at (17, 8) {};
		\node  (5) at (17.5, 8.25) {};
		\node  (6) at (25.25, 0.5) {};
		\node  (7) at (18, 8.5) {};
		\node  (8) at (25.5, 1) {};
		\node  (10) at (0, 0) {};
		\node  (12) at (8, 0) {};
		\node  (13) at (0, 8) {};
		\node  (14) at (2, 8) {};
		\node  (16) at (8, 2) {};
		\node  (17) at (-2, -0.5) {};
		\node  (19) at (-2, 0) {};
		\node  (20) at (0, -2) {};
		\node  (21) at (-0.5, -2) {};
		\node  (23) at (15, 0) {};
		\node  (24) at (17, -2) {};
		\node  [fill = black, circle, scale=0.35, label={above:$v$}](25) at (6, 1.5) {};
		\node  [fill = black, circle, scale=0.35, label ={right:$w$}](26) at (1.5, 6) {};
		\node  [fill = black, circle, scale=0.35, label={above:$w$}](27) at (22, 1.5) {};
		\node [fill = black, circle, scale=0.35, label={below:$v$}] (28) at (19, 1.9) {};
		\node  (29) at (19, 1.75) {};

	\end{pgfonlayer}
	\begin{pgfonlayer}{edgelayer}
		\draw [bend right=45, looseness=1.50, line width = 1.5pt] (5.center) to (6.center);
		\draw [bend right=45, looseness=1.25, line width = 1.5pt] (7.center) to (8.center);
		\draw [line width = 1.5pt](17.center) to (16.center);
		\draw [line width = 1.5pt](19.center) to (12.center);
		\draw [line width = 1.5pt](13.center) to (20.center);
		\draw [line width = 1.5pt](14.center) to (21.center);
		\draw [line width = 1.5pt](4.center) to (24.center);
		\draw [line width = 1.5pt](23.center) to (3.center);
  \node [ fill = white , circle, scale = 0.35 ](30) at (0,0) {};
	\end{pgfonlayer}
\end{tikzpicture}
\caption{\small
(a)~Action of~$T = \CC^\times$ on~$V = \CC^2$ given by~$t\cdot (x,y)=(tx,ty)$: All orbits are lines through the origin,
with the origin excluded. Thus, all orbit closures intersect in the origin.
(b)~Action of~$T = \CC^\times$ on~$V = \CC^2$ given by~$t\cdot (x,y)=(tx,t^{-1}y)$:
All orbits of vectors with non-zero coordinates $x,y\neq0$ are closed (they are hyperbolas).
}
\label{fig:example-actions}
\end{figure}

In many applications,
e.g., the ones in statistics or physics mentioned above,
the vectors $v,w$ are not given exactly, but only up to a certain error.
If this is the case, the orbit equality problem has to be replaced
by a robust generalization, which meaningfully tolerates small perturbations in the input.

The goal of this paper is to define such robust generalizations of the orbit equality problem and to investigate their computational complexity.
More precisely, we aim to \emph{approximate the distance between two orbits} given by vectors~$v,w \in \CC^n$
up to an approximation factor~$\gamma>1$.
In particular, this allows deciding whether two points are approximately in the same orbit or far from being so.

At first glance, this appears to simplify the situation, since we no longer need to distinguish between an orbit and its closure.
However, surprisingly, the picture arising is far more intricate than for the problems that were dealt with in~\cite{torus}.
On the one hand, we prove \NP-hardness when $\gamma = n^{\Omega(1/\!\log\log n)}$.
On the other hand, using a combination of tools from invariant theory and algorithmic lattice theory,
we give algorithms that achieve an approximation factor~$\gamma = \exp(\poly(n))$ and run in polynomial time
if and only if a version of the well-known number-theoretic \abc-conjecture holds!
Our algorithms also return a witness~$t \in T$ such that the distance between~$t \cdot v$ and~$w$ is at most~$\gamma$ times
the distance between the orbits.
We give precise technical statements below.

We note that our results are not the first examples of an open problem in number theory
having an impact on the complexity of algorithms.
As is widely known, the \textit{generalized (or extended) Riemann hypothesis} has been used, e.g.,
for testing primality in polynomial time~\cite{miller-primality}
(an unconditional deterministic polynomial time algorithm was later given in~\cite{aks-primality}),
the containment of knottedness in~\NP~\cite{kuperberg-knot},
or the containment of Hilbert's Nullstellensatz in the polynomial hierarchy~\cite{koiran-hn}.
To our best knowledge, however, our work, together with~\cite{mono-testing-abc},
constitutes one of the first examples that the \abc-conjecture has some bearing
on the performance of numerical algorithms.

In the remainder of this introduction we first recap the structure of torus actions and
define the input model and complexity parameters used in all our computational problems, algorithms,
and results (\cref{subsec:structure and input}).
We then state the computational problems (\cref{subsec:intro defs}) and discuss our
hardness and algorithmic results (\cref{subsec:intro hardness,subsec:intro algos}).
Next, we explain how the ``separation hypotheses'' that ensure that our algorithms run
in polynomial time are intimately related to well-known variants of
the $abc$-conjecture in number theory (\cref{subsec:intro-abc}).
Finally, we sketch the proof idea of our lattice lifting result (\cref{subsec:proof ideas}),
and we conclude with open problems for future research (\cref{subsec:intro conclusion and outlook}).

\subsection{Torus actions: structure, input model, parameters}\label{subsec:structure and input}
Every rational
action of~$T = (\CC^\times)^d$ or~$K = (S^1)^d$ on a finite-dimensional complex vector space~$V=\CC^n$
can be parameterized by a \emph{weight matrix} $M\in\ZZ^{d\times n}$, as follows:
For $t=(t_1,t_2,\dots,t_d)\in T$ and $v=(v_1,v_2,\dots,v_n)\in V$,
\begin{equation}\label{eq:def-action}
  t\,\cdot\, v \; \coloneqq \; \bigg( v_1 \prod_{i=1}^d t_i^{M_{i1}},\, v_2 \prod_{i=1}^d t_i^{M_{i2}},\,
     \dots,\, v_n \prod_{i=1}^d t_i^{M_{in}} \bigg).
\end{equation}
See \cref{fig:example-actions} for two examples with weight matrix~$M=(1 \;\; 1)$ and $M = (1 \ -\!1)$, respectively.

Our computational problems and algorithms take as their input torus actions as well as vectors~$v,w\in\CC^n$.
The former are encoded in terms of the weight matrix~$M$, with each entry represented in binary.
The latter are assumed to be vectors~$v,w \in \QQ(\ii)^n$, i.e., their components are Gaussian rationals
(complex numbers of the form $q + \ii r$ with $q,r\in\QQ$)
and given by encoding the numerators and denominators of the real and imaginary parts in binary.

Throughout the paper, we denote by~$B$ the maximum bit-length of the entries of~$M$, and by~$b$
the maximum bit-length of the components of~$v$ and~$w$, respectively.
The parameters~$B$ and~$b$ will have different effects in the bounds.

\subsection{Robust orbit problems: definitions}\label{subsec:intro defs}

We will now define precisely the computational problems that we address in this paper.
Given a torus action on a vector space~$V=\CC^n$ and two vectors~$v,w\in V$,
we wish to approximate the distance of the orbit under either the action of the algebraic torus~$T$
or the compact torus~$K$.

We start with the latter.
Because~$K = (S^1)^d$ is compact, so are its orbits.
Thus, if two $K$-orbits are distinct then they must have a finite distance with respect to any norm on~$V$.
We find it natural to use the Euclidean norm~$\norm{v} := \sqrt{\sum_{j=1}^n \abs{v_j}^2}$.
Given two vectors~$v,w\in\CC^n$, we denote their Euclidean distance by~$\dist(v,w) \coloneqq \norm{v-w}$
and we extend this notation to arbitrary subsets~$A,B\subseteq\CC^n$
by setting~$\dist(A,B) \coloneqq \inf_{a\in A, b\in B} \norm{a-b}$.
Then we are interested in the following approximation problem,
where the approximation factor~$\gamma$ is a parameter that may depend on~$d$, $n$, $b$, or $B$.

\begin{problem}
\label{prob:k-approx-dist}
The \emph{compact orbit distance approximation problem} with approximation factor~$\gamma\geq1$ is defined as follows:
Given $v,w\in (S^1)^n$, compute a number~$D$ such that
\[
 \dist(\KK_v,\KK_w) \le D \le \gamma \dist(\KK_v,\KK_w).
\]
\end{problem}

We now consider the algebraic torus~$T = (\CC^\times)^d$.
Since it not compact, in general neither are its orbits and hence two distinct orbits
need not have any finite distance -- in other words,
we can have~$\OO_v \neq \OO_w$ but $\dist(\OO_v, \OO_w) = 0$.
This can be due to two phenomena:
(a)~when orbits are not closed, they can still intersect in their closure;
(b)~even if orbits are closed, they may still become arbitrarily close at infinity.
See \cref{fig:example-actions} for an illustration.
The first phenomenon is natural:
if $\overline{\OO_v} \cap \overline{\OO_w} \neq \emptyset$
then it can be natural to define the distance to be zero.
The second phenomenon however is undesirable -- clearly we would like to be able to tell apart,
e.g., the two hyperbolas in \cref{fig:example-actions}~(b)! To overcome this problem, we will instead consider a \emph{logarithmic distance}, which always assigns a positive distance between distinct orbits, even when their closures intersect.

For simplicity, we assume that $v, w$ have non-zero coordinates, that is,
they are elements of~$(\CC^\times)^n \subseteq V$.
The key idea is to ``linearize'' or ``flatten'' the group action by using
the (coordinatewise) complex logarithm map:
\begin{equation}\label{eq:Log-def}
  \Log \colon (\CC^\times)^n \rightarrow \CC^n/2\pi\ii\ZZ^n,
  \quad
  v \mapsto (\log v_1, \log v_2,\dots, \log v_n) ,
\end{equation}
which is a group isomorphism.
This is natural because it converts the multiplicative action of~$T$ into an additive one. 
Indeed, one finds from \eqref{eq:def-action} that~$\Log$ maps any $T$-orbit~$\OO_v$ to
the image of the affine subspace~$\ran(M^T) + \Log(v)$ under the quotient map~$\CC^n \to \CC^{n}/2\pi\ii\ZZ^n$.
For distinct orbits these will have a finite Euclidean distance because
the corresponding subspaces are parallel (see \cref{fig:example-action-log}).
This motivates using the following distance:
Given~$v,w \in (\CC^\times)^n$, we define their \emph{logarithmic distance} by
\begin{equation}
\label{eq:def-distLD}
 \distLD(v,w)
\coloneqq \dist\mleft(\Log(v) + 2 \pi \ii \ZZ^n, \Log(w) + 2 \pi \ii \ZZ^n\mright)
\coloneqq \min_{\alpha \in 2 \pi \ii \ZZ^n} \norm{\Log(v) - \Log(w) - \alpha},
\end{equation}
using the natural Euclidean distance on the quotient~$\CC^{n}/2\pi\ii\ZZ^n$, and we extend this
to arbitrary subsets~$A,B\subseteq (\CC^\times)^n$ as above.
In particular, we have $\distLD(\OO_v,\OO_w)=0$ if and only if $\OO_v=\OO_w$, which is exactly the property that we wanted.
Thus we arrive at the following approximation problem, where the approximation factor~$\gamma$ may depend on~$d$, $n$, $b$, or $B$ as above.%
\footnote{\label{footnote:norm conversion} One can also use the metric~$\distLD$ in definition of \cref{prob:k-approx-dist} in place of~$\dist$. However, this is essentially equivalent, since for~$v,w\in(S^1)^n$ the two metrics are related by~$\frac2\pi \distLD(v,w) \leq \dist(v,w) \leq \distLD(v,w)$, see \cref{sec:log}.
}

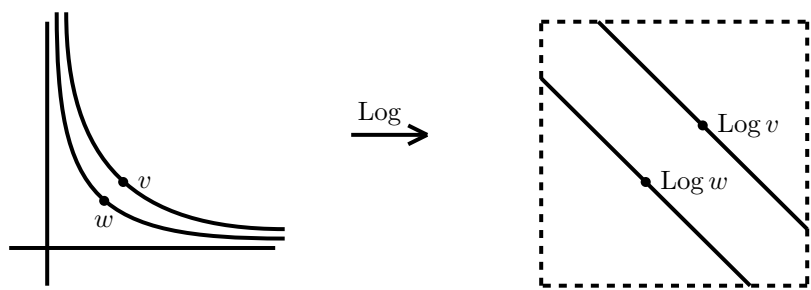
\begin{figure}[t]
\centering
\begin{tikzpicture}[scale=0.5]
    \pgfdeclarelayer{nodelayer}
    \pgfdeclarelayer{edgelayer}
    \pgfsetlayers{main,nodelayer,edgelayer}
	\begin{pgfonlayer}{nodelayer}
		\node  (0) at (0, 0) {};
		\node  (1) at (0, -1) {};
		\node  (2) at (0, 6) {};
		\node  (3) at (-1, 0) {};
		\node  (4) at (6, 0) {};
		\node  (7) at (0.25, 6.25) {};
		\node  (8) at (0.5, 6.25) {};
		\node  (9) at (6.25, 0.5) {};
		\node  (10) at (6.25, 0.25) {};
		\node  (11) at (13, 6) {};
		\node  (12) at (13, -1) {};
		\node  (13) at (20, -1) {};
		\node  (14) at (20, 6) {};
		\node  (16) at (13, 4.5) {};
		\node  (17) at (18.5, -1) {};
		\node  (18) at (20, 0.5) {};
		\node  (19) at (14.5, 6) {};
		\node  (20) at (8, 3) {};
		\node  (21) at (10, 3) {};
		\node  (22) at (9.5, 3.25) {};
		\node  (23) at (9.5, 2.75) {};
		\node  (24) at (1.5, 1.25) {};
		\node  [fill = black, circle, scale = 0.4, label={below:$w$}](25) at (1.5, 1.25) {};
		\node  [fill = black, circle, scale = 0.4, label={right:$v$}](26) at (2, 1.75) {};
		\node  [fill = black, circle, scale = 0.4, label={right:$\Log v$}](27) at (17.25, 3.25) {};
		\node  [fill = black, circle, scale = 0.4, label={right:$\Log w$}](28) at (15.75, 1.75) {};
		\node  [label={above:$\Log$}](29) at (8.75, 2.75) {};
	\end{pgfonlayer}
	\begin{pgfonlayer}{edgelayer}
		\draw [line width=1.5pt](2.center) to (1.center);
		\draw [line width=1.5pt](3.center) to (4.center);
		\draw [bend right=45, looseness=1.50, line width=1.5pt] (7.center) to (10.center);
		\draw [bend left=45, looseness=1.25, line width=1.5pt] (9.center) to (8.center);
		\draw [dashed,line width=1.5pt](11.center) to (14.center);
		\draw [dashed,line width=1.5pt](14.center) to (13.center);
		\draw [dashed, line width=1.5pt](13.center) to (12.center);
		\draw [dashed, line width=1.5pt](12.center) to (11.center);
		\draw [line width=1.5pt](16.center) to (17.center);
		\draw [line width=1.5pt](18.center) to (19.center);
		\draw [line width=1.5pt](20.center) to (21.center);
		\draw [line width=1.5pt](21.center) to (22.center);
		\draw [line width=1.5pt](21.center) to (23.center);
	\end{pgfonlayer}
\end{tikzpicture}

\caption{\small
The logarithm maps the ``hyperbolic'' orbits (left) into parallel ``lines'' (right), which have a finite Euclidean distance.
We define the logarithmic distance of the former as the Euclidean distance of the latter.
}
\label{fig:example-action-log}
\end{figure}

\begin{problem}
\label{prob:g-approx-dist-LD}
The \emph{algebraic orbit distance approximation problem} with approximation factor~$\gamma\geq1$ is defined as follows:
Given $v,w\in (\CC^\times)^n$, compute a number~$D$ such that
\[
 \distLD(\OO_v,\OO_w) \le D \le \gamma \, \distLD(\OO_v,\OO_w).
\]
\end{problem}

\subsection{Robust orbit problems: hardness}\label{subsec:intro hardness}
The appearance of a lattice in the above suggests that there might be some hardness lurking behind this problem.
In fact, we find that \emph{both} orbit distance approximation problems are \NP-hard for a sufficiently small approximation factor.
We will prove the following in \cref{sec:hardness} by showing that there is a polynomial time reduction from the \textit{closest vector problem} (CVP) to \cref{prob:k-approx-dist,prob:g-approx-dist-LD}.

\begin{theorem}\label{thm:intro-orbit-hardness}
There is a constant~$c>0$ such that \cref{prob:k-approx-dist,prob:g-approx-dist-LD}
for~$\gamma = n^{c/\log\log n}$ are \NP-hard.
\end{theorem}

We note that a solution to \cref{prob:k-approx-dist} allows distinguishing between
the cases~$\dist(\KK_v,\KK_w)\leq\varepsilon$
and~$\dist(\KK_v,\KK_w)\geq \gamma\varepsilon$ on input~$\varepsilon$,
and similarly for \cref{prob:g-approx-dist-LD}.
The CVP problem has recently attracted significant interest due to its relevance to \emph{lattice-based cryptosystems},
which are conjectured to be secure against quantum computers
~\cite{ggh-crypto, ntru-crytpo-1, ntru-crypto-2, homomorphic-crypto}.

Interestingly, our reduction from CVP is not straightforward,
but rather uses a quantitative lattice lifting result
in the spirit of~\cite{conway-sloane} that might be of independent interest.
We sketch the main idea for the $K$-action and for~$\distLD$ instead of the Euclidean distance
(which by \cref{footnote:norm conversion} makes no difference).
Let $\theta \coloneqq \frac1{2\pi\ii} \Log(v)$ and $\phi\coloneqq \frac1{2\pi\ii} \Log(w)$.
Then,
\begin{equation}
\label{eq:distLD-to-CVP}
  \frac1{2\pi} \distLD(\KK_v, \KK_w)
 = \dist(P(\theta - \phi), P(\ZZ^n)),
\end{equation}
where $P$ denotes the orthogonal projection from~$\RR^n$ onto the orthogonal complement of the row space of~$M$.
Note that the right-hand side quantity amounts to a CVP for the lattice~$P(\ZZ^n)$.
Our lattice lifting result states that, up to scaling, \emph{every} full-dimensional lattice~$\mathcal L$
can be obtained as the orthogonal projection of a cubic lattice~$\ZZ^n$, where $n$ is not much larger than~$\dim(\mathcal L)$.
Moreover, this projection can be computed in polynomial time.
This yields the desired reduction from CVP.

\begin{theorem}[Lattice lifting]\label{thm:intro lattice lifting}
Suppose $\mathcal{L} \subset \RR^m$ is a lattice of rank~$m$ given by a generator matrix $G\in\ZZ^{m\times m}$,
i.e., $\mathcal{L} = G(\ZZ^m)$.
Then we can, in polynomial time, compute $n\geq m$, a scaling factor $s\in\ZZ_{>0}$,
and an orthonormal basis~$v_1,v_2,\dots,v_n\in\QQ^n$ such that
\[
    \mathcal{L} = s\, P(\mathcal{L}'),\qquad \text{where}\quad \mathcal{L}'\coloneqq \ZZ v_1 +\ZZ v_2+\dots+\ZZ v_n
\]
and
$P \colon \RR^n\rightarrow\RR^m$ denotes the orthogonal projection onto the first~$m$~coordinates.
Moreover, we may assume that $n = O(m\log m + m \log\langle G\rangle)$,
where $\langle G\rangle$ denotes the bit-length of~$G$.
\end{theorem}

To the best of our knowledge, \cref{thm:intro lattice lifting} was not previously observed in the literature.
We sketch its proof in \cref{subsec:proof ideas} below and give the full proof in \cref{sec:liftings}.

\subsection{Robust orbit problems: algorithms}\label{subsec:intro algos}
The picture changes markedly if we allow for larger approximation ratios~$\gamma$,
where we recall that~$B$ denotes the maximum bit-length of the entries of~$M$.
Roughly speaking, our intuitive result is the following:
\begin{quote}
  \emph{If the distance between any two distinct orbits is no smaller than $\exp(-\poly(d,n,B,b))$,
  then these distances can be approximated to~$\gamma=\exp(\poly(d,n,B))$ in polynomial time.}
\end{quote}
To state this precisely, let us define the \emph{compact} and the \emph{algebraic separation parameters} as
\begin{align}\label{eq:def-sep}
  \sep_K(d,n,B,b) \coloneqq \min_{M,v,w} \dist(\KK_v,\KK_w)
\quad\text{and}\quad
  \sep_T(d,n,B,b) \coloneqq \min_{M,v,w} \distLD(\OO_v,\OO_w),
\end{align}
respectively, where the minima are taken over all weight matrices $M\in\ZZ^{d\times n}$
with entries of maximum bit-length at most~$B$ and over all vectors $v,w\in (\QQ(\ii)^\times)^n$
with components of bit-length at most~$b$ such that~$v$ and~$w$ are in distinct orbits.

\begin{sephypothesis}[for the compact torus]\label{hyp:compact-separation}
$\sep_{K}(d,n,B,b)\geq\exp(-\poly(d,n,b,B))$.
\end{sephypothesis}

\begin{sephypothesis}[for the algebraic torus]\label{hyp:sep}
$\sep_T(d,n,B,b)\geq \exp(-\poly(d,n,b,B))$. 
\end{sephypothesis}

The latter hypothesis is readily seen to imply the former
(this follows essentially from \cref{footnote:norm conversion},
see \cref{cor:rel-btw-seps}).
While both hypotheses appear geometric in nature, it turns out that they are intimately related
to well-known quantitative versions
of the \abc-conjecture in number theory.
We will explain this connection in \cref{subsec:intro-abc}.
Asssuming these hypotheses, we can solve the orbit distance approximation problems
with an exponential approximation factor:

\begin{theorem}\label{thm:intro-approximation-algorithm}
If \cref{hyp:compact-separation} holds, then there is a polynomial-time algorithm
that solves \cref{prob:k-approx-dist}
for an approximation factor~$\gamma = \exp(\poly(d,n,B))$.
Similarly, if \cref{hyp:sep} holds, then there is a polynomial-time algorithm
that solves \cref{prob:g-approx-dist-LD}
for~$\gamma = \exp(\poly(d,n,B))$.
\end{theorem}

In fact, we show in \cref{sec:main-approx} without any hypotheses that for~$\gamma = \exp(\poly(d,n,B))$,
there is an algorithm solving \cref{prob:g-approx-dist-LD} in time $O(\poly(d,b,n,B) \log(\sep^{-1}_T(d,n,B,b)))$,
and similarly for \cref{prob:k-approx-dist}.
This algorithm runs in polynomial time if and only if the separation hypothesis is true, see \cref{rem:poly-time-iff-sep-hyp}.
Our algorithm also solves the corresponding \emph{search problem} (see \cref{thm:rop-witness}):
When $\OO_v\ne\OO_w$, it finds a group element~$t \in T$ such that%
\begin{equation}\label{eq:witness}
  \distLD( t \cdot v, w ) \leq \gamma \, \distLD(\OO_v,\OO_w).
\end{equation}
In particular, $t$ can serve as a witness of the approximate distance of the orbits.
In general, the coordinates of any~$t$ satisfying \cref{eq:witness} may require a superpolynomial bit-length.
Accordingly, our algorithm will output instead a vector~$x\in \QQ(\ii)^d$ such that~$t = \Exp(x)$.
When $\OO_v = \OO_w$, then it will in general not be possible to output such a group element,
since~$t \cdot v = w$
will in general not have a rational solution (neither in~$t$ nor in~$x$).

To explain the main idea behind proving \cref{thm:intro-approximation-algorithm}
and motivate the separation hypothesis,
we recall the algorithm for solving the orbit equality problem in~\cite{torus}.
For simplicity, we consider vectors~$v,w\in(\CC^\times)^n$ with closed $T$-orbits.
The algorithm in~\cite{torus} first computes a basis~$\mathcal B$ of the lattice
$\mathcal{K} \coloneqq \{ \alpha\in\ZZ^n \mid M\alpha=0 \}$ in polynomial time.
Each basis vector~$\alpha \in \mathcal B$ determines an invariant function (Laurent monomials)
of the form
\begin{align*}
  f_\alpha(x) 
\coloneqq x_1^{\alpha_1} x_2^{\alpha_2}\dots x_n^{\alpha_n}.
\end{align*}
Then, $\OO_v=\OO_w$ if and only if~$v$ and~$w$ take the same value
on all these functions~\cite[Corollary~5.2]{torus}.
Even though the~$\alpha$ have polynomial bit-length,
testing whether~$f_\alpha(v) = f_\alpha(w)$ is not obvious,
as the bit-length of the evaluations can be exponentially large in the input size.
Nevertheless, it is possible to determine equality of these numbers in polynomial time
without actually computing them, see~\cite[Proposition~5.5]{torus}.

When dealing with the problem of approximating orbit distances, one may use similar ideas.
One finds that one now needs to test whether~$f_\alpha(v)$
and $f_\alpha(w)$ are \emph{close} (rather than equal),
for a suitable choice of a metric.
This turns out to be substantially more difficult than the equality testing,%
\footnote{Just testing an inequality~$f_\alpha(v) \ge f_\alpha(w)$ in polynomial time is only known assuming a certain strengthening of the \abc-conjecture; see~\cite{mono-testing-abc}.}
and indeed it must be by the hardness results in \cref{subsec:intro hardness}.
To solve it with an exponential approximation factor we can proceed as follows:
Let $H\in\ZZ^{k\times n}$ denote the matrix with rows the vectors in~$\mathcal B$.
We prove in \cref{sec:log} that the logarithmic distance between the orbits~$\OO_v$ and~$\OO_w$
can be approximated by~$\norm{H \Log(v) - H \Log(w)}$,
up to a factor~$\sigma_{\max}(H)/\sigma_{\min}(H) = \exp(\poly(d,n,B))$.
Here, $\sigma_{\max}(H),\sigma_{\min}(H)$ denote the largest and smallest singular values of $H$, respectively.
It remains to approximate $\norm{H \Log(v) - H \Log(w)}$ with relative error.
The key challenge is in computing rational approximations of the logarithms.
This can be done by standard algorithms, but we have to make sure
that polynomially many bits of accuracy are sufficient.
This is guaranteed by the separation hypothesis!
The idea of relating the distance between two orbits with the difference
between the values that the invariants take on these orbits,
and approximating the distance between the orbits in this way is inspired by the invariant theory
and it will be important later
when we explain the close relationship between the robust orbit problems and the \abc-conjecture.
On the other hand,
we will also develop an alternative algorithm to the robust orbit problems in \cref{sec:sldp-to-cvp}
by providing a polynomial time reduction to CVP under the assumption that \cref{hyp:sep} is true
by using \cref{eq:distLD-to-CVP}.


Finally, we study the complexity of the orbit distance approximation problems when~$n$ is fixed.
In this situation, \cref{hyp:compact-separation,hyp:sep} are true,
as we will discuss in \cref{subsec:intro-abc},
and hence there are polynomial-time algorithms for \cref{prob:k-approx-dist,prob:g-approx-dist-LD}
with $\gamma=\exp(\poly(d,B))$.
However, we can in fact choose~$\gamma$ to be much smaller:

\begin{theorem}\label{thm:compact-fixed-dimension}
For fixed~$n$, \cref{prob:k-approx-dist} can be solved in polynomial time
for the approximation factor~$\gamma = 2$.
\end{theorem}

\begin{theorem}\label{thm:fixed-dimension}
For fixed~$n$,
and any fixed~$\gamma>1$,
\cref{prob:g-approx-dist-LD} can be solved in polynomial time.
More specifically, there is a polynomial-time algorithm that on
input~$M\in\ZZ^{d\times n}$, $v,w \in \QQ(\ii)^n$, and $\eps>0$ outputs a number~$D$
such that~$\distLD(\OO_v,\OO_w) \,\leq \, D \, \leq\, (1+\varepsilon)\, \distLD(\OO_v,\OO_w)$.
\end{theorem}

We prove these results at the end of \cref{sec:sldp-to-cvp}.

\subsection{Separation hypotheses and the \abc-conjecture}\label{subsec:intro-abc}
The famous \abc-conjecture by Oesterl\'{e} and Masser~\cite{oesterle} states that for every~$\varepsilon>0$
there exists~$\kappa_\varepsilon>0$ such that, if~$a$ and~$b$ are coprime positive integers and $c=a+b$, then
\[
c < \kappa_\varepsilon \, \rad(abc)^{1+\varepsilon},
\] where $\rad(abc)$ is the \textit{radical} of $abc$, i.e., the product of all primes dividing~$abc$.
Motivated by an earlier conjecture of Szpiro, the \abc-conjecture is one of the most important open problems
in number theory,
due to its many consequences~\cite{granvilletucker,abc-consequences}.

Baker~\cite{baker98} observed that the \abc-conjecture is intimately related to lower bounds
for linear forms in logarithms.
Let $v_1,v_2,\dots,v_n$ be positive integers and $e_1,e_2,\dots,e_n$ be integers such that
the product $v_1^{e_1} v_2^{e_2}\cdots v_n^{e_n}\neq 1$ is different from one.
How close to $1$ can such a product be?
Equivalently, how close can the linear combination of logarithms,
\begin{equation}\label{eq:def-Lambda(v,e)}
  \Lambda(v,e) \coloneqq e_1\log v_1 +\ldots + e_n \log v_n ,
\end{equation}
be to zero?
We note that if
$|v_1^{e_1}\dots v_n^{e_n}-1|\leq \frac{1}{2}$,
then
\[
  \frac{1}{2} |\Lambda(v,e)| \, \leq \, |v_1^{e_1} \cdots v_n^{e_n}-1| \, \leq \, 2 |\Lambda(v,e)| ,
\]
so lower bounds for $|\Lambda(v,e)|$ and for $|v_1^{e_1}\dots v_n^{e_n}-1|$ are equivalent;
see \cref{eq:Lambda-inequ}.
Improving work by Baker and W\"ustholz~\cite{wustholz1993}, Matveev~\cite{matveev} proved that
if $\Lambda(v,e)\ne 0$, then
\begin{equation}\label{eq:BW:93}
  \abs{ \Lambda(v,e) } \ge e^{-B \cdot b^{O(n)}}.
\end{equation}
Here and below $B$ and $b$ denote upper bounds on the bit-length of the $e_j$ and $v_j$, respectively.
Note that for fixed~$n$, this bound is polynomial in~$B$ and $b$.
Furthermore, every lower bound for~$|\Lambda(v,e)|$ implies a version of the \abc-conjecture:
Stewart, Yu and Tijdemann~\cite{yu-stewart, stewart_oesterle-masser_1986}
used the Baker-W\"{u}stholz-Matveev bound~\eqref{eq:BW:93} and its $p$-adic version by Van der Poorten~\cite{vdpoorten}
to prove the inequality
\[
  \log c < \kappa_\varepsilon\, \rad(abc)^{\frac{1}{3}+\varepsilon}.
\]
To our knowledge, this bound is the strongest bound to date given towards a solution to
the \abc-conjecture and demonstrates the strong connection between the \abc-conjecture and
lower bounds for linear forms in logarithms.
Baker~\cite{baker98} proved that a stronger lower bound than \eqref{eq:BW:93}
(together with its $p$-adic version)
is equivalent to a certain strengthening of the \abc-conjecture.
We will say more about this connection in \cref{sec:abc}.

It is widely conjectured that \eqref{eq:BW:93} is \emph{not} optimal.
In particular, famous conjectures in number theory,
such as Waldschmidt's conjectures~\cite[Conjecture~14.25, p.~547]{waldschmidt-2},~\cite[Conjecture~4.14]{waldschmidt-1},
or the Lang-Waldschmidt conjecture
for Gaussian rationals~\cite[Introduction to Chapters~X and~XI]{lang-elliptic-curves}
imply that the following hypothesis holds (in fact, they make even stronger predictions!):

\begin{numhypothesis}\label{hyp:linearforms}
For any~$n\in\NN$, $v_1,v_2,\dots,v_n\in\QQ(\ii)^\times$, and $e_1,e_2,\dots,e_n\in\ZZ$ such that,
with $\Log$ the principal branch of the complex logarithm,
\[
  \Lambda(v,e) \coloneqq e_1 \Log v_1 + e_2 \Log v_2 +\dots +e_n\Log v_n \neq 0,
\]
we have
\[
	|\Lambda(v,e)| \geq \exp(-\poly(n,b,B)),
\]
where $B$ and $b$ are the maximum bit-lengths of the $e_j$ and~$v_j$, respectively.
\end{numhypothesis}

Here we prove a novel connection between this number-theoretic conjecture and
the separation hypothesis of group orbits, which can be seen as further evidence for the latter:

\begin{theorem}\label{thm:abc-and-sep}
\Cref{hyp:sep} and \cref{hyp:linearforms} are equivalent.
Moreover, \cref{hyp:compact-separation} is equivalent to \cref{hyp:linearforms}
when the latter is restricted to~$\abs{v_1}=\dots=\abs{v_n}=1$.
\end{theorem}

We prove this in \cref{sec:abc} by using similar ideas as sketched in \cref{subsec:intro algos}.
There, we also show that Matveev's bound~\eqref{eq:BW:93} implies
the lower bound~$\sep_T(d,n,B,b) \geq e^{-B^{O(1)} \cdot b^{O(n)}}$,
which is only exponentially small for constant~$n$.

\subsection{Proof sketch of the lattice lifting theorem}\label{subsec:proof ideas}
We now give a summary of the idea behind the proof of \cref{thm:intro lattice lifting}.
We call a collection $u_1,u_2,\dots,u_n$ of vectors in $\RR^m$, $n\ge m$,
an \emph{eutactic star of scale~$s$}
if each $u_j$ can be extended to a vector~$v_j\in\RR^n$ such that the $v_j$ are pairwise orthogonal
and of the same norm~$s$.
Note that $u_i=P(v_i)$ where $P$ denotes the projection onto the first~$m$ coordinates.
Let $X\in\RR^{m\times n}$ denote the matrix whose $i$-th column is $u_i$.
It is easy to see that the collection $u_1,u_2,\dots,u_n$ forms an eutactic star of scale $s$
iff $X\, X^T = s^2 I_m$.

Assume that the lattice $\mathcal{L}\subset\RR^m$
is given as the $\ZZ$-span of the columns of the matrix $G\in\ZZ^{m\times m}$.
Our aim is to understand when $\mathcal{L}$ is generated (as a lattice) by an eutactic star.
We first note that the collection $u_1,u_2,\dots,u_n$ is contained in $\mathcal{L}$
iff there exists an integer matrix $L\in\ZZ^{m\times n}$ such that $X=GL$.
To enforce further that the vectors $u_i$ generate $\mathcal{L}$ as a lattice,
we require $L$ to be \textit{right invertible},
meaning there exists $R\in\ZZ^{n\times m}$ such that $LR=I_m$.
In \cref{prop:eutactic} we prove:
\begin{equation*}
\begin{split}
 \mathcal{L}\text{ is generated by}&\text{ an eutactic star of scale } s\in \ZZ_{>0}\, \iff\, \\
 &\exists \text{ right invertible }L\in\ZZ^{m\times n} \text{ s.t. }
    (GL)(GL)^T = s^2 I_m.
\end{split}
\end{equation*}
Therefore, 
writing $\mathcal{L}$ as in \cref{thm:intro lattice lifting} amounts to finding such~$L$ and~$s$.
In \cref{sec:eutactic} we show that for given $G\in\ZZ^{m\times m}$,
such $L$ and $s$ can be computed in polynomial time.
For this, we choose the integer~$s$ so large that the matrix
\[
A\coloneqq s^2 (G^{-1}G^{-T})-I_m.
\]
is positive definite.
In \cref{thm:effectivemordell} we prove that a decomposition of the form $A=(L')(L')^T$
can be computed in polynomial time (efficient Waring decomposition).
Note that this amounts to writing the given positive definite quadratic form $x^T Ax$
as sum of squares of rational linear forms.
Such decompositions are known to exist if $n=m+3$~\cite{mordell},
and randomized polynomial time algorithms
computing them are available~\cite{rabin-shallit}.
However, to the best of our knowledge, it is open
whether such small decompositions can be computed in deterministic polynomial time.
Instead, we will use the simple \cref{lem:sos}, which writes an integer $D$ as a sum of $O(\log\log D)$ many squares,
but can be easily shown to run in deterministic polynomial time.
Finally, one checks that $L\coloneqq \begin{bmatrix} I_m & L'\end{bmatrix}$
satisfies $(GL)(GL)^T=s^2 I_m$, which is what we wanted to get.

\subsection{The method of Kempf-Ness}
A major challenge is to extend the results in this paper beyond torus actions!
The Kempf-Ness theorem~\cite{kempfness} indicates a strategy to do so.
This is a general theorem on rational actions of reductive algebraic groups,
that in principle allows to reduce orbit closure intersection problems
to compact orbit equality problems.
It states that the vectors of minimal norm in the orbit closure of $v\in V$
form a single $K$-orbit that we denote by $\KK_{v^\star}$.
Here $K$ denotes a maximal compact subgroup of a reductive algebraic group $G$
that is assumed to act isometrically on $V=\CC^n$.
Moreover, the Kempf-Ness theorem states that for $v,w\in\CC^n$,
\[
 \overline{\OO_v}\cap\overline{\OO_w}\neq\varnothing \quad \iff\quad \KK_{v^\star}=\KK_{w^\star}.
\]
In an ideal world, this reduces the problem of testing
$\overline{\OO_v}\cap\overline{\OO_w}\neq\varnothing$
to the problem of testing $\KK_{v^\star}=\KK_{w^\star}$,
provided that we can compute $v^\star$ on input $v$.
Unfortunately, even when the coordinates of $v$ are rational,
$v^\star$ may not have rational entries.
Instead one may work with a numerical approximation of $v^\star$.
This was successfully carried out in~\cite{AZGLOW}
for the left-right action on tuples of complex matrices.
As a result, polynomial time algorithms for the orbit closure problems for the left-right action were obtained along this way.

For torus actions, approximating $v^\star$ becomes a convex optimization problem. More precisely, we consider the \emph{Kempf-Ness function} corresponding to $v$, \[
f(x) \coloneqq \log\; \Vert e^{x/2}\cdot v\Vert^2 \; = \log\Big(\;\sum_{i=1}^n q_i\, e^{\omega_i^T x}\;\Big),
\] where $\omega_1,\omega_2,\dots,\omega_n\in\ZZ^d$ denote the columns of $M$ and $q_i\coloneqq |v_i|^2$. We have $\min_x f(x)=2\log\Vert v^\star\Vert$ by the definition of $v^\star$. This function finds applications in a surprising number of different areas: In machine learning~\cite{ml-lse-1, ml-lse-2, ml-lse-3, ml-tropical},
$f(x)$ is known as the \textit{log-sum-exp} function and it is used as a smooth approximation
to the piecewise affine function $L(x)\coloneqq \max_i \big(\langle x,\omega_i\rangle +\log q_i\big)$: we have
$
 L(x) \,\leq\, f(x) \,\leq\, L(x) + \log n .
$
The optimization of the Kempf-Ness function also has uses in the area of statistical physics: The convex program $\inf_x f(x)$ is the Lagrange dual of an entropy maximization problem with mean constraints. More precisely, one has \[
\inf_x f(x) = \sup\, \Big\{\,  -D_{KL}(p||q)\; \Big| \; \ \; \sum_{i=1}^n p_i \omega_i =0,\; \sum_{i=1}^n p_i=1,\; \forall i, p_i\geq 0\, \Big\}
\] where $D_{KL}(p || q)\coloneqq \sum_{i=1}^n p_i \log(p_i/q_i)$ denotes the \emph{Kullback-Leibler divergence} between the probability distribution $p$ and (possibly non-normalized) distribution $q$. For more on this connection, we refer to~\cite{vishnoi-singh,vishnoi-straszak,leake-vishnoi,BLNW:20}.

In \cref{sec:kempf-ness}, we consider the problem of approximating the optimal solution of $f(x)$. We conjecture that a point $x\in\RR^d$ that is $\varepsilon$-close to $\arg\min f(x)$ in the Euclidean distance can be computed in polynomial time, see \cref{conj:optimal}. We show that if the conjecture is true then the logarithmic distance $\distLD(\KK_{v^\star},\KK_{w^\star})$ between two Kempf-Ness orbits can be efficiently approximated.
This leads to new polynomial time algorithms for the orbit equality problems,
provided the separation hypothesis holds.
We interpret our finding as some positive evidence towards the feasibility
of the Kempf-Ness approach in more general settings, but also interpret it
as a warning about the difficulties to be encountered.

\subsection{Conclusion and outlook}\label{subsec:intro conclusion and outlook}
Summarizing, we defined the problem of approximating orbit distances for torus actions and
showed that these problems are \NP-hard for an almost-polynomial factor,
but can be solved in polynomial time within an at most single exponentially large factor
under the assumption that a variant of the famous \abc-conjecture is true.
Our results point to an exciting connection between orbit problems and deep theorems
and conjectures in the areas of number theory and algorithmic lattice theory.

Let us point out that
our approximation algorithm 
for the robust orbit problem,
jointly with our reduction result from CVP to this problem,
yield a new polynomial time approximation algorithm for CVP,
which is not based on the LLL algorithm. 
Unfortunately,
the resulting approximation factor $\gamma$ 
not only depends on the rank~$m$ of the lattice,
but also on the bit-length of the generators of the lattice,
which cannot compete with the well-known $2^{O(m)}$-approximation algorithms for CVP
(see, for example,~\cite[Lemma~2.12]{micciancio}).
It is an interesting question whether competitive algorithms for CVP can be developed along this way and whether there are other uses of our lattice lifting result
(\cref{thm:intro lattice lifting}).
A further interesting problem is 
if one can remove the separation hypothesis
by allowing for significantly larger approximation factors~$\gamma$.

\subsection{Organization of the paper}
The preliminary \cref{sec:prelim} collects known facts from
different areas that we need to establish our main results.
\Cref{sec:log} is devoted to an analysis of the structure of the orbits
in the logarithmic space $\CC^n/2\pi\ii\ZZ^n$, and the
logarithmic distance metric $\distLD$ defined in \cref{eq:def-distLD}.
We show that $\distLD(\OO_v,\OO_w)$ is approximated by a linear form in $\Log v,\Log w$.
The goal of \cref{sec:abc} is to explain the connection
between the \abc-conjecture and \cref{hyp:linearforms}
in more detail and to formally prove \cref{thm:abc-and-sep}.
\Cref{sec:main-approx} describes our algorithms to solve \cref{prob:g-approx-dist-LD}
and \cref{prob:k-approx-dist} in polynomial time,
proving \cref{thm:intro-approximation-algorithm}.
The goal of \cref{sec:hardness} is to prove the hardness results by reducing
the closest vector problem 
to the orbit distance problems (\cref{prob:k-approx-dist} and \cref{prob:g-approx-dist-LD}).
There, we also provide the proof of \cref{thm:intro lattice lifting} on efficient lattice lifting
and we prove
\cref{thm:compact-fixed-dimension} and \cref{thm:fixed-dimension}.
Finally, \cref{sec:kempf-ness} is devoted to the analysis of
the Kempf-Ness approach for general torus actions.
We conjecture a result on the complexity of approximating the optimal solution of the Kempf-Ness function and based on that, we devise a numerical algorithm for deciding the equality of two orbits,
which runs in polynomial time if \cref{hyp:sep} is true.

\section{Preliminaries}\label{sec:prelim}

The goal of this section is to collect known facts from
different areas: invariants for torus actions, singular values of matrices,
the complexity of numerically computing elementary functions.
We end with general observations on quotient metrics.

\subsection{Notation}
Throughout the paper, $T=(\CC^\times)^d$ denotes the algebraic torus of rank $d$
and $V=\CC^n$. 
We always denote by $M\in\ZZ^{d\times n}$ a weight matrix that determines the action of $T$ on $V$.
The Euclidean norm of $v\in\CC^n$ is denoted~$\norm{v}$.
We write $\dist(v,w) \coloneqq \norm{v-w}$
for the Euclidean distance of $v,w\in\CC^n$
and we extend this notation to nonempty subsets~$A,B\subseteq\CC^n$
by setting~$\dist(A,B) \coloneqq \inf_{a\in A, b\in B} \norm{a-b}$.

We denote by $\Re(z)$ and $\Im(z)$ the real and imaginary part of a complex number $z\in\CC$, respectively.
Moreover,
$\log(z)\in\CC/2\pi\ii\ZZ$ denotes complex logarithm.
The componentwise defined logarithm and exponentiation functions are denoted by $\Log\colon (\CC^\times)^n\rightarrow\CC^n/2\pi\ii\ZZ^n$ and $\Exp:\CC^n/2\pi\ii\ZZ^n\rightarrow (\CC^\times)^n$.

The letters $v,w$ denote vectors in $V$,
while the Greek letters $\eta,\zeta$
refer to vectors that are exponentiated, as in $v=\Exp(\eta)$.
We always denote by $b$ a bound for the bit-lengths of the components
of the input vectors $v,w\in\QQ(\ii)^n$ (or $\eta,\zeta$),
and $B$ always denotes a bound for the bit-length of the entries of the weight matrix $M\in\ZZ^{d\times n}$.

\subsection{Invariant theory of torus actions}\label{se:invar-th-torus}
In this section, we present a brief summary of the setting and the results of~\cite{torus}.
We will discuss the proof strategy sketched in \cref{subsec:intro algos} in more detail.
As always, $T=(\CC^\times)^d$ is the $d$-dimensional torus and $V=\CC^n$.

Any rational action of $T$ on $V$, up to some base change,
can be simultaneously diagonalized and brought to the form \cref{eq:def-action}.
The matrix $M=[M_{ij}]\in\ZZ^{d\times n}$ occurring in \cref{eq:def-action}
is called the \textit{weight matrix} and its columns
$\omega_1,\omega_2,\dots,\omega_n\in\ZZ^d$ are called
the \textit{weights} of the action.
There is an induced action of $T$ on the polynomial ring $\CC[x_1,x_2,\dots,x_n]$, given by
\[
 (t\cdot f)(v) \coloneqq f(t^{-1}\cdot v), \qquad f\in\CC[x_1,x_2,\dots,x_n],\; t\in T,\; v\in V.
\]
A \textit{polynomial invariant} is a fixed point of this action, i.e.,
$t\cdot f=f$ for all $t\in T$.
Note that a polynomial is an invariant iff it is constant along orbits.
We denote the ring of polynomial invariants by $\CC[V]^T$.
For a monomial $x^\alpha\coloneqq x^{\alpha_1}x^{\alpha_2}\dots x^{\alpha_n}\in\CC[x_1,x_2,\dots,x_n]$,
we have \begin{equation}
\label{eq:monomial}
x^\alpha(t\cdot v) = t^{M\alpha} x^\alpha(v),
\end{equation}
\Cref{eq:monomial} implies that the line $\CC x^\alpha$ is preserved by the action of $T$.
Moreover, $x^\alpha$ is an invariant iff $M\alpha=0$.
Hence a polynomial $f$ is an invariant if and only if each monomial appearing in $f$ is invariant.
In particular, the space of invariant polynomials is linearly spanned by the invariant monomials.

The action of $T$ leaves $X := (\CC^\times)^n$ invariant.
We have an induced action of $T$ on the algebra
\[
 \CC[X] = \CC[x_1,x_1^{-1},x_2,x_2^{-1},\dots,x_n,x_n^{-1}],
\]
of Laurent polynomials,
the ring of regular functions $\CC[X]$ of $X$,
which is easier to study.
The above observations extend:
\cref{eq:monomial} also holds for \textit{Laurent monomials}
$x^{\alpha}$ with exponent vector $\alpha\in\ZZ^n$,
which is thus invariant iff $M\alpha=0$.
The space of invariant Laurent polynomials is linearly spanned by the invariant Laurent monomials.
We call
\begin{equation}
\label{eq:def-L}
    \mathcal{K} \coloneqq \{\alpha\in\ZZ^n\mid M\alpha=0\} ,
\end{equation}
the \emph{lattice of rational invariants} defined by $M$.
Note that the rank of $\mathcal{K}$ is given by $k= n - \rk M$.

Thus $\alpha_1,\alpha_2,\dots,\alpha_k$ is a lattice basis of $\mathcal{K}$,
then the Laurent monomials $x^{\alpha_1},x^{\alpha_2},\dots,x^{\alpha_k}$ generate
$\CC[X]^T$ as an algebra. Then, for $v,w\in(\CC^\times)^n$, we have
(see~\cite[Proposition 4.1]{torus})
\begin{equation}\label{prop:t-invariants}
\OO_v= \OO_w \quad \iff \quad \forall i\in [k] \quad x^{\alpha_i}(v)=x^{\alpha_i}(w) .
\end{equation}
There is an analogous result for the orbits of the compact torus $K=(S^1)^d$
(see~\cite[Proposition~8.1]{torus}):
For two vectors $v,w\in(\CC^\times)^n$, we have
\begin{equation}\label{prop:k-invariants}
\begin{split}
	\KK_v=\KK_w\quad &\iff\quad \forall \alpha\in\mathcal{K}, \;
     x^{\alpha}(v)=x^{\alpha}(w)\,\text{ and }\, \forall i\in [n], |v_i|=|w_i|\\
	         &\iff\quad \OO_v = \OO_w \;\text{ and }\;\forall i\in [n], \, |v_i|=|w_i|.
\end{split}
\end{equation}
The next theorem (see~\cite[Corollary~4.4 and Proposition~5.5]{torus})
shows that one can decide $\OO_v=\OO_w$ in polynomial time.

\begin{theorem}
\label{thm:basis-and-equality}
	\begin{alphaenumerate}
		\item\label{it:basis 1}
             We can compute in $\poly(d,n,B)$-time a basis for the lattice $\mathcal{K}$ of $T$-invariant
             Laurent monomials. In particular, basis elements of $\mathcal{K}$ have bit-lengths bounded by $\poly(d,n,B)$.
		\item\label{it:basis 2}
            Suppose $v,w\in(\QQ(\ii)^\times)^n$ and $\alpha\in\ZZ^n$. Then in $\poly(n,b,\langle \alpha\rangle)$-time
            one can decide whether $x^{\alpha}(v)=x^{\alpha}(w)$.
	\end{alphaenumerate}
\end{theorem}

Let $H\in\ZZ^{k\times n}$ be a matrix whose rows $\alpha_1^T,\alpha_2^T,\dots,\alpha_k^T$ form a basis of the lattice $\mathcal{K}$ defined by~$M$, see \cref{eq:def-L}.
We call $H$ a \textit{matrix of rational invariants}.

\begin{proposition}\label{prop:allaboutH}
We have $\ker H = \ran M^T$ and $\rk H=k$.
Moreover, $H(\ZZ^n)=\ZZ^k$, i.e., for every $\beta\in\ZZ^k$,
there exists $\alpha\in\ZZ^n$ such that $H\alpha=\beta$.
If the rows of $H$ are produced from $M$ by
a polynomial time algorithm (as in \cref{thm:basis-and-equality}),
then the bit-length of $H$ is at most $\poly(d,n,B)$.
\end{proposition}

\begin{proof}
The first claim is obvious from the construction of $\mathcal{K}$ and $H$.
The second claim is obvious. 
For the third claim, we are going to use the integral analogue of Farkas' lemma
(see~\cite[Corollary~4.1a]{integer-farkas}), which states that for $H\in\ZZ^{k\times n}, \beta\in\ZZ^k$,
\[
	H \alpha = \beta \text{ has a solution }\alpha\in\ZZ^n \quad \iff\quad \forall \gamma\in\QQ^k
     \text{ with }H^T \gamma\in\ZZ^n, \text{ it holds that } \beta^T \gamma\in\ZZ.
\]
To reach a contradiction, suppose $H(\ZZ^n)\neq\ZZ^k$, i.e., there exists a vector $\beta\in\ZZ^k\setminus H(\ZZ^n)$.
Then there exists a rational vector $\gamma\in\QQ^k$ such that $H^T \gamma\in\ZZ^n$ but $\beta^T \gamma\not\in\ZZ$.

We note that if $H^T\gamma\in\ZZ^n$ for some $\gamma\in\QQ^k$, then $H^T\gamma\in\mathcal{K}$ since $MH^T=0$.
Moreover, since the rows of $H$ generate $\mathcal{K}$,
we further have $H^T \gamma = H^T \tilde{\gamma}$ for some integral vector~$\tilde{\gamma}\in\ZZ^k$.
However, $\rk H=k$ so we must have $\gamma=\tilde{\gamma}$, so $\gamma$ is integral.
This contradicts~$\beta^T \gamma\not\in\ZZ$.
\end{proof}

\begin{remark}
The property that $H(\ZZ^n)=\ZZ^k$ will be used frequently throughout the paper.
We note that this is equivalent to the diagonal entries of the Smith normal form of $H$ being all one. Moreover, the proof of \cref{prop:allaboutH} shows that $H(\ZZ^n)=\ZZ^k$ holds whenever
the rows of $H$ form a basis of a \emph{saturated lattice} in $\ZZ^n$ with rank $k$. We recall that this means for $s\in\ZZ_{>0}$ and $\alpha\in\ZZ^n$, $s\alpha\in\mathcal{K}$ implies $\alpha\in\mathcal{K}$.
\end{remark}

\begin{remark}
We could ignore the dependence of the complexity parameters on $d$
since $d\leq n$  can be assumed  without loss of generality.
For seeing this, assume $d>n$ and let
$A\in\ZZ^{d\times n}$ denote the (row reduced) Hermite normal form of
the weight matrix $M\in\ZZ^{d\times n}$.
Then the orbits with respect to the actions defined by $M$ and by $A$ are the same.
However, since the last $d-n$ rows of $A$ are zero,
the last $d-n$ coordinates of the torus act trivially on $V$
and can be ignored.
\end{remark}

\subsection{Singular values}
Suppose $k\leq n$. For $H\in\CC^{k\times n}$, there exist unitary matrices
$X\in\mathrm{U}(k), Y\in \mathrm{U}(n)$
such that $XHY$ is a diagonal matrix with non-negative real diagonal entries
$\sigma_1\geq\sigma_2\geq\dots\geq\sigma_k\geq 0$.
The values $\sigma_i$ are called the \textit{singular values} of the matrix~$H$.
We denote by $\sigma_{\max}(H)\coloneqq\sigma_1$ and $\sigma_{\min}(H)\coloneqq\sigma_k$.
We have $\sigma_{\min}(H)\neq 0$ if and only if $\rk H=k$.
Recall that $\|x\|$ stands for the Euclidean norm of $x\in\CC^n$
and $\dist$ denotes the corresponding distance for subsets of~$\CC^n$.

The maximum and the minimum singular values $\sigma_{\max}(H), \sigma_{\min}(H)$
are characterized by the following properties:
\begin{equation}\label{eq:sings}
	\sigma_{\max}(H) = \max_{\Vert x\Vert =1} \Vert Hx\Vert
	\quad\text{and}\quad
       \forall x \;\; \Vert Hx\Vert\geq\sigma_{\min}(H)\cdot \dist(x, \ker H).
\end{equation}

\begin{lemma}
\label{lem:sing-value-inequality}
Suppose $k\leq n$ and $H\in\CC^{k\times n}$.
For nonempty subsets $A,B\subset\CC^n$ we have
\[
 \sigma_{\min}(H)\dist(A+\ker H,B+\ker H)\,\leq\,\dist(HA, HB)\,\leq\,\sigma_{\max}(H) \dist(A+\ker H,B+\ker H).
\]
\end{lemma}

\begin{proof}
We first note that $\dist(A+\ker H, B+\ker H)$ is the infimum of $\Vert a-b+u\Vert$ over $a\in A, b\in B, u \in \ker H$,
and
$\dist(HA,HB)$ is the infimum of $\Vert H(a-b)\Vert$ over $a\in A, b\in B$.

For any $a\in A, b\in B$ and $u\in \ker H$ we have, by the equation in \eqref{eq:sings},
\[
 \dist(HA, HB) \leq \Vert H(a-b)\Vert = \Vert H(a-b+u)\Vert \leq \sigma_{\max}(H)\, \Vert a-b+u\Vert.
\]
Taking the infimum over $a,b,u$ we have $\dist(HA,HB)\leq \sigma_{\max}(H)\dist(A+\ker H, B+\ker H)$.

For the other inequality we use the inequality in \eqref{eq:sings}: For any $a\in A, b\in B$
\[
 \Vert H(a-b)\Vert \geq \sigma_{\min}(H) \dist( a-b, \ker H ) \geq \sigma_{\min}(H) \dist(A+\ker H, B+\ker H).
\]
Taking the infimum over $a$, $b$ we get $\dist(HA,HB)\geq \sigma_{\min}(H)\dist(A+\ker H, B+\ker H)$.
\end{proof}

The distance $\dist(A+\ker H, B+\ker H)$ depends only on $\ker H$ but not on $H$ itself.
Consequently, \cref{lem:sing-value-inequality} gives different approximations for different matrices with the same kernels.
The optimal choice is the orthogonal projection onto $\ker(H)^\perp$.
The singular values of the orthogonal projection are all $1$, so in this case the inequality from \cref{lem:sing-value-inequality}
becomes an equality and we get the following.

\begin{corollary}\label{cor:ort-proj-dist}
For nonempty subsets $A,B\subset\CC^n$ and an orthogonal projection $P:\CC^n\rightarrow \CC^n$, we have
\[
  \dist(A+\ker P, B ) = \dist(A+\ker P, B+\ker P ) = \dist(P(A), P(B) ).
\]
\end{corollary}

We will use \cref{lem:sing-value-inequality} mostly in the case where $H$ is the matrix of rational invariants
defined in the previous section. In this case $H$ is integral and the singular values can be bounded
in terms of the bit-length of $H$.

\begin{lemma}
\label{lem:singular}
Suppose $k\leq n$ and $H\in\ZZ^{k\times n}$ is an integer matrix of rank $\rk H=k$. Then
\[
   \sigma_{\max}(H)\leq n \Vert H\Vert_{\max},\quad \sigma_{\min}(H) \geq n^{-(n-1)} \Vert H\Vert_{\max}^{-(n-1)},
\]
where $\Vert H\Vert_{\max}\coloneqq\max_{i,j}|H_{ij}|$ is the max norm of $H$.
Consequently, if $B\coloneqq\langle H\rangle$ is the bit-length of $H$, then
\[
     \kappa(H)\coloneqq \sigma_{\max}(H)/\sigma_{\min}(H) \leq n^n 2^{Bn}.
\]
Given an $H$ as above, one can compute in polynomial time a number~$D\in\QQ$ such that
$\sigma_{\min}(H) \le D \le 2 \sigma_{\min}(H)$.
\end{lemma}

\begin{proof}
The upper bound on $\sigma_{\max}(H)$ follows from
$\Vert Hx\Vert  \leq  \sqrt{nk}\Vert H\Vert_{\max} \Vert x\Vert $.
For the lower bound, note that the product of the singular values of $H$ equals
$\prod_{i=1}^k \sigma_i(H)=\sqrt{\det(HH^T)}$.
Since $H$ has rank $k$, the matrix $HH^T$ is invertible and positive definite.
Hence, $\det(HH^T)\geq 1$ and we deduce with part one that
\[
  \sigma_{\min}(H) \geq \  \sigma_{\max}(H)^{1-k}  \sqrt{\det(HH^T)} \geq n^{-(k-1)}\Vert H\Vert_{\max}^{-(k-1)}
    \geq n^{-(n-1)}\Vert H\Vert_{\max}^{-(n-1)},
\]
which shows the second inequality.

We refer to~\cite{pan-chen:99} for the algorithmic claim on computing $\sigma_{\min}(H)$.
\end{proof}

\subsection{Complexity of elementary functions}
Two important functions used in this paper are the complex logarithm and exponentiation, $\log$ and $\exp$.
Both functions are transcendental and rarely assume rational values on rational inputs.
Fortunately, they can be efficiently approximated by rational functions with
the arithmetic mean-geometric mean iteration of Gauss, Lagrange and Legendre.
We refer to the paper~\cite{BorweinBorwein:88} and the book~\cite{agm} for a detailed study of the AM-GM iteration
and for the following results.

\begin{lemma}
\label{lem:log}
 Suppose $\log$ denotes the standard branch of the complex logarithm whose imaginary part satisfies $\Im(\log(z))\in (-\pi,\pi]$, and $\Log$ is the componentwise logarithm. Assume that $v\in(\QQ(\ii)^\times)^n$ is a vector with Gaussian rational entries.
 Given $\varepsilon\in\QQ_{>0}$, in $\poly(n,\langle v\rangle,\log\varepsilon^{-1})$-time we can compute an approximation to
 $\Log v$ with absolute error $\varepsilon$. That is, we can compute a vector $\eta\in\QQ^n+2\pi\ii\QQ^n$ such that
 \[
	\Vert \Log v-\eta\Vert\,<\,\varepsilon.
 \]
\end{lemma}

\begin{lemma}
\label{lem:exp}
 Suppose $\exp$ denotes the complex exponentation, $\exp(\rho+\ii\theta)=e^{\rho}(\cos(\theta)+\ii\sin(\theta))$ and
 $\Exp$ is the componentwise $\exp$. Assume that $\eta\in \QQ^n+2\pi\ii\QQ^n$.
 Given $\varepsilon\in\QQ_{>0}$, in $\poly(n,\langle \eta\rangle,\log\varepsilon^{-1})$-time
 we can compute an approximation to $\Exp(\eta)$ with relative error~$\varepsilon$.
 That is, we can compute a vector $v\in(\QQ(\ii)^\times)^n$ such that
 \[
	\frac{\Vert \Exp(\eta)-v\Vert}{\Vert \Exp(\eta)\Vert}\,<\,\varepsilon.
 \]
\end{lemma}

\subsection{Quotient topology and quotient metric}\label{sec:quotients}
Throughout the paper, we consider various group actions on various metric spaces. The orbit space, i.e.,
the set of orbits of the action, has a natural quotient topology and it is an important question for us to decide
when it is possible to carry over the metric space structure to the orbit space.
This is possible in all cases we consider in this section.

\begin{definition}
Suppose $X$ is a topological space and $\sim$ is an equivalence relation on~$X$.
We denote by $\Xquot\coloneqq\{[x]\mid x\in X\}$ the set of equivalence classes of $X$,
where $[x]$ is the equivalence class of $x$. Let $\pi:X\rightarrow \Xquot$ be the canonical projection map.
The \textit{quotient topology} on $\Xquot$ is defined by calling
a subset $U$ of $\Xquot$ open iff $\pi^{-1}(U)$ is open in~$X$.
\end{definition}

Assume that $(X,\delta)$ is a metric space with distance function $\delta$.
We define the induced distance function on $\Xquot$ by
\[
\delta([x],[y]) \coloneqq \inf\{\delta (x',y')\mid x\sim x', y\sim y'\} .
\]
Clearly, this is a \textit{pseudo-metric} on $\Xquot$, which means that
$\delta(x,x)=0, \, \delta(x,y)=\delta(y,x)$
and $\delta(x,z)\leq\delta(x,y)+\delta(y,z)$ for all $x,y,z\in X$.
Moreover, one can show that $\delta$ is \textit{compatible}
with the quotient topology on $\Xquot$,
which means that the balls
$\{ [x] \mid \delta([x],[y]) < r \}$
form a basis of the quotient topology on $\Xquot$,
see~\cite[Theorem~4]{himmelberg}.
A pseudo-metric is a metric iff $\delta(x,y)=0$ implies $x=y$.

We note that the induced distance is not always a metric.
Indeed, the quotient topology on~$\Xquot$ does not need to be Hausdorff.
(The so-called line with two origins is an example for this.)

However, in the following two cases,
the pseudo-metric $\delta([x],[y])$ is indeed a metric.

\begin{proposition}
\label{prop:compact-quot}
Suppose $(X,\delta)$ is a metric space and $K$ is a compact topological group
acting continuously on $X$. 
For $x\in X$, we denote by $\KK_x$ the orbit of $x$.
Then $\delta(\KK_x, \KK_y)$ defines a metric on the space of orbits
$X/K := \{ \KK_x \mid x\in X \})$,
which is compatible with the quotient topology.
\end{proposition}

\begin{proof}
The orbit $\KK_x$ is compact since it is the image of $K$ under the continuous map~$k\mapsto k\cdot x$.
By the above observation, $\dist(\KK_x,\KK_y)$ defines a compatible pseudo-metric on $X/K$.
Different orbits $\KK_x,\KK_y$ are disjoint. Since they are compact, they have a positive distance:
$\dist(\KK_x,\KK_y)>0$.
This shows that the induced distance is a metric.
\end{proof}

For the second example, let $V$ be a finite dimensional real vector space
and $\mathcal{L}\subseteq V$ be a lattice. This means that
$\mathcal{L}$ is the set of $\ZZ$-linear combinations of a collection of
linearly independent vectors in $V$.
Alternatively, a lattice $\mathcal{L}$ can be defined as an additive subgroup of $V$
that is \textit{discrete},
i.e., there exists a positive constant $\varepsilon>0$ such that
for any distinct lattice points $x,y\in\mathcal{L}$ we have
$\dist(x,y)\geq\varepsilon$.
Thus, any convergent sequence of lattice points must be eventually constant.

\begin{proposition}
\label{prop:lattice-quot}
Let $V$ be a finite dimensional real vector space,
$\mathcal{L}$ a lattice in $V$,
and~$\delta$ a translation invariant metric on~$V$ that is compatible with the standard topology on $V$.
Then~$\delta(x+\mathcal{L}, y+\mathcal{L})$ defines a metric on the quotient space
$V/\mathcal{L}$ which is compatible with the quotient topology.
\end{proposition}

\begin{proof}
We argue as in the proof of \cref{prop:compact-quot}.
Suppose that $\delta(x+\mathcal{L},y+\mathcal{L})=0$.
Then there exist sequences $u_i$ and $v_i$ in $\mathcal{L}$ such that
$\delta(x+u_i,y+v_i)= \delta(x-y,v_i-u_i)$ converges to zero.
Since the sequence $v_i-u_i$ of lattice points converges,
it must be eventually constant. Hence $x-y\in\mathcal{L}$,
which completes the proof.
\end{proof}

Sometimes in the paper the lattice $\mathcal{L}$ will be given
as an orthogonal projection of another lattice.
Note that the orthogonal projection of a lattice is not always a lattice.
(For instance, the orthogonal projection of $\ZZ^2$ onto
$\RR(1,y)$ with irrational $y$ is not discrete.)

\begin{lemma}\label{lem:ort-proj-lattice}
Suppose $U\subset V$ is a rational subspace of $V=\RR^n$,
spanned by linearly independent rational vectors
$u_1,\dots,u_k\in\QQ^n$ and $\mathcal{L}\subset V$ is a lattice, generated by integral vectors.
Then, the orthogonal projection of $\mathcal{L}$ to $U$ is a lattice. 
\end{lemma}

\begin{proof}
Using Gram-Schmidt orthogonalization,
we may assume that the $u_i$ are pairwise orthogonal.
Then the orthogonal projection
$P:\RR^n\rightarrow U$ is given by
$P(v) = \sum_{i=1}^k \frac{\langle u_i, v\rangle}{\langle u_i,u_i \rangle} \, u_i$.
This shows that $P(\QQ^n)\subset\QQ^n$.
Therefore, if $v_1,v_2,\dots,v_l\in\ZZ^n$ generate $\mathcal{L}$,
there is a positive integer~$N$ such that $N\, P(v_i)\in\ZZ^n$ for all $i$.
Hence  $N\,P(\mathcal{L})\subset\ZZ^n$, which
implies that $P(\mathcal{L})$ is discrete. 
\end{proof}

\section{Logarithmic image of orbits}\label{sec:log}

The goal of this section is to discuss the structure of the logarithmic image of orbits in the quotient space $\CC^n/2\pi\ii\ZZ$,
to study the metric $\distLD$ defined in \eqref{eq:def-distLD},
and to prove that $\distLD(\OO_v,\OO_w)$ is approximated by a linear form in~$\Log v$ and~$\Log w$.

We assume the setting of \cref{se:invar-th-torus}:
$T=(\CC^\times)^d$ acts on~$V=\CC^n$ via the weight matrix~$M\in\ZZ^{d\times n}$.
The columns of $M$ are called the weights of the action and
the \textit{weight polytope} $P\subset\RR^d$ is defined as the convex hull of these weights.
We consider orbits of vectors~$v\in (\CC^\times)^n$, thus we assume that all components of~$v$ are nonzero.
We note that~$\OO_v$ is closed in~$(\CC^\times)^n$,
see~\cite[Prop.~5.1]{torus}.
However, $\OO_v$ may not be closed.
It is well known that~$\OO_v$ is closed iff~$0$ lies in the interior of~$P$;
see~\cite[Section~3]{torus}. This will only become relevant later in \cref{sec:EA-KN0}.


\subsection{Structure of logarithms of orbits}
The key idea is to use the group isomorphism
$\Log \colon (\CC^\times)^n \rightarrow \CC^n/2\pi\ii\ZZ^n$
provided by the componentwise complex logarithm, compare \cref{eq:Log-def}.
The inverse is given by the (componentwise) exponential function $\Exp$.

The action of the group~$T$ on $(\CC^\times)^n$ induces via $\Log$
an action of~$T$ on $\CC^n/2\pi\ii\ZZ^n$, which we shall denote by $\ast$.
This action is simpler to understand, since it works by \textit{translations}.
More specifically, for $x\in \CC^n$ and $v\in(\CC^\times)^n$, we have
\begin{equation}
\label{eq:translation}
	\Log(e^x\cdot v) = \Log v + M^T x ,
\end{equation}
see \cref{fig:flatten-action}.
Note that, by the periodicity of $\exp$, the right-hand side
in $\CC^n/2\pi\ii\ZZ^n$
indeed only depends on $x\bmod 2\pi\ii\ZZ^n$.
(We could also consider the corresponding action of the Lie algebra $\Lie(T)=\CC^d$,
which has the same orbits.)
This leads us to the following definition.

\begin{definition}\label{def:ast}
We denote by $\ast$ the induced action of $T$ on $\CC^n/2\pi\ii\ZZ^n$ via translations
defined by
\[
  e^x \ast \eta \coloneqq \eta+M^Tx, \qquad \text{for }x\in \CC^d,\; \eta\in\CC^n/2\pi\ii\ZZ^n.
\]

\end{definition}

By construction, the orbits~$\OO_v$ and~$\KK_v$ of~$v\in (\CC^\times)^n$ are mapped
to the corresponding orbits of~$\eta=\Log v$.
If we denote
\[
 T\ast\eta \coloneqq \{t\ast\eta\mid t\in T\}\quad\text{and}\quad K\ast\eta\coloneqq\{k\ast\eta\mid k\in K\} ,
\]
then
$\Log (\OO_v) = T\ast \Log(v)$ and $\Log (\KK_v) = K\ast \Log(v)$.

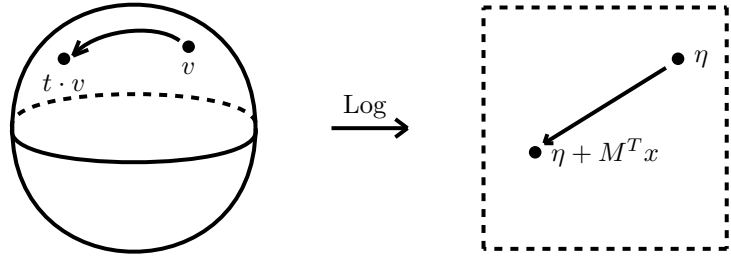
\begin{figure*}[ht]
	\centering
    \begin{tikzpicture}[scale=0.8]
    \pgfdeclarelayer{nodelayer}
    \pgfdeclarelayer{edgelayer}
    \pgfsetlayers{main,nodelayer,edgelayer}
	\begin{pgfonlayer}{nodelayer}
		\node  (0) at (-2, 0) {};
		\node  (1) at (2, 0) {};
        \node  [fill = black, circle, label={below:$v$}, scale=0.5](20) at (0.9, 1.35) {};
		\node  (4) at (0.75, 1.45) {};
		\node  (5) at (-1, 1.2) {};
        \node  [fill = black, circle, label={below:$t\cdot v$}, scale=0.5](21) at (-1.15, 1.15) {};
		\node  (6) at (-0.9, 1.45) {};
		\node  (7) at (-0.75, 1.3) {};
		\node  (8) at (3.25, 0) {};
		\node  (9) at (4.5, 0) {};
        \node [label={$\Log$}] (24) at (3.8, -0.15) {};
		\node  (10) at (4.3, 0.15) {};
		\node  (11) at (4.3, -0.15) {};
		\node  (12) at (5.75, 2) {};
		\node  (13) at (5.75, -2) {};
		\node  (14) at (9.75, -2) {};
		\node  (15) at (9.75, 2) {};
        \node  [fill = black, circle, scale=0.5, label = {right:$\eta$}] (23) at (8.95, 1.15) {};
		\node  (16) at (8.75, 1) {};
        \node  [fill = black, circle, scale=0.5, label = {right:$\eta+M^T x$}] (22) at (6.6, -0.4) {};
		\node  (17) at (6.75, -0.25) {};
		\node  (18) at (6.80, -0.1) {};
		\node  (19) at (6.9, -0.25) {};
	\end{pgfonlayer}
	\begin{pgfonlayer}{edgelayer}
		\draw [bend left=90, looseness=1.75, line width = 1.5pt] (0.center) to (1.center);
		\draw [bend right=90, looseness=1.75, line width = 1.5pt] (0.center) to (1.center);
		\draw [bend right=105, looseness=0.50, line width = 1.5pt] (0.center) to (1.center);
		\draw [dashed, bend left=105, looseness=0.50, line width = 1.5pt] (0.center) to (1.center);
		\draw [bend right=45, looseness=0.75, line width = 1.5pt] (4.center) to (5.center);
		\draw [line width = 1.5pt](5.center) to (6.center);
		\draw [line width = 1.5pt](5.center) to (7.center);
		\draw [line width = 1.5pt](8.center) to (9.center);
		\draw [line width = 1.5pt](9.center) to (10.center);
		\draw [line width = 1.5pt](9.center) to (11.center);
		\draw [dashed, line width = 1.5pt](12.center) to (15.center);
		\draw [dashed, line width = 1.5pt](15.center) to (14.center);
		\draw [dashed, line width = 1.5pt](14.center) to (13.center);
		\draw [dashed, line width = 1.5pt](13.center) to (12.center);
		\draw [line width = 1.5pt](16.center) to (17.center);
		\draw [line width = 1.5pt](17.center) to (19.center);
		\draw [line width = 1.5pt](17.center) to (18.center);
	\end{pgfonlayer}
\end{tikzpicture}
\caption{\small The logarithm \textit{flattens} the group action.
 The vector $\eta=\Log v$ is translated to $\eta+M^T x$ via the group element $t= e^{x}\in T$.}%
 \label{fig:flatten-action}
\end{figure*}

We observe that the $\ast$-orbits of $T$ (and $K$) are the
\textit{images of affine subspaces}
under the canonical surjection $\CC^n\rightarrow\CC^n/2\pi\ii\ZZ^n$.
More concretely, we denote by
$U_\CC:=\ran M^T = \{M^T x \mid x\in \CC^d\}$
the row space of $M$.
Then it immediately follows from the definition of the $\ast$-action that
\begin{equation}\label{eq:row-space}
	T\ast\eta = \Big(\eta + U_\CC +2\pi\ii\ZZ^n\Big)\,\big/\,2\pi\ii\ZZ^n .
\end{equation}
We equip now $\CC^n/2\pi\ii\ZZ^n$ with the quotient metric~$\Delta$
induced by the Euclidean metric on $\CC^n$, denoted by $\dist$
(by \cref{prop:lattice-quot} this is indeed a metric).
Explicitly,
\begin{equation}
\label{eq:delta}
	\begin{split}
	\Delta\big(\, \eta\, ,\, \zeta\, \big)\, &\coloneqq\, \dist\big(\, \eta-\zeta \, ,\, 2\pi\ii\ZZ^n \, \big)\\
	&=\sqrt{\Vert\, \rho-\tau\,\Vert^2 \, +\,4\pi^2\, \dist^2\left(\,\theta-\phi\,,\, \ZZ^n\right)}\, ,
\end{split}
\end{equation}
where $\eta=\rho+2\pi\ii\theta$ and $\zeta=\tau+2\pi\ii\phi$
with $\rho,\theta,\tau,\phi\in\RR^n$
(note that the imaginary parts~$\theta,\phi$ are only determined modulo $\ZZ^n$).

\begin{proposition}\label{prop:delta-is-a-metric}
The metric $\Delta$ on $X:=\CC^n/2\pi\ii\ZZ^n$ induces a metric
$$
 \Delta(T\ast \eta, T\ast \zeta) := \inf_{t\in T} \Delta( t\ast \eta,\zeta)
$$
on the space $X/T$ of $T$-orbits with respect to the $\ast$-action,
which is compatible with the quotient topology.
An analogous result holds for $K$-orbits.
\end{proposition}

\begin{proof}

Denote by $P:\CC^n\rightarrow U_\CC^\perp$ the orthogonal projection onto the orthogonal complement of $U_\CC$,
thus $\ker P = U_\CC$.
\Cref{lem:ort-proj-lattice} shows that $2\pi\ii P(\ZZ^n)$ is a lattice
and \cref{prop:lattice-quot} implies that
$Y := U_\CC^\perp/P(2\pi\ii\ZZ^n)$
is a metric space with respect to the quotient metric of the Euclidean metric.

The projection $P$ induces a surjective group morphism
$P'\colon X \to Y$.
From \cref{eq:row-space} we see that $P'$
is constant on $T$-orbits. More specifically,
$P'$ induces a continuous bijection
$P''\colon X/T \to Y$.
By definition, $P''$ preserves the distance $\Delta$:
$$
 \Delta(T\ast \eta, T\ast \zeta) = \Delta(\eta, \zeta) .
$$
Since $\Delta$ is a metric on $Y$, this implies that
the pseudo-metric $\Delta$ on $X/T$ in fact is a metric.
Therefore, $\Delta(T\ast\eta,T\ast\zeta)=0$ implies $T\ast\eta=T\ast\zeta$.
The claim about $K$-orbits follows analogously.
\end{proof}

\begin{remark}
The fact that $\Delta$ is metric on $X/T$ is our essential gain.
Note that the orbit space $(\CC^\times)/T$ equipped with the induced
Euclidean distance $\dist(\OO_v,\OO_w)$ between $T$-orbits
is not a metric: it can be zero even though $\OO_v\neq\OO_w$,
as illustrated in the case of the hyperbolas in \cref{fig:example-actions}.
\end{remark}

\Cref{eq:delta} implies the important equations
\begin{equation}\label{eq:Delta=dist}
 \Delta(T\ast\eta,T\ast \zeta) = \dist( \eta-\zeta + U_\CC , 2\pi\ii\ZZ^n) ,\quad
 \Delta(K\ast \ii\theta, K\ast \ii\phi) = \dist(\theta-\phi + U, \ZZ^n) ,
\end{equation}
where $U_\CC$ and $U$ are the complex and real row spaces  of $M$,
respectively.

The logarithmic distance of $v,w\in (\CC^\times)^n$,
introduced in \cref{eq:def-distLD},
can now be expressed as
\begin{equation}\label{eq:logD=Delta}
	\distLD(v,w) = \Delta(\Log v, \Log w) .
\end{equation}
Hence the distances of orbits in the metric $\distLD$
are given by
\begin{equation}\label{eq:delta-ld}
  \distLD(\OO_v, \OO_w) = \Delta(T\ast\Log v, T\ast\Log w) ,\quad
  \distLD(\KK_v,\KK_w)=\Delta(K\ast \Log v, K\ast \Log w) .
\end{equation}

\begin{figure}
    \centering
    \begin{tikzpicture}[scale=0.4]
    \pgfdeclarelayer{nodelayer}
    \pgfdeclarelayer{edgelayer}
    \pgfsetlayers{main,nodelayer,edgelayer}
	\begin{pgfonlayer}{nodelayer}
		\node  (0) at (0, -4) {};
		\node  (1) at (0, 4) {};
		\node  (2) at (8, -4) {};
		\node  (3) at (8, 4) {};
		\node  [fill = black, scale=0.4, circle, label={right:$0$}](4) at (4, 0) {};
		\node  (5) at (0, 2) {};
		\node  (6) at (6, -4) {};
		\node  (7) at (2, 4) {};
		\node  (8) at (8, -2) {};
		\node  (9) at (10, 0) {};
		\node  (10) at (12, 0) {};
		\node  (11) at (11.75, 0.25) {};
		\node  (12) at (11.75, -0.25) {};
		\node  (13) at (14.5, 4) {};
		\node  (14) at (14.5, -4) {};
		\node  (15) at (14, -3.5) {};
		\node  (16) at (22, -3.5) {};
		\node  (17) at (15.25, 4.25) {};
		\node  (18) at (22.25, -2.75) {};
		\node  (19) at (15, 4.25) {};
		\node  (20) at (22.25, -3) {};
		\node  (21) at (22.25, -2.5) {};
		\node  (22) at (15.5, 4.25) {};
		\node  [fill = black, scale=0.4, circle, label={below:$v$}](23) at (17.1, -1.25) {};
	\end{pgfonlayer}
	\begin{pgfonlayer}{edgelayer}
		\draw [line width=0.5pt](3.center) to (1.center);
		\draw [line width=0.5pt](1.center) to (0.center);
		\draw [line width=0.5pt](0.center) to (2.center);
		\draw [line width=0.5pt](2.center) to (3.center);
		\draw [line width=1.5pt](1.center) to (2.center);
		\draw [dashed, line width=1pt](5.center) to (6.center);
		\draw [dashed, line width=1pt](7.center) to (8.center);
		\draw [line width=1.5pt](9.center) to (10.center);
		\draw [line width=1.5pt](10.center) to (11.center);
		\draw [line width=1.5pt](10.center) to (12.center);
		\draw [line width=1.5pt](13.center) to (14.center);
		\draw [line width=1.5pt](15.center) to (16.center);
		\draw [bend right=45, looseness=1.25, line width=1.5pt] (17.center) to (18.center);
		\draw [dashed, bend right, looseness=1.25, line width=1pt] (22.center) to (21.center);
		\draw [dashed, bend right=45, looseness=1.75, line width=1pt] (19.center) to (20.center);
	\end{pgfonlayer}
\end{tikzpicture}
\caption{Suppose $\CC^\times$ acts on $\CC^2$ via $t\cdot (x,y)=(tx,t^{-1}y)$ as in \cref{fig:example-actions}.
The logarithmic orbit $T\ast 0$ is mapped via $\Exp$ to the orbit $\OO_v$ of $v\coloneqq (1,1)$.
The image shows the $\varepsilon$-neighbourhood of $T\ast 0$ and its image under $\Exp$.
\label{fig:kempf ness example}
}
\end{figure}

The action $*$ splits into a real and imaginary part:
for $x=y+\ii z$ with $y,z\in\RR^d$ and $\eta = \rho+2\pi\ii\theta$
with $\rho\in\RR^n$, $\theta\in\CC^n$, we have by \cref{def:ast},
\begin{equation}\label{eq:indep}
	e^x \ast \eta
  = \big(\rho + M^T y\big) + \ii\big( 2\pi\theta + M^T z\big)
  = e^y \ast \rho +e^{\ii z}\ast (2\pi\ii\theta) ,
\end{equation}
hence
\begin{equation*}
		T\ast\eta
		= \Big( \big(\rho+ U)\big) + \ii \big(2\pi\theta+ U +2\pi\ZZ^n\big)\Big)\big/2\pi\ii\ZZ^n .
\end{equation*}
Let $\zeta=\tau+2\pi\ii\phi$.
Together with \cref{eq:delta,eq:Delta=dist}, we obtain the following formula for the $\Delta$-distance between $\ast$-orbits:
\begin{equation}\label{prop:delta-orbit-distance}
\begin{aligned}
\Delta^2(T\ast\eta,T\ast\zeta)
 &= \dist^2(\rho-\tau, U) + 4\pi^2 \Delta^2(K\ast \ii\theta, K\ast \ii\phi) \\
 &= \dist^2(\rho-\tau, U) + 4\pi^2 \dist^2(\theta-\phi + U, \ZZ^n)
\end{aligned}
\end{equation}
and similarly,
\begin{equation}\label{prop:Kdelta-orbit-distance}
\begin{aligned}
 \Delta^2(K\ast\eta, K\ast\zeta)
 &= \Vert \rho-\tau \Vert^2 + 4\pi^2 \Delta^2(K\ast \ii\theta, K\ast \ii\phi) \\
 &= \Vert \rho-\tau \Vert^2 + 4\pi^2 \dist^2(\theta-\phi + U, \ZZ^n) .
\end{aligned}
\end{equation}

The contributions of the real parts,
$\dist^2(\rho-\tau,U)$ in the $T$-case and $\Vert\rho-\tau\Vert^2$ in the $K$-case,
are easy to compute.
On the other hand, the contribution of the imaginary parts
$\dist^2(\theta-\phi + U, \ZZ^n)$,
turn out to be difficult to compute.
We will show in \cref{sec:hardness} that approximating
$\dist(\theta-\phi +U,\ZZ^n)$ within a subpolynomial factor is \textsc{NP}-hard,
by providing a polynomial time reduction from the \textit{closest vector problem} (CVP) to it.
By contrast, deciding whether this distance equals zero can be done in polynomial time, see \cref{re:equ-lattice-dist}.
(This is analogous to the \textsc{CVP} problem, see \cref{sec:sldp-to-cvp}.)

\begin{remark}\label{re:equ-lattice-dist}
Given a subspace $U\subset\RR^n$
by generators $v_1,\ldots,v_m\in\ZZ^n$ and
$t\in\QQ^n$, we can decide
$(t + U) \cap \ZZ^n \neq \varnothing$
in polynomial time. This can be seen by
putting the matrix with columns $v_1,\ldots,v_m$
into Smith normal form, which is
possible in polynomial time~\cite{snf}.
\end{remark}

\subsection{Approximation of orbit distances by linear forms in logarithms: \texorpdfstring{$T$}{T}-orbits}
We use now invariant theory to analyze the orbits.
Let $M\in\ZZ^{k\times n}$ be a weight matrix and denote by
$H\in\ZZ^{k\times n}$ a matrix of rational invariants
of $H\in\ZZ^{k\times n}$, see \cref{prop:allaboutH}.

For $v,w\in(\CC^\times)^n$, \cref{prop:t-invariants} characterizes the equality of orbits~$\OO_v = \OO_w$ by~$H\eta = H\zeta$, where~$\eta =\Log v$ and~$\zeta = \Log w$.
Our goal is a robust version of this:
to show that the closeness of the corresponding logarithmic orbits
$T\ast \eta$ and $T\ast \zeta$,
measured in terms of the metric $\Delta$ on $\CC^n/2\pi i\ZZ^n$,
is quantitatively related to the closeness of
$H\eta$ and $H\zeta$, measured in terms of the metric~$\Delta$
on the quotient space $\CC^k/2\pi i \ZZ^k$.
The correction factors are provided by the minimum and maximum singular values of~$H$.

\begin{proposition}\label{cor:delta-ld-linear-forms} 
If $\eta,\zeta\in\CC^n/2\pi\ii\ZZ^n$, we have for the distance of $T$-orbits:
\[
\sigma_{\min}(H)\, \Delta(T\ast\eta,T\ast\zeta) \,\leq\, \Delta(H\eta,H\zeta)
   \,\leq\, \sigma_{\max}(H) \, \Delta(T\ast\eta,T\ast\zeta).
\]
For $v,w\in(\CC^\times)^n$ we have
\[
 \sigma_{\min}(H)\; \distLD(\OO_v, \OO_w) \,\leq\, \Delta(H\Log v,\, H\Log w)\,\leq\,
  \sigma_{\max}(H)\; \distLD(\OO_v, \OO_w).
\]
If $\eta,\zeta$ are purely imaginary, then the same bounds hold for the distances of $K$-orbits.
\end{proposition}

\begin{proof}
We have
$\Delta(T\ast\eta,T\ast \zeta) = \dist( \eta-\zeta + U_\CC , 2\pi\ii\ZZ^n)$
by \cref{eq:Delta=dist},
where $U_\CC=\ker H =\ran M^T$.
The first assertion follows from \cref{lem:sing-value-inequality}
with $A =\{\eta-\zeta\}$ and $B=2\pi\ii\ZZ^n$.
The second assertion is just a rewriting of the first, using \cref{eq:delta-ld}.
The statement on $K$-orbits follows analogously.
\end{proof}

\subsection{Approximation of orbit distances by linear forms in logarithms: \texorpdfstring{$K$}{K}-orbits}
Now the task is to provide an approximation
for the distance $\dist(\KK_v,\KK_w)$ between $K$-orbits,
similarly to \cref{cor:delta-ld-linear-forms}.
The analogous formula in \cref{cor:K-orbits-approximation} below is more complicated.
It involves a priori upper and lower bounds on the norms of $v$ and $w$.
For $0<r\leq R$ we consider the closed region defined by
\begin{equation}\label{eq:DrR}
 D_{r, R} \,\coloneqq\, \{ v\in\CC^n \, \mid\, \forall i\in [n]\; r \leq |v_i|\leq R \}.
\end{equation}

\begin{lemma}
\label{lem:delta-dist}
Let $v,w\in D_{r, R}$ and put $\eta = \Log v$, $\zeta = \Log w\in\CC^n/2\pi\ii\ZZ^n$. Then
\[
 \frac{2r}{\pi}\, \Delta(\eta,\zeta) \, \leq \, \Vert v-w\Vert\, \leq \, R\, \Delta(\eta,\zeta).
\]
\end{lemma}

\begin{proof}
It is enough to show the equality for $n=1$.
We write here $\eta = \rho+\ii\theta,\, \zeta=\tau+\ii\phi$ (dropping the factor $2\pi$ to simplify notation).
The cosine theorem gives
\[
  |v-w|^2 = |v|^2 + |w|^2 - 2|v||w|\cos(\theta-\phi) = (|v|-|w|)^2 + 2|v||w|\left(1-\cos(\theta-\phi)\right).
\]
The first contribution on the right-hand side can be upper and lower bounded by
\[
	r^2(\rho-\tau)^2\leq (|v|-|w|)^2\leq R^2 (\rho-\tau)^2
\]
since the mean value theorem for $\exp$ implies (w.l.o.g.\ $\rho < \tau$)
$$
 r \le |v| = e^\rho \le \Big|\frac{e^\tau - e^\rho}{\tau - \rho}\Big| \le e^\tau = |w| \le R .
$$
The second contribution on the right-hand side can be be upper and lower bounded by
\[
 \frac{4 r^2}{\pi^2}\,\dist(  \theta-\phi,2\pi\ZZ)^2\, \leq \,
  2|v||w|\left(1-\cos(\theta-\phi)\right)\,\leq \,R^2 \, \dist(  \theta-\phi,2\pi\ZZ)^2,
\]
using the inequality
$\frac{4}{\pi^2}\psi^2\leq 2-2\cos\psi \leq \psi^2$ for $\psi\in[-\pi,\pi]$.
Bringing the inequalities together and observing that $4/\pi^2\leq 1$, we obtain
\[
 \frac{4r^2}{\pi^2} \left( (\rho-\tau)^2 + \dist(  \theta-\phi,2\pi\ZZ)^2 \right)
   \leq |v-w|^2 \leq R^2 \left( (\rho-\tau)^2 + \dist(  \theta-\phi,2\pi\ZZ)^2 \right),
\]
which completes the proof.
\end{proof}

We can now relate 
$\dist(\KK_v,\KK_w)$ to $\distLD(\KK_v,\KK_w)=\Delta(K\ast \Log v, K\ast \Log w)$.

\begin{proposition}
\label{thm:logorbits-equivalence}
For $v,w\in D_{r,R}$ we have
\[
 \frac{2r}{\pi}\, \Delta(K\ast \Log v,K\ast \Log w) \, \leq \, \dist(\KK_v,\KK_w) \, \leq \,
  R\, \Delta(K\ast \Log v,K\ast \Log w).
\]
\end{proposition}

\begin{proof}
Clearly, the $K$-action preserves $D_{r,R}$.
\Cref{lem:delta-dist} implies that
\[
 \frac{2r}{\pi} \Delta( \Log( k\cdot v ) , w ) \, \leq \,
  \Vert k\cdot v-w\Vert\, \leq \, R\, \Delta(\Log(k\cdot v), w)
\]
for all $k\in K$. The claim follows from by taking the infimum over $k\in K$.
\end{proof}

Combining \cref{thm:logorbits-equivalence} with
\cref{prop:delta-orbit-distance} and \cref{cor:delta-ld-linear-forms}, 
we obtain the following corollary.

\begin{corollary}\label{cor:K-orbits-approximation}
Assume that $0<r\leq R$ and $v,w\in D_{r,R}$. Suppose we have $\Log v=\rho+2\pi\ii\theta$ and
$\Log w = \tau+2\pi\ii\phi$.
Let $H$ by a matrix of rational invariants as in \cref{prop:allaboutH}. Then,
\[
 \frac{4r^2}{\pi^2} \big(\Vert\rho-\tau\Vert^2 +\frac{4\pi^2}{\sigma^2_{\max}(H) }
  \Delta^2(H\theta,H \phi)\big)  \leq \dist^2(\KK_v,\KK_w)
 \leq  R^2 \big(\Vert\rho-\tau\Vert^2 + \frac{4\pi^2}{\sigma^2_{\min}(H) }
   \Delta^2(H\theta, H\phi) \big) .
\]
\end{corollary}

We finally relate the separation parameters $\sep_K(d,n,B,b)$ and $\sep_T(d,n,B,b)$
defined in \cref{eq:def-sep}.

\begin{corollary}\label{cor:rel-btw-seps}
We have $\sep_T(d,n,B,b)\leq n\, 2^{O(b)}\sep_K(d,n,B,b)$.
\end{corollary}

\begin{proof}
Suppose that $\sep_K(d,n,B,b)= \dist(\KK_v,\KK_w)$.
We let $r$ be the minimum and $R$ be the maximum of $|v_i|,|w_i|, i=1,2,\dots,n$.
Clearly, $2^{-b} \le r$.
If moreover $\OO_v\neq\OO_w$, then
\cref{thm:logorbits-equivalence} implies
\[
 \dist(\KK_v,\KK_w)\geq \frac{2r}{\pi} \, \Delta(K\ast\Log v,K\ast\Log w)
    \geq \frac{2r}{\pi}\Delta(T\ast\Log v, T\ast\Log w)
     \geq \frac{2r}{\pi} \sep_T(d,n,B,b) ,
\]
hence
$\sep_T(d,n,B,b) \le \frac{\pi}{2r}\dist(\KK_v,\KK_w) \le 2^{O(b)} \sep_K(d,n,B,b)$.
On the other hand, if $\OO_v=\OO_w$, then \cref{prop:k-invariants} implies that $|v_i|\neq|w_i|$
for some~$i$. In this case,
$\sep_K(d,n,B,b) \ge\Big| |v_i| -|w_i| \Big|\ge 2^{-b}$.
We note that since $|v_i/w_i|\leq 2^{2b}$,  we have $\log(|v_i|/|w_i|)\leq 2b$. Furthermore, $\dist(\theta,2\pi\ZZ^n)\leq \sqrt{n}\pi$ for any $\theta\in\RR^n$. Hence,
$\sep_T(d,n,B,b)\le \Delta(\Log v, \Log w) \le  \sqrt{n}(2b+\pi \sqrt{n})$.
Altogether,
$$
 \sep_T(d,n,B,b)\le 2\sqrt{n}\, b + \pi n \le (\pi+2)\,n\, 2^b \le (\pi+2)\, n\, 2^{2b} \sep_K(d,n,B,b) ,
$$
which proves the assertion.
\end{proof}

\section{The separation hypotheses and the \abc-conjecture}\label{sec:abc}

The goal of this section is to explain in detail the connection between
the \abc-conjecture and \cref{hyp:linearforms} and to prove \cref{thm:abc-and-sep},
i.e., to show that
\cref{hyp:sep} and \cref{hyp:linearforms} are equivalent.

\subsection{The \abc-conjecture and the Number-Theoretic Hypothesis~\ref{hyp:linearforms}}\label{sec:abc-sec1}
Oesterl\'{e} and Masser's $abc$ conjecture~\cite[Conjecture 3]{oesterle} claims the following.

\begin{conjecture}[$abc$ conjecture] 
\label{conj:abc}
For every $\varepsilon>0$ there exists a constant $\kappa_\varepsilon>0$ such that
for all nonzero coprime integers $a,b,c$ satisfying $a+b+c=0$, we have
\[
	\max(|a|,|b|,|c|)\leq \kappa_\varepsilon \rad(abc)^{1+\varepsilon}.
\]
The radical $\rad(abc)$ is defined as the product of the distinct primes dividing $abc$,
taken with multiplicity one.
\end{conjecture}

This simple looking conjecture is one of the most powerful statements in number theory.
In the words of Dorian Goldfeld~\cite{goldfeld}:
\say{The \abc~conjecture is the most important unsolved problem in Diophantine analysis} since
\say{it provides a way of reformulating an infinite number of Diophantine problems---and, if it is true, of solving them.}

Baker~\cite{baker98} proposed a more precise version of the conjecture:
he conjectured that there exist absolute constants $\kappa,\kappa'>0$ (not depending on $\varepsilon$)
such that for all $\varepsilon>0$
\[
 \max(|a|,|b|,|c|) \leq \kappa \varepsilon^{-\kappa' \omega(ab)} \rad(abc)^{1+\varepsilon},
\]
where $\omega(ab)$ denotes the number of distinct prime factors of $ab$.
Moreover, Baker observed that\footnote{Baker attributes this observation to Andrew Granville.}
the minimum of the right-hand side over all $\varepsilon>0$ occurs when $\varepsilon=\omega(ab)/\log N$
and suggested the following $\varepsilon$-free version of the conjecture:

\begin{conjecture}[Baker's refinement of the $abc$ conjecture, {\cite[Conjecture 3]{baker_explicit}}]
\label{conj:refined}
There exist constants $\kappa,\kappa'>0$ such that
for all nonzero coprime integers $a,b,c$ satisfying $a+b+c=0$, we have,
setting $N:=\rad(abc)$,
\begin{equation}
	\label{eq:bakers-refinement}
	 \max(|a|,|b|,|c|) \leq \kappa N \Big(\frac{\log N}{\omega(ab)}\Big)^{\kappa'\omega(ab)}.
\end{equation}
\end{conjecture}

Baker's refinement of the $abc$-conjecture is closely related to lower bounds for linear forms in logarithms.
The following was observed in~\cite{baker98}. We include a proof for the sake of completeness.

\begin{theorem}
\label{thm:abc-bound}
Suppose Baker's refinement of the $abc$ conjecture (\cref{conj:refined}) is true.
Then there exists a constant $\kappa>0$ such that for any list of positive integers $v_1,v_2,\dots,v_n$
and any list of integers $e_1,e_2,\dots,e_n$ satisfying $v^e\coloneqq\prod_{i=1}^n v_i^{e_i}\neq 1$, we have
\[
 \log|v^e-1| \geq -\kappa \log(\max_i|e_i|)\sum_{i=1}^n \log v_i.
\]
\end{theorem}

\begin{proof}
W.l.o.g.\  $e_i\neq 0$ for all $i$,
$E := \max |e_i| \ge 2$, and $u := v_1v_2\dots v_n \ge 3$.
We write
$v^e= \frac{a'}{b'} = \frac{a}{b}$,
where
\[
   a' := \prod_{e_i> 0} v_i^{e_i}, \quad
   b' := \prod_{e_i<0} v_i^{-e_i} , \quad
 d:= \gcd(a', b'), \quad
 a := a'/d, \quad
b:= b'/d .
\]
We have
$a'b' = \prod v_i^{|e_i|}$, hence
$\rad(a'b') = \rad(u)$.
On the other hand,
$a'b'=d^2 ab$, which gives
$\rad(a'b') = d \rad(ab)$.
Hence $\rad(ab)$ is a divisor of $\rad(u)$,
which implies
$\omega(ab) \le \omega(u)$.

Setting $c\coloneqq a-b$, we have
$(-a)+b+c=0$ and $\gcd(a,b,c)=1$.
Since $a,b$ are positive, we have $\max(a,b,|c|)=\max(a,b)$.
We claim that $|c| \le u^E$,
where $ E := \max |e_i|$.
Indeed,
$c \le a \le a' \le u^E$,
and similarly,
$-c \le b \le b' \le u^E$.

With the above estimates, we obtain
\begin{equation}\label{eq:N-bound}
 N = \rad(abc) = \rad(c) \rad(ab) \le |c| \rad(u) \le |c|\, u \le u^{E +1} , \quad \omega(ab) \le \omega(u) .
\end{equation}
\Cref{conj:refined} implies
\[
 \max(a,b)\leq \kappa N \Big(\frac{\log N}{\omega(ab)}\Big)^{\kappa'\omega(ab)}\leq
  \kappa\, N (\log N)^{\kappa'\omega(ab)}.
\]
Using \cref{eq:N-bound},
we can bound this as
\[
 \max(a,b)\leq \kappa\, |c|\, u \big((E+1)\log u \big)^{2\kappa' \omega(u)}.
\]
Since $|v^e-1|=|c/b|$ we get
\[
 |v^e-1|\geq |c|/\max(a,b) \geq  \kappa^{-1} u^{-1} \big((E+1)\log u \big)^{-2\kappa' \omega(u)}.
\]
Taking logarithms of both sides, we obtain
\[
 \log|v^e-1| \geq -\log\kappa -\log u - 2\kappa' \, \omega(u)\big(\log(E+1) + \log \log u \big) .
\]
It is known that $\omega(u)\log\log u = O(\log u)$, see~\cite[\S 22.10]{hardy-wright:08}.
Using this, we conclude that indeed
$\log|v^e-1| \geq -\kappa''\, \log E \, \log u$
for a suitable constant $\kappa''>0$.
\end{proof}

Given $v_1,v_2,\dots,v_n\in\ZZ_{>0}$ and $e_1,e_2,\dots,e_n\in\ZZ$,
we denote by $\Lambda(v,e)$ the \textit{linear form in logarithms}
\begin{equation}\label{eq:DEF-Lambda}
 \Lambda(v,e)\coloneqq \log v^e = e_1 \log v_1 + e_2 \log v_2 +\dots +e_n \log v_n.
\end{equation}
A bound similar to the one in the previous theorem can also be given in terms of the quantity
$\log|\Lambda(v,e)|$, which relates the $abc$-conjecture to \cref{hyp:linearforms}.

\begin{corollary}\label{cor:abc-bound-log}
Assume Baker's refinement of the $abc$-conjecture is true.
Then there exists a constant $\kappa>0$ such that for all
$v_1,v_2,\dots,v_n\in\ZZ_{>0}$ and $e_1,e_2,\dots,e_n\in\ZZ$ with $\Lambda(v,e)\neq 0$,
we have
\[
 \log|\Lambda(v,e)|\geq -\kappa\log(\max_i |e_i|)\sum_{i=1}^n \log v_i.
\]
In particular, \cref{hyp:linearforms} is true if $v_1,\ldots, v_n\in\QQ_{>0}$.
\end{corollary}

\begin{proof}
First note that the function $h(x)\coloneqq \log(x)/(x-1)$ is monotonically decreasing
and satisfies
$\frac12 \le h(x) \le \frac32$
on the interval $[\frac12,\frac32]$.
This implies for $|x-1| \le \frac12$,
\begin{equation}\label{eq:Lambda-inequ}
    \frac{1}{2} |\log(x)| \leq |x-1| \leq 2 |\log(x)|.
\end{equation}
For showing the stated bound with positive integers $v_i$,
we may assume without loss of generality that $|v^e-1|\leq \frac{1}{2}$.
Then \cref{eq:Lambda-inequ} implies
$\frac{1}{2}|\Lambda(v,e)| \, \leq \, |v^e-1|$
and \cref{thm:abc-bound} gives the desired bound.

The assertion for positive rationals $v_i = \frac{p_i}{q_i}$ follows from
the observation that
$\Lambda(v,e) = \Lambda(\tilde v,\tilde e)$, where $\tilde v$ is defined by concatenating~$p$ and~$q$, and $\tilde e$ by concatenating~$e$ and~$-e$.
\end{proof}

\begin{remark}
Conversely, lower bounds similar to the one in \cref{cor:abc-bound-log},
together with their $p$-adic versions imply the $abc$-conjecture.
For a prime $p$, let $|x|_p$ denote the $p$-adic absolute value of~$x\in\QQ$,
e.g., $|p^\nu y|_p = p^{-\nu}$ for $y\in\ZZ$ not divisible by~$p$ and $\nu\in\ZZ$.
Suppose the $v_i$ are nonzero integers, $e_i\in\ZZ$ and write
$v_1^{e_1} v_2^{e_2}\dots v_n^{e_n} = a/b$ for coprime integers $a,b$.
If $p$ is a prime dividing $c\coloneqq a-b$,
then $|b|_p=1$ and $|c/b|_p = |c|_p <1$.
Hence the $p$-adic logarithm $\log(a/b)=\log(1+c/b)$ exists 
and $|\log(1+c/b)|_p = |c|_p$.
In analogy with \cref{eq:DEF-Lambda}, 
we define
\[
 |\Lambda(v,e)|_p := |\log(1+c/b)|_p = |c|_p.
\]
Baker considered the product of linear forms in logarithms
\[
 \Xi \coloneqq \min( 1 , |\Lambda(v,e)| ) \, \prod_{p\text{ prime}} \min( 1, p |\Lambda(v,e)|_p) ,
\]
where the product is over all primes.
He proved that a slightly stronger lower bound for $\Xi$ than the one of \cref{cor:abc-bound-log}
is equivalent to his refinement of the $abc$-conjecture \cref{eq:bakers-refinement};
see~\cite[Section~5]{baker98}.
Moreover, any lower bound for $\Xi$ implies a version of the $abc$-conjecture:
Stewart and Yu~\cite{yu-stewart},
refining an earlier work by Stewart and Tijdemann~\cite{stewart_oesterle-masser_1986},
used this approach to prove the inequality
\[
 \log c < \kappa_\varepsilon\, N^{\frac{1}{3}+\varepsilon},
\]
using the Baker-W\"{u}stholz bound for the linear forms in logarithms~\cite{wustholz1993}
and its $p$-adic versions by van der Poorten~\cite{vdpoorten}.
We refer to~\cite[Chapter~3.7]{baker_wustholz_2008}
(see also~\cite{granvilletucker}) for a summary of these results.
\end{remark}

In \cref{thm:abc-bound} and \cref{cor:abc-bound-log}
we assumed that the vectors $v_i$ to be positive rationals.
More generally, for our purposes, we would like to allow the $v_i$ to
be nonzero Gaussian rationals.
In the following, assume that $\log:\CC^\times\rightarrow\CC$ denotes the principal branch of the logarithm,
so $\Im(\log z)\in (-\pi,\pi]$, e.g., $\log(-1)=\ii\pi$.
In analogy, we define  the linear form in logarithms 
\begin{equation}
\label{eq:form-in-log-def}
	\Lambda(v,e)\coloneqq e_1 \log v_1 + e_2\log v_2 +\dots +e_n\log v_n.
\end{equation}
How small can $|\Lambda(v,e)|$ be, provided it is nonzero?

In the introduction, we formulated \cref{hyp:linearforms} and noted that it follows
from famous conjectures in number theory, such as
Waldschmidt's conjectures~\cite[Conjecture~14.25]{waldschmidt-2},~\cite[Conjecture~4.14]{waldschmidt-1}
or the Lang-Waldschmidt conjecture
for Gaussian rationals~\cite[Introduction to Chapters X and XI]{lang-elliptic-curves}.

We also note that there are analogues of the $abc$-conjecture in different number fields~\cite{bombieri_gubler_2006}.
Recall that $\ZZ[\ii]$ is a unique factorization domain,
so for an element $x\in\ZZ[\ii]$
there exist distinct prime elements $p_1,p_2,\dots,p_k$,
exponents $\mu_1,\mu_2,\dots,\mu_k\in\ZZ_{>0}$,
and a unit $u\in\ZZ[\ii]^\times$ such that
$x = u \prod_{i=1}^k p_i^{\mu_i}$.
So we can define the radical of $x$ as $\rad(x)\coloneqq \prod_{i=1}^m p_i$.
We may conjecture the Gaussian integer analogue of Baker's $abc$-conjecture\footnote{We have not seen this conjecture in the literature.}:
There exist constants $\kappa,\kappa'>0$ such that for all nonzero
$a,b,c\in\ZZ[\ii]$ with $a+b+c=0$ and $\gcd(a,b,c)=1$,
the radical $N\coloneqq\rad(abc)$ satisfies
\begin{equation}\label{eq:gaussian-baker-abc}
	\max (|a|, |b|, |c|) \leq \kappa |N| \, \left(\frac{\log |N|}{\omega(ab)}\right)^{\kappa' \omega(ab)}.
\end{equation}
With essentially the same arguments as in the proof of \cref{thm:abc-bound}
and \cref{cor:abc-bound-log},
one shows that \cref{eq:gaussian-baker-abc}
implies \cref{hyp:linearforms}.

\Cref{hyp:linearforms} lower bounds the Euclidean distance of $\Lambda(v,e)$ to $0$.
For our purposes, we need to lower bound the distance
$\Delta( \Lambda(v,e) , 0 )\coloneqq \dist( \Lambda(v,e), 2\pi\ii\ZZ )$
in the $\Delta$-metric, see \cref{eq:delta}.
This follows easily.

\begin{proposition}
\label{cor:hyp-in-delta-metric}
Assume \cref{hyp:linearforms}.
Suppose $v_1,\dots,v_n\in\QQ(\ii)^\times$
and $e_1,\dots,e_n\in\ZZ$.
If $\Delta( \Lambda(v,e),0)$ is nonzero, then
$\Delta( \Lambda(v,e) , 0 ) \geq \exp(-\poly(n,B,b))$,
where $B$ (resp.~$b$) is a bound for the bit-length of $e_i$ (resp.~$v_i$).
\end{proposition}

\begin{proof}
We have $\dist(\Lambda(v,e),2\pi\ii\ZZ) = |\Lambda(v,e)- d \ii\pi|$
for some $d\in\ZZ$ satisfying $|d| \leq O(n\, 2^B \, b)$.
Note that
$\Lambda(-1,-d) = -d \Log(-1) = -d\ii\pi$.
Therefore
$\Lambda(v,e)- d\ii\pi = \Lambda(\tilde{v},\tilde{e})$,
where $\tilde{v}$ is obtained from $v$ by appending $-1$
and $\tilde{e}$ is obtained from $e$ by appending $-d$.
Now we apply \cref{hyp:linearforms}.
\end{proof}

\subsection{Equivalence of Separation Hypotheses with the Number-Theoretic Hypothesis~\ref{hyp:linearforms}}
\label{sec:proof-of-abc-and-sep}
We prove here \cref{thm:abc-and-sep} based on \cref{cor:delta-ld-linear-forms},
which relates the logarithmic distance between the orbits $\OO_v$ and $\OO_w$
to $\Delta(H\Log v,H\Log w)$.
The essential observation is that $\Delta(H\Log v,H\Log w)$
is determined by linear forms in logarithms.

\begin{proof}[Proof of \cref{thm:abc-and-sep}]
We will only prove the first equivalence; the second is shown in the same way,
but using \cref{cor:K-orbits-approximation}
instead of \cref{cor:delta-ld-linear-forms}.

($\Rightarrow$)
We first assume \cref{hyp:linearforms} and deduce \cref{hyp:sep}.
Suppose $v,w\in(\QQ(\ii)^\times)^n$ have bit-lengths bounded by $b$
and that the bit-length of the weight matrix~$M$ is bounded by $B$.
By  \cref{prop:allaboutH}, there is a matrix $H$ of rational invariants
of bit-length $\poly(d,n,B)$.
We denote by $\alpha_1^T,\dots,\alpha_k^T$ the rows of $H\in\ZZ^{k\times n}$.

Assume that $\distLD(\OO_v,\OO_w)\neq 0$.
By \cref{cor:delta-ld-linear-forms},
$\Delta(H\Log v,H\Log w)=\Delta(H(\Log v-\Log w), 0)$ is nonzero.
Hence at least one of the one-dimensional distances
$\Delta(\alpha_j^T (\Log v-\Log w),0)$ 
must be nonzero.
By \cref{cor:hyp-in-delta-metric} it
is lower bounded by $\exp(-\poly(d,n,B,b))$.
\Cref{cor:delta-ld-linear-forms} also implies
\begin{equation}\label{eq:form-lower-bound}
	\distLD(\OO_v,\OO_w) \geq \frac{1}{\sigma_{\max}(H)} \Delta(H\Log v, H\Log w)
    \ge \frac{1}{\sigma_{\max}(H)} \Delta(\alpha_j^T (\Log v-\Log w),0) .
\end{equation}
Moreover, by \cref{lem:singular},
$\sigma^{-1}_{\max}(H)$ is lower bounded by $\exp(-\poly(d,n,B))$.
This
finishes the proof of the first implication.

($\Leftarrow$) Now we assume \cref{hyp:sep} and deduce \cref{hyp:linearforms}.
Suppose $v_i\in(\QQ(\ii)^\times)^n$ have bit-length bounded by~$b$,
and $e_i\in\ZZ$ have bit-lengths bounded by $B$ such that $\Lambda(v,e)\neq 0$.
We assume without loss of generality that $\gcd(e_1,e_2,\dots,e_n)=1$.
Put $e\coloneqq(e_1,e_2,\dots,e_n)\in\ZZ^n$ and consider
the rational hyperplane $e^{\perp}$ of $\RR^n$.
Note $\Lambda(v,e)\neq 0$ expresses that $\Log v \not\in e^{\perp}$.
There is a lattice basis $\alpha_1,\alpha_2,\dots,\alpha_{n-1}\in\ZZ^n$ of $e^{\perp}$
having bit-length $\poly(n,B)$;
see \cref{thm:basis-and-equality}~(\ref{it:basis 1}).
Let $M\in\ZZ^{d\times n}$ denote the matrix with rows $\alpha_i$,
where $d\coloneqq n-1$.

We let $(\CC^\times)^d$ act on $\CC^n$ via the weight matrix $M$.
The vector $e$ generates the lattice of invariant Laurent monomials
(recall $\gcd(e_1,e_2,\dots,e_n)=1$).
In this case the matrix $H\in\ZZ^{1\times n}$ of rational invariants is just $H=e^T$.

Consider $v\coloneqq (v_1,v_2,\dots,v_n)$ and
$w\coloneqq (1,\ldots,1) \in\CC^n$, so $\Log w =0$.
Since $H$ has only one row, we have $\sigma_{\max}(H)=\sigma_{\min}(H)=\Vert e\Vert$,
which implies that the inequality from \cref{cor:delta-ld-linear-forms}
is actually an equality:
\begin{equation}\label{eq:Delta-delta}
  \Delta(e^T\Log v, 0) = \Vert e\Vert\,\distLD(\OO_v,\OO_w).
\end{equation}
If $\Delta(e^T\Log v, 0) \ne 0$, then by
\cref{hyp:sep},
$\distLD(\OO_v,\OO_w) \ge \exp(-\poly(n,b,B))$.
Hence
$|\Lambda(v,e)| \geq \Delta(e^T\Log v,0)\ge \exp(-\poly(n,b,B))$.

When $\Delta(e^T\Log v, 0) =0$, we have
$0 \ne \Lambda(v,e) = e^T\Log v \in 2\pi\ZZ$ so $|\Lambda(v,e)| \ge 2\pi$.
\end{proof}

\begin{remark}
Matveev's lower bound~\eqref{eq:BW:93}
combined with inequality \cref{eq:form-lower-bound}
shows that unconditionally
\[
   \sep_T(d,n,B,b)\geq \exp( - B^{O(1)}\, b^{O(n)}) .
\]
The same lower bound holds for $\sep_K(d,n,B,b)$.
\end{remark}

\section{Approximations of orbit distances}\label{sec:main-approx}

In this section, we provide polynomial time approximation algorithms
for the orbit distance approximation problems.
We solve \cref{prob:k-approx-dist} in polynomial time and prove
\cref{thm:intro-approximation-algorithm}.
The main theorem of this section is \cref{thm:gap-algorithms-main}.

\subsection{Approximate orbit distance problems}\label{se:ADP}
So far, we discussed four {\em group actions}:
the actions of $T=(\CC^\times)^d$ or $K=(S^1)^d$ on the spaces
$V=\CC^n$ and $\CC^n/2\pi\ii\ZZ^n$ given by a weight matrix $M\in\ZZ^{d\times n}$.
We also defined a {\em metric} $\delta$ on the resulting spaces of orbits.
Let us list all these cases:
\[
\begin{split}
&(T,\distLD) \qquad \qquad \distLD(\OO_v,\OO_w) \coloneqq \Delta(\Log(\OO_v),\Log(\OO_w))\\
  \vspace{1.5mm}\\
  &(K,\dist) \qquad\qquad \dist(\KK_v,\KK_w) \coloneqq \inf_{k\in K}\Vert k\cdot v-w\Vert\\
 \vspace{1.5mm}\\
  &(T,\Delta) \qquad\qquad \Delta( T\ast \eta, T\ast\zeta ) \coloneqq \inf_{t\in T}\Delta(t\ast\eta, \zeta)\\
 \vspace{1.5mm}\\
 &(K,\Delta) \qquad\qquad \Delta(K\ast\eta,K\ast\zeta)\coloneqq \inf_{k\in K}\Delta(k\ast\eta, \zeta),
\end{split}
\]
see \cref{def:ast} for the $\ast$-action and \cref{eq:delta} for the $\Delta$-metric.
We define an orbit distance approximation problem in each setting.

\begin{definition}
\label{def:gaprop}
In each of above case, abbreviated $(G,\delta)$,
the \textit{orbit distance approximation problem} $\textsc{ROP}(G,\delta)_\gamma$
with approximation factor~$\gamma\geq1$ is the following:
on input a weight matrix $M\in\ZZ^{d\times n}$ and
$v,w \in (\QQ(\ii)^\times)^n$, resp.\ $\eta,\zeta\in\QQ^n+2\pi\ii\QQ^n$,
compute a number $D\in\QQ_{\geq 0}$ such that
\[
  \delta( G\cdot v, G\cdot w ) \, \leq \, D \, \leq  \, \gamma \, \delta(G\cdot v, G\cdot w).
\]
We think of $\gamma$ as a function of the input bit-length, determined by the usual
parameters $d,n,B,b$.
\end{definition}

We can now precisely state the main algorithmic result of this section.
Recall from \cref{eq:DrR} the regions $D_{r, R}$ defined for $0<r\leq R$.
Note that \cref{thm:intro-approximation-algorithm} follows from
the second and third part of the next result.

\begin{theorem}
\label{thm:gap-algorithms-main}
There are approximation factors $\gamma_i=\gamma_i(d,n,B)$
bounded $\exp(\poly(d,n,B))$ such that:
\begin{alphaenumerate}
\item\label{it:gap 3}
$\textsc{ROP}(G,\delta)_{\gamma_1}$ admits a polynomial time algorithm
in the cases $(T, \Delta)$ and $(K, \Delta)$.

\item\label{it:gap 1}
$\textsc{ROP}(T,\distLD)_{\gamma_2}$
admits a polynomial time algorithm,
if \cref{hyp:sep} holds.

\item\label{it:gap 2}
$\textsc{ROP}(K,\dist)_{\gamma}$, when restricted to inputs in $D_{r,R}$,
admits a polynomial time algorithm with
approximation factor $\gamma= \frac{R}{r}\gamma_3(d,n,B)$,
if \cref{hyp:compact-separation} holds.

\end{alphaenumerate}
\end{theorem}

One may wonder why the first part of \cref{thm:gap-algorithms-main}
holds unconditionally. This is explained by the next result, which
gives an unconditional separation of orbits with respect to
the $\Delta$-metric.

\begin{proposition}
\label{prop:logarithmic-separation}
If $\eta,\zeta\in\QQ^n+2\pi\ii\QQ^n$ with
$T\ast\eta\neq T\ast\zeta$, then
\[
 \Delta(T\ast\eta, T\ast\zeta) \ge \exp(-\poly(d,n,b,B)) .
\]
The analogous statement holds for the action of $K$.
\end{proposition}

\begin{proof}
We write $\eta=\rho+2\pi\ii\theta,\; \zeta=\tau+2\pi\ii\phi$
and recall from  \cref{prop:delta-orbit-distance} that
\[
 \Delta^2(T\ast\eta,T\ast\zeta) =
 \dist^2( \rho-\tau, U) + 4\pi^2 \Delta^2(K\ast \ii\theta, K\ast \ii\phi) .
\]
The first contribution $\dist^2( \rho-\tau, U)$,
if nonzero, is easily lower bounded 
using the Gram-Schmidt orthogonalization.
So we may assume $K\ast \ii\theta\neq K\ast \ii\phi$.
In this case, \cref{cor:delta-ld-linear-forms}
implies
$\dist(H(\theta - \phi),\, \ZZ^k)\ne 0$ and
\[
  \Delta(K\ast\ii\theta,K\ast\ii\phi)
     \geq \frac{1}{\sigma_{\max}(H)}\, \dist(H(\theta - \phi),\, \ZZ^k) ,
\]
where $H$ is a matrix of rational invariants whose bit-length is bounded by \cref{prop:allaboutH}.
We can upper bound $\sigma_{\max}(H)$
by \cref{prop:allaboutH} and \cref{lem:singular},
which provides the desired lower bound for the $T$-action.
The distance of $H(\theta-\phi)$ to $\ZZ^k$ is also at most exponentially small since it is lower bounded by the reciprocal of the largest denominator of the entries of the rational vector $H(\theta-\phi)$.
The proof for the $K$-action is analogous.
\end{proof}

We also show that our algorithms for approximating the distance between orbits
can be modified to produce a group element witnessing the $\gamma$-proximity of orbits.
To avoid repetition, we only state this for the case $(T,\distLD)$, even though
this can also be achieved in the other settings in a similar way.

\begin{theorem}\label{thm:rop-witness}
Assume \cref{hyp:sep}.
On input $M,v,w$,
one can compute in polynomial time a vector $x\in\CC^d$ that satisfies
\begin{equation}
\label{eq:rop-witness}
  \distLD( e^x\cdot v, w ) \, \leq \, \gamma \, \distLD(\OO_v,\OO_w)
\end{equation}
if $\OO_v\neq \OO_w$. If $\OO_v=\OO_w$,
then the algorithm correctly identifies this case instead of returning~$x$.
The approximation parameter is exponentially bounded in the
bit-length of the input.
\end{theorem}

In the next section, we introduce the distance computation problem \textsc{SLDP}
and exhibit a polynomial time algorithm for it.
\Cref{thm:gap-algorithms-main,thm:rop-witness} are then proved in \cref{se:pf-gap-algs} by
reducing the  orbit distance approximation problems to \textsc{SLDP}.

\subsection{The subspace-to-cubic-lattice distance problem}\label{sec:sldp-and-rop}
The orbit distance approximation problems are closely related to the following problem.

\begin{definition}
\label{def:gapsldp}
The {\em subspace-to-cubic-lattice distance approximation problem} $\textsc{SLDP}_\gamma$
with approximation factor~$\gamma$
is the task of computing on input
\begin{itemize}
\item a target vector $t\in\QQ^n$,
\item a subspace $U\subset\RR^n$, spanned by given linearly independent input vectors
   $u_1,\dots,u_{n-k}\in\ZZ^n$,
\end{itemize}
a number $D\in\QQ_{\geq 0}$ such that
\[
 \dist(t+U,\ZZ^n)\, \leq \, D \, \leq \, \gamma \, \dist(t+U,\ZZ^n).
\]
We think of $\gamma$ a function of the input bit-length.
\end{definition}

\begin{remark}
The problem $\textsc{SLDP}_\gamma$ is easy if $U=0$.
Indeed, $\dist^2(t,\ZZ^n)$ can be exactly and efficiently computed by rounding
the coordinates of $t$ to integers.
At the other extreme, the problem is trivial if $U=\RR^n$.
However, we will show in \cref{sec:hardness}, that $\textsc{SLDP}_\gamma$ is
\textsc{NP}-hard for constant $\gamma$ (or more generally, an almost polynomial~$\gamma$),
by providing a reduction from the \textit{closest vector problem} to it.
\end{remark}

\begin{proposition}\label{thm:sldp-is-in-p}
There is an approximation factor~$\gamma$ bounded exponentially in the input bit-length such that $\textsc{SLDP}_\gamma$ admits a polynomial time algorithm.
This algorithm can be modified to compute
on input $(t,U)$ a witness of proximity
$(u,\alpha)\in U\times\ZZ^n$ such that
\begin{equation*}
  \dist(t+u,\alpha) \, \leq \, \gamma\, \dist(t+U,\ZZ^n).
\end{equation*}
\end{proposition}

\begin{proof}
For given $M\in\ZZ^{(n-k)\times n}$ we compute
$H\in\ZZ^{k\times n}$ such that $U : =\ran M^T =\ker H$
in polynomial time,
see \cref{prop:allaboutH}.
When applying \cref{lem:sing-value-inequality}
to $A =\{t\}$ and $B =\ZZ^n$, we get
\begin{equation}
\label{eq:nice-ineq}
\dist(t+U,\ZZ^n)\, \leq\, D_1 := \frac{\dist(Ht, \ZZ^k)}{\sigma_{\min}(H)}\,
 \leq\, \frac{\sigma_{\max}(H)}{\sigma_{\min}(H)} \, \dist(t+U,\ZZ^n) .
\end{equation}
Hence, $D_1$
approximates $\dist(t+U,\ZZ^n)$ within a factor of
$\sigma_{\max}(H)/\sigma_{\min}(H)$.
By \cref{lem:singular} we can compute a rational approximation $D$ of $D_1$
within a factor of~$2$ in polynomial time.
Moreover $\sigma_{\max}(H)/\sigma_{\min}(H)$
is exponentially upper bounded in the input bit-length.

We show now how to compute the witness.
By rounding the entries of $Ht$ to nearest integers, we can compute an integral vector
$\beta\in\ZZ^k$ such that $\dist(Ht,\ZZ^k)=\dist(Ht,\beta)$.
Recall that $H(\ZZ^n) = \ZZ^k$.
We can further efficiently compute $\alpha\in\ZZ^n$ such that $H\alpha=\beta$.
By projecting $t-\alpha$ orthogonally onto $U$,
we can compute in polynomial time $u\in U$ such that
$\dist(t+u,\alpha)=\dist(t+U,\alpha)$.
\end{proof}

\begin{remark}
In \cref{sec:sldp-to-cvp}, we analyse an alternative algorithm
(based on the LLL-algorithm),
which solves $\textsc{SLDP}_\gamma$ for $\gamma =2^{O(n)}$
in polynomial time (\cref{cor:sldp-in-p-for-two-to-n}). In comparison, the algorithm given in Proposition~\ref{thm:sldp-is-in-p} provides an approximation factor of $\gamma=\kappa(H)\coloneqq\sigma_{\max}(H)/\sigma_{\min}(H)$ which is bounded by $2^{n(B+\log_2 n)}$ in the worst case.
\end{remark}

\subsection{Proof of Theorem~\ref{thm:gap-algorithms-main}}
\label{se:pf-gap-algs}
The first part of this theorem follows from
the following reduction, jointly with \cref{thm:sldp-is-in-p}.

\begin{lemma}\label{le:ROP-SLP}
Both $\textsc{ROP}(T,\Delta)_{2\gamma}$ and
$\textsc{ROP}(K,\Delta)_{2\gamma}$
reduce to $\textsc{SLDP}_{\gamma}$ in polynomial time,
for any $\gamma \ge 1$.
Hereby the ambient dimension does not change.
\end{lemma}

\begin{proof}
We will only prove the reduction for $G=K$ since the $G=T$ case is analogous.
Suppose that $M,\eta,\zeta$ are given as input, write
$\eta=\rho+2\pi\ii\theta, \zeta=\tau+2\pi\ii\phi$,
set $t\coloneqq \theta-\phi$, and
$U\coloneqq \ran M^T$.
We can compute an integral basis of $U$ in polynomial time.
 \cref{eq:Delta=dist} and \cref{prop:Kdelta-orbit-distance} imply
\[
\Delta^2(K\ast \eta,K\ast\zeta)
   = \Vert\, \rho-\tau\,\Vert^2 \, + 4\pi^2 \, \dist^2(\,t + U, \, \ZZ^n) .
\]
If $\dist(t+U,\ZZ^n)\leq D\leq \gamma \dist(t+U,\ZZ^n)$, then
\[
 \Delta^2(K\ast\eta,K\ast\zeta)\leq \Vert \rho-\tau\Vert^2 + 4\pi^2 D^2
   \leq \gamma^2 \Vert\rho-\tau\Vert^2 + 4\pi^2 D^2 \leq \gamma^2 \Delta^2(K\ast\eta,K\ast\zeta),
\]
hence $D'\coloneqq \sqrt{\Vert \rho-\tau\Vert^2 + 4\pi^2 D^2}$ is a $\gamma$-approximate solution of
$\textsc{ROP}(K,\Delta)$.
We compute a rational number $D''$ such that $D'\leq D''\leq 2 D'$. Then
$\Delta(K\ast\eta,K\ast\zeta) \leq D''\leq 2 \gamma \, \Delta(K\ast\eta,K\ast\zeta)$,
\end{proof}

A reduction in the reverse direction will be needed later for the hardness proof,
see \cref{thm:sldp-is-hard-and-reductions}.
The next lemma studies that happens with \textsc{ROP}, when the
metric $\distLD$, resp.~$\dist$, are replaced by $\Delta$.
This lemma proves the second and third part of \cref{thm:gap-algorithms-main}
by reducing it to its first part.

\begin{lemma}
\label{prop:delta-to-Delta-reduction}
Let $\gamma \ge 1$ be any function of the inputs, let $0 <r \le R$
and put $\gamma_1\coloneqq \frac{2R}{r}\,\gamma$.
\begin{alphaenumerate}
\item\label{it:dD 1}
There is a polynomial time reduction from
$\textsc{ROP}(T,\distLD)_{2\gamma}$ to $\textsc{ROP}(T,\Delta)_{\gamma}$,
provided  \cref{hyp:sep} holds.

\item\label{it:dD 2}
There is a polynomial time reduction from
$\textsc{ROP}(K,\dist)_{\gamma_1}$, when restricted to inputs in $D_{r,R}$,
to $\textsc{ROP}(K,\Delta)_{\gamma}$,
provided \cref{hyp:compact-separation} holds.
\end{alphaenumerate}
\end{lemma}

\begin{proof}
We only give the proof for part two since the proof of part one
is analogous. The only difference between the two cases is that in the second we need to use \cref{thm:logorbits-equivalence} which gives an equivalence between the metrics $\dist$ and $\distLD$ within an $O(R/r)$-factor, and in the first we do not need it. 

Given $v,w\in D_{r,R}$, we test in polynomial time whether
$\KK_v=\KK_w$ (see \cref{se:invar-th-torus}). 
If~$\KK_v=\KK_w$ we return~$D=0$.\footnote{To get a Karp reduction, we can simply return the $\textsc{ROP}(K,\Delta)$ instance $(\eta,\zeta)=(0,0)$.}
Otherwise, the reduction just consists of replacing the exact values
$\eta:=\Log v$ and $\zeta:=\Log w$
by approximations
\begin{equation}\label{eq:eta-zeta-approx}
  \Vert \tilde{\eta} - \eta\Vert< \kappa/2, \quad \Vert \tilde{\zeta} -\zeta\Vert < \kappa/2 ,
\end{equation}
where the accuracy is prescribed by the separation parameter.
For $\kappa\coloneqq \frac{1}{9R}\sep_K(d,n,B,b)$,
this can be achieved in polynomial time
assuming \cref{hyp:compact-separation}.
For convenience, we denote the orbit distances of the exact vectors
and the approximations computed by
$$
 \Delta:= \Delta(K\ast\eta, K\ast\zeta), \quad
 \tilde{\Delta} := \Delta(K\ast\tilde{\eta}, K\ast\tilde{\zeta}) .
$$
It suffices to prove that
$\tilde{\Delta} \leq D\leq \gamma \tilde{\Delta}$
implies
\[
  \dist(\KK_v,\KK_w) \, \leq \, \frac{9RD}8 \, \leq \, \gamma_1\, \dist(\KK_v,\KK_w),
\]
where $\gamma_1= 2\,R\,\gamma/r$.
We now prove this implication.

By \cref{thm:logorbits-equivalence},
we have 
\begin{equation}\label{eq:logorbits-equiv}
  \frac{2r}{\pi} \, \Delta \leq \dist(\KK_v,\KK_w)
    \leq R \, \Delta ,
\end{equation}
which implies $\Delta \ge 9 \kappa$ by our choice of $\kappa$.
With \cref{eq:eta-zeta-approx} and
the triangle inequality, we get
\begin{equation*} 
	| \tilde{\Delta} - \Delta | < \kappa, \quad \frac{\kappa}{\Delta} \le \frac{1}{9} .
\end{equation*}
By dividing through $\Delta$,
this implies
$\frac{8}{9} <\frac{\tilde{\Delta}}{\Delta} < \frac{10}{9}$.
From the assumption
$\tilde{\Delta} \leq D\leq \gamma \tilde{\Delta}$,
we infer
\[
  \Delta \, \leq \, \frac{9}{8} \, D \, \leq \, \frac{10}{8}\, \gamma \, \Delta .
\]
\cref{eq:logorbits-equiv} implies
\[
 \dist(\KK_v,\KK_w)\,  \leq \, \frac{ 9\, R  \, D}{8} \, \leq \,
 \frac{10\, \pi \, R \, \gamma}{16 \, r} \,\dist(\KK_v,\KK_w) \, <\, \frac{2 R\, \gamma}{r} \dist(\KK_v,\KK_w),
\]
which completes the proof.
\end{proof}

We can prove \cref{thm:rop-witness} in a similar way.

\begin{proof}[Proof of \cref{thm:rop-witness}]
On input $v,w\in (\CC^\times)^n$,
we first test in polynomial time whether $\OO_v=\OO_w$
(see \cref{se:invar-th-torus}). 
If this is the case, we return $D=0$.
Otherwise, as in \cref{eq:eta-zeta-approx},
we compute approximations
$\tilde{\eta}$ and $\tilde{\zeta}$ of the exact values
$\eta:=\Log v$ and $\zeta:=\Log w$
with accuracy $\kappa=\sep_T(d,n,B,b)/2$:
\cref{hyp:sep}
guarantees that this can be done in
polynomial time.
We have for all $x\in\CC^d$
\[
 \Delta(e^x\ast\eta,\zeta)\geq \Delta(T\ast\eta,T\ast\zeta)=\distLD(\OO_v,\OO_w)\geq\sep_T(d,n,B,b)=2\kappa.
\]
Hence, using the triangle inequality, we get  for all $x\in\CC^d$
\begin{equation*}
  \Big| \Delta(e^x\ast \tilde{\eta}, \tilde{\zeta})  - \Delta(e^x\ast \eta, \zeta) \Big| < \kappa \leq \frac12\Delta(e^x\ast \eta, \zeta) ,
\end{equation*}
therefore
\begin{equation}\label{eq:required-ineqs}
 \frac{1}{2}\, \Delta(e^x\ast \eta, \zeta)  \ <\  \Delta(e^x\ast\tilde{\eta},\tilde{\zeta})\   <\  \frac{3}{2}\, \Delta(e^x\ast \eta, \zeta) .
\end{equation}
Write $\tilde{\eta}=\tilde{\rho}+2\pi\ii\tilde{\theta}$ and $\tilde{\zeta}=\tilde{\tau}+2\pi\ii\tilde{\phi}$ and
recall from \cref{prop:delta-orbit-distance} that
$$
 \Delta^2(T\ast\tilde{\eta},T\ast\tilde{\zeta})
      =\dist^2(\tilde{\rho} -\tilde{\tau}, U)
     + 4\pi^2\dist^2 (\tilde{\theta} -\tilde{\phi} + U,\ZZ^n) .
$$
For the first summand,
we compute the orthogonal projection $u_1$ of $\tilde{\rho} -\tilde{\tau}$ onto $U$.
For the second summand, we use \cref{thm:sldp-is-in-p}
to compute a witness $(u_2,\alpha)$ such that
\begin{equation}\label{eq:dada}
  \dist(\tilde{\theta} -\tilde{\phi} + u_2, \alpha) \,
   \leq \, \gamma\, \dist(\tilde{\theta} -\tilde{\phi}+ U, \ZZ^n)
\end{equation}
for some $\gamma$ satisfying $\gamma=\exp(\poly(B,n))$.
We then compute $y,z\in\RR^d$ such that $M^Ty=u_1$ and $M^Tz=u_2$
and compute the group element $x\coloneqq y + 2\pi\ii z$.
Then we have by \cref{eq:delta}
\[
\begin{split}
 \Delta^2(e^x\ast\tilde{\eta},\tilde{\zeta}) &= \dist^2(\tilde{\rho} -\tilde{\tau},M^Ty)
  + 4\pi^2 \dist^2( \tilde{\theta} -\tilde{\phi} + u_2,\alpha)\\
   &\leq \dist^2(\tilde{\rho} -\tilde{\tau}, U)
    + 4\pi^2\, \gamma^2 \, \dist^2( \tilde\theta -\tilde\phi + U, \ZZ^n)\\
   &\leq \gamma^2 \Delta^2(T\ast\tilde{\eta},T\ast\tilde{\zeta}) ,
\end{split}
\]
where we used \cref{eq:dada} for the second inequality.
Combining with \cref{eq:required-ineqs}, we obtain
\[
 \Delta(e^x\ast\eta,\zeta) \, \leq \, 2 \Delta(e^x\ast\tilde{\eta},\tilde{\zeta})
  \, \leq \, 2\gamma \Delta(T\ast\eta,T\ast\zeta) ,
\]
which finishes the proof.
\end{proof}

\begin{remark}\label{rem:poly-time-iff-sep-hyp}
The algorithm underlying the reduction
in the proof of \cref{prop:delta-to-Delta-reduction}
runs in time polynomial in the input bit-length and $\log\sep_T^{-1}(d,n,B,b)$.
Thus it is a polynomial time algorithm if and only if the separation hypothesis is true.
\end{remark}

\subsection{Reductions from \texorpdfstring{$\textsc{SLDP}$}{SLDP} to \texorpdfstring{$\textsc{ROP}$}{ROP}}
The following reductions are needed in the next section.
While the proofs are straightforward perturbation arguments,
dealing with the relative errors requires some care,
so that we write down the detailed arguments at least in one case.

\begin{lemma}\label{thm:sldp-is-hard-and-reductions}
Let $(G,\delta)$ denote any of the four cases in \cref{def:gaprop}.
Then there is a polynomial-time reduction from
$\textsc{SLDP}_{2\gamma}$ to $\textsc{ROP}(G,\delta)_\gamma$,
for any $\gamma \ge 1$.
Hereby the ambient dimension is preserved.
\end{lemma}

\begin{proof}
We only provide the argument for $(G,\delta)=(K,\dist)$, the case $(T,\distLD)$ being similar
and the remaining two cases being trivial.
Consider an instance $U,t$ of \textsc{SLDP}, where $U\subset\RR^n$ is a subspace
of dimension $d:=n-k$ given as the row span of $M\in\ZZ^{(n-k)\times n}$, and $t\in\QQ^n$.
The matrix $M$ defines actions of $T=(\CC^\times)^d$ and $K=(S^1)^d$ on $\CC^n$.
The essential connection is \cref{eq:Delta=dist},
which implies that
\begin{equation}\label{eq:d=Delta}
  \dist(t+U,\ZZ^n) = \Delta( K\ast \ii t , K\ast 0 ) = \Delta(T\ast \ii t, T\ast 0).
\end{equation}
By \cref{prop:logarithmic-separation},
there is a separation function $0<\epsilon \le 1$ such that
$\log\epsilon^{-1}$ polynomially bounded in the parameters $d,n,B,b$ and
$\Delta(K\ast 2\pi\ii t, K\ast 0)\geq \epsilon$ when the orbits are different.

By \cref{lem:exp} we compute
$v\in(\QQ(\ii)^\times)^n$ such that
$ \Vert v-\Exp(2\pi\ii t)\Vert< \kappa\coloneqq\epsilon/18$.
Hence $v\in D_{r,R}$ for $r=17/18$ and $R=19/18$.
\Cref{le:Lipsch-log} below expresses a Lipschitz property of $\Log$
and gives
$$
 \Vert \Log v -2\pi\ii t\Vert_\infty < \frac{\kappa}{1-\kappa} < \frac{18}{17}\, \kappa .
$$
Thus the distances
$\Delta_1 := \Delta(K\ast\Log v, 0)$ and
$\Delta_2 :=\Delta(K\ast 2\pi\ii t, 0)$
of the corresponding $K$-orbits to the zero orbit are close:
$|\Delta_1 - \Delta_2| <  \frac{18}{17}\, \kappa$.
Using $\Delta_2 \ge \epsilon = 18\kappa$,
this bounds the relative errors as
$\frac{|\Delta_1 - \Delta_2|}{\Delta_2} \le \frac{1}{17}$,
and hence
$$
 \Delta_2 \ \le\  \frac{17}{16}\Delta_1 \ \le\  \frac{18}{17} \Delta_2 .
$$
Note that
$\dist(t+U,\ZZ^n) = \frac{1}{2\pi} \Delta_2$
by \cref{eq:d=Delta}.
Consider the vector $w\coloneqq (1,\ldots,1)$.
\Cref{thm:logorbits-equivalence} implies,
using $v\in D_{r,R}$ with $r=17/18$ and $R=19/18$, that
\begin{equation}\label{eq:ineq}
 \frac{2}{\pi}\, \frac{17}{18}\, \Delta_1 \, \leq\,\dist(\KK_v,\KK_w)\, \leq \, \frac{19}{18} \; \Delta_1 .
\end{equation}
Now assume that
$\dist(\KK_v,\KK_w) \leq \, D\, \leq \, \gamma\, \dist(\KK_v,\KK_w)$.
Then, 
$$
 \Delta_2 \ \le\ \frac{17}{16}\Delta_1 \ \le\ \frac{17}{16} \frac{\pi}{2} \frac{18}{17} \, D
  \ \le\
 \frac{17}{16} \frac{\pi}{2} \frac{18}{17}\; \gamma\; \frac{19}{18} \frac{18}{17} \Delta_2
  \ \le\  2\; \frac{\Delta_2}{2\pi} .
$$
This means that
$$
   \dist(t+U,\ZZ^n) \ \le\  \frac{9}{32}\, D \ \le\ 2 \gamma \ \dist(t+U,\ZZ^n) ,
$$
which completes the proof of the reduction.
\end{proof}


\begin{lemma}\label{le:Lipsch-log}
For any $t\in\RR^n$ and $z\in\CC^n$,
$\|z- \Exp(2\pi\ii t)\| < \kappa<1$ implies that
$\|\Log(z)- 2\pi\ii t \| < \kappa/(1-\kappa)$.
\end{lemma}

\begin{proof}
It is sufficent to verify the claim for $n=1$, in which case the claim follows using $|\Log' z| = |\frac{1}{z}| \le \frac{1}{1-\kappa}$.
\end{proof}

\section{Hardness of robust orbit problems} \label{sec:hardness}

The main goal of this section is to prove \cref{thm:intro-orbit-hardness}.
Actually, we prove the following, slightly more general result,
that covers all four settings introduced in \cref{se:ADP}.

\begin{theorem}
\label{thm:rop-is-hard}
There is $c>0$ such that the robust orbit distance approximation problem
$\textsc{ROP}(G,\delta)_{\gamma}$ is \textsc{NP}-hard for the approximation
factor $\gamma(n)= n^{c/\log\log n}$,
for each of the four combinations of group actions and metrics
considered in \cref{def:gaprop}.
\end{theorem}

Of course, this implies the hardness of the decisional version of the considered orbit distance approximation problems, namely:

\begin{corollary}
There is no polynomial time algorithm that on input $M,v,w$ and $\varepsilon > 0$ decides
\[
	\dist(\KK_v,\KK_w)\leq \varepsilon,
\]
unless $\textsc{P}=\textsc{NP}$. The analogous result also holds
for the logarithmic distance
$\distLD(\OO_v,\OO_w)$ of the orbits of the $T$-action.
\end{corollary}

\begin{proof}
By binary search, making oracle calls to a hypothetical polynomial time algorithm for the above
decision problem, we can compute an approximation to $\dist(\KK_v,\KK_w)$ within any constant factor, say two.
This provides a polynomial time algorithm for
the robust orbit distance approximation problem
with approximation factor $\gamma=2$.
By \cref{thm:rop-is-hard}, this implies $\textsc{P}=\textsc{NP}$.
\end{proof}

The proof of \cref{thm:rop-is-hard} relies on the \textsc{NP}-hardness of
the closest vector problem $\textsc{CVP}_\gamma$
of algorithmic lattice theory due to Dinur et al.~\cite{dinurcvp}.
We discuss the closest vector problem in \cref{sec:sldp-to-cvp}.
The proof of \cref{thm:rop-is-hard} then proceeds by
exhibiting a polynomial time reduction from
$\textsc{CVP}_{2\gamma}$ to $\textsc{ROP}(G,\delta)_{\gamma}$.
This goes in two steps: from \cref{thm:sldp-is-hard-and-reductions}
we already know that
$\textsc{SLDP}_{2\gamma}$ can be reduced to $\textsc{ROP}(G,\delta)_\gamma$
in polynomial time.
The main ingredient of the proof is a polynomial time reduction of
$\textsc{CVP}_\gamma$ to $\textsc{SLDP}_\gamma$,
which is presented in \cref{sec:liftings}.
This relies on \cref{thm:intro lattice lifting}
on lattice lifting, whose proof is given in \cref{sec:eutactic}.

\subsection{The closest vector problem}\label{sec:sldp-to-cvp}
This problem attracted a lot of research due to lattice based cryptography, such as GGH, NTRU and
homomorphic encryption~\cite{ggh-crypto, ntru-crytpo-1, ntru-crypto-2, homomorphic-crypto}.
These cryptosystems are conjectured to be secure agains quantum computers and rely on a (conjectured) hardness of \textsc{CVP}.

Throughout the section, $m$ will always denote the dimension of a $\textsc{CVP}$ instance and $n$ will always denote the dimension of a $\textsc{SLDP}$ instance.

\begin{definition}
\label{def:gapcvp}
{\em The closest vector problem} with approximation factor $\gamma\geq1$,
denoted $\textsc{CVP}_\gamma$,
is the task of computing on input
\begin{itemize}
\item a target vector $t\in\QQ^m$,
\item a lattice $\mathcal{L}$ spanned by the columns of generator matrix $G\in\ZZ^{m\times m}$
 with $\det G\neq 0$,
\end{itemize}
a number $D\in\QQ_{\geq 0}$
such that
$\dist(t,\mathcal{L}) \leq  D \leq \gamma \, \dist( t , \mathcal{L} )$.
\end{definition}

The first observation is that $\textsc{SLDP}_\gamma$ can be easily reduced to $\textsc{CVP}_\gamma$.
Indeed, if $P:\RR^n\rightarrow U^\perp$ denotes the orthogonal projection along $U$, then
$\dist(t+U,\ZZ^n)=\dist(P(t),\mathcal{L})$ by \cref{cor:ort-proj-dist}.
Here, $\mathcal{L}\coloneqq P(\ZZ^n)$ is a lattice by \cref{lem:ort-proj-lattice},
and we can compute a lattice basis of $\mathcal{L}$ in polynomial time. 
This reduces $\textsc{SLDP}_\gamma$ to  $\textsc{CVP}_\gamma$.
We note that this reduction preserves~$\gamma$.
Moreover, the ambient dimension~$n$ of the \textsc{SLDP} instance
upper bounds the lattice dimension $m$ of the constructed \textsc{CVP} instance.

The key contribution in the proof of \cref{thm:rop-is-hard}
is a polynomial time reduction in the reverse direction,
see \cref{thm:red-cvp-2-sldp}.

$\textsc{CVP}_\gamma$ is known to be \textsc{NP}-hard for a nearly polynomial approximation ratio,
$\gamma=m^{c/\log\log m}$ for some $c>0$, see~\cite{dinurcvp}.
On the other hand, it is well known that
$\textsc{CVP}_{\gamma}$ can be solved in polynomial time
with approximation factor $\gamma(m) = 2^{O(m)}$
by the well-known LLL-basis reduction algorithm of~\cite{lll}.
This implies the following.

\begin{corollary}\label{cor:sldp-exact}
\label{cor:sldp-in-p-for-two-to-n}
 $\textsc{SLDP}_\gamma$ admits a polynomial time algorithm for the
 approximation factor $\gamma(n) = 2^{O(n)}$.
\end{corollary}

If we do not insist on polynomial time algorithms,
then the lattice element $\alpha\in\mathcal{L}$ closest to the target vector~$t$
can be computed exactly.
Kannan~\cite{kannan1987} showed that there is an algorithm for doing so,
that runs in time $2^{O(m\log m)}$ times the input size.
Combining this with the above reduction, we obtain:

\begin{corollary}\label{cor:sldp-fixed-dimension}
$\textsc{SLDP}_1$ can be solved by a polynomial time algorithm
when $n$ is fixed, if we allow exact computation of square roots.
\end{corollary}

We use this to show polynomial time feasibility of
\cref{prob:k-approx-dist} and \cref{prob:g-approx-dist-LD}
when $n$ is fixed.

\begin{proof}[Proof of \cref{thm:compact-fixed-dimension} and \cref{thm:fixed-dimension}]
Matveev's bound~\eqref{eq:BW:93} implies \cref{hyp:linearforms} if $n$ is fixed.
Then also \cref{hyp:sep} holds if $n$ is fixed, as shown by the proof of \cref{thm:abc-and-sep}.
Note that the reductions in \cref{prop:delta-to-Delta-reduction} preserve the ambient dimension.
The assertion follows now with \cref{cor:sldp-fixed-dimension}.
\end{proof}

\begin{remark}
The interesting regime of $\textsc{CVP}_\gamma$ is when $\gamma$ is polynomial in $m$.
Lattice based cryptosystems rely on the hardness of
$\textsc{CVP}_\gamma$ in this regime, but as of now,
proving this is a major open problem in the area.
In fact it is known that  $\textsc{CVP}_\gamma$
is in $\textsc{NP}\cap\text{co}\textsc{NP}$~\cite{aharonov-regev},
for $\gamma=100\sqrt{m}$,
which implies that $\textsc{CVP}_\gamma$
cannot be \textsc{NP}-hard in this regime,
unless $\textsc{NP}\subset\text{co}\textsc{NP}$ and the polynomial hierarchy collapses.
\end{remark}

\subsection{The reduction from \textsc{CVP} to \textsc{SLDP} via lattice liftings}\label{sec:liftings}
We now establish the key reduction from
\textsc{CVP} to \textsc{SLDP}
by relying on the lattice lifting result \cref{thm:intro lattice lifting}.

\begin{theorem}\label{thm:red-cvp-2-sldp}
There is a polynomial time (Turing) reduction from
$\textsc{CVP}_{\gamma}$ to $\textsc{SLDP}_{\gamma}$
(with the same~$\gamma$).
In more detail, given as input a \textsc{CVP} instance $(t,\mathcal{L})$ with ambient dimension~$m$, the reduction either solves \textsc{CVP} exactly or it computes a scaling factor $s\in\ZZ_{>0}$ and an \textsc{SLDP} instance $(\tilde{t},U)$ of ambient dimension~$n$ where $\dist(t,\mathcal{L})=s\dist(\tilde{t}+U,\ZZ^n)$ and $m\le n\le O(m^2\log m)$.
\end{theorem}

\begin{proof}[Proof of \cref{thm:red-cvp-2-sldp}]
An instance of $\textsc{CVP}_\gamma$
consists of a rank~$m$ lattice $\mathcal{L}$ in $\RR^m$, given by a generator matrix,
and a target vector $t\in\QQ^m$.
We apply \cref{thm:intro lattice lifting}.
Thus, for a given such instance,
we can compute in polynomial time
an orthonormal basis $v_1,v_2,\dots,v_n\in\QQ^n$, where
$n\geq m$, and a scaling factor $s\in\ZZ_{>0}$, such that
the lattice $\mathcal{L}'$ generated by $v_1,v_2,\dots,v_n$
satisfies
$\mathcal{L} = s P(\mathcal{L}')$,
where $P\colon\RR^n\to\RR^m$ denotes the orthogonal projection
onto the first $m$ coordinates. We view $\RR^m$ as the subspace of $\RR^n$ whose last $n-m$ coordinates are zero.

Consider the orthogonal matrix $Q\in\QQ^{n\times n}$ with columns $v_1,v_2,\dots,v_n$,
thus $Q e_i = v_i$ if the $e_i$~denote the standard basis vectors.
This implies that $Q(\ZZ^n)=\mathcal{L}'$.
Let $w_i$ be the rows of $Q$, thus $w_i = Q^T e_i$ and $Qw_i= e_i$.
We define the subspace $U$ as the span of $w_{m+1},\ldots,w_n$,
thus $Q(U)= \langle e_{m+1},\ldots, e_{n} \rangle = \ker P$.
Moreover, we define $t'\coloneqq s^{-1} Q^T t$.
Then $Qt' = s^{-1} t$. Summarizing,
$$
  s P(\mathcal{L}') = \mathcal{L}, \quad
 Q(\ZZ^n)=\mathcal{L}', \quad Q(U) =  \ker P, \quad Qt' = s^{-1} t .
$$
Let us verify that
\begin{equation}\label{eq:dist-equality}
	\dist(t,\mathcal{L}) = s\, \dist(t'+U,\ZZ^n) .
\end{equation}
Indeed,
$\dist(t'+U,\ZZ^n) = \dist(Q(t'+U),Q(\ZZ^n))$
by the orthogonality of $Q$. Moreover,
\begin{equation*}
 \dist(Q(t') + Q(U),Q(\ZZ^n)) = \dist \big( s^{-1} t + \ker P, \mathcal{L}' \big)
 = \dist \big(  s^{-1} P( t),  P(\mathcal{L}') \big)
= s^{-1}  \dist \big( t , \mathcal{L} \big)
\end{equation*}
where we used \cref{cor:ort-proj-dist} in the second and $P(t)=t$ in the last equality.

Clearly, this defines a polynomial time reduction of
$\textsc{CVP}_{\gamma}$ to $\textsc{SLDP}_{\gamma}$
that does not change the approximation factor $\gamma$.

Suppose the given instance of \textsc{CVP} has bit-length $l$.
\Cref{thm:intro lattice lifting} implies that we can assume
$n=O(m (\log m +\log l))$.
We can bound $l$ in terms of $m$ by making
a case distinction.
If $l$ is so large that
$l \ge 2^{\Omega(m\log m)}$,
then we apply the algorithm in~\cite{kannan1987},
which solves \textsc{CVP} exactly, in time
$l2^{O(m\log m)}$ which is polynomially bounded in $l$ in this case.
Otherwise, if $l \le 2^{O(m\log m)}$, we perform the above described
reduction to \textsc{SLDP}. Then we have
$m(\log m +\log l) \le O(m^2 \log m)$ and hence
we can upper bound $n= O(m^2 \log m)$.
This completes the proof.
\end{proof}

\begin{proof}[Proof of \cref{thm:rop-is-hard}]
We compose the polynomial time reduction
$\textsc{CVP}_{2\gamma}$ to $\textsc{SLDP}_{2\gamma}$
of \cref{thm:red-cvp-2-sldp}
with the  polynomial time  reductions
$\textsc{SLDP}_{2\gamma}$ to $\textsc{ROP}(G,\delta)_{\gamma}$
in \cref{thm:sldp-is-hard-and-reductions}.
If $m$ denotes the dimension of the given instance of \textsc{CVP},
then the dimension $n$ of the constructed instance of \textsc{ROP}
satisfies $m\le n \le O(m^2 \log m)$.

Using~\cite{dinurcvp}, we choose $c$ such that
$\textsc{CVP}_{2\gamma}$ is $\textsc{NP}$-hard
with the approximation factor
$2\gamma(m) = m^{c/\log\log m}$.
Since $m\le n \le m^3$ for sufficiently large~$m$,
we have
\[ \gamma(m)=\frac{1}{2}m^{c/\log\log m} \geq \frac{1}{2}n^{c/3\log\log n} \geq n^{c'/\log\log n}
\] for a suitably chosen $c'$.
Thus we see that $\textsc{ROP}(G,\delta)_{\gamma}$ is
$\textsc{NP}$-hard for the approximation factor
$n^{c'/\log\log n}$,
which completes the proof.
\end{proof}

\subsection{Proof of Theorem~\ref{thm:intro lattice lifting}}\label{sec:eutactic}
Recall that $L\in\ZZ^{m\times n}$ 
is called \textit{right-invertible} over $\ZZ$
if there exists $R\in\ZZ^{n\times m}$ such that $LR=I_m$.
This means that the lattice generated by the columns of~$L$ equals $\ZZ^m$, that is, $L(\ZZ^n) = \ZZ^m$.

\begin{proposition}\label{prop:eutactic}
Suppose $\mathcal{L}\subset\RR^m$ is a lattice generated by the columns of $G\in\ZZ^{m\times m}$ with $\det G\neq 0$.
Then the following are equivalent for integers $n\geq m$ and $s\in\ZZ_{>0}$:
\begin{enumerate}[(1)]
\item\label{it:eutactic 1} There exists a lattice $\mathcal{L}'\subset\RR^n$ generated by an orthonormal basis such that
\[
  \mathcal{L} = s\, P(\mathcal{L}')
\]
where $P:\RR^n\rightarrow\RR^m$ denotes the orthogonal projection onto the first $m$-coordinates.

\item\label{it:eutactic 2} There exists a matrix $L\in\ZZ^{m\times n}$ that is right-invertible over $\ZZ$, such that
\[
(GL)(GL)^T = s^2 I_m.
\]
\end{enumerate}
\end{proposition}

We need the following observation, which easily follows with
Gram-Schmidt orthogonalization.

\begin{lemma}\label{lem:ort-extension}
Suppose $m\leq n$ and $X\in\RR^{m\times n}$.
Then $X$ can be completed to an orthogonal matrix if and only if $XX^T=I_m$.
\end{lemma}

\begin{proof}[Proof of \cref{prop:eutactic}]
$(\ref{it:eutactic 1})\Rightarrow(\ref{it:eutactic 2})$:
Suppose $v_1,v_2,\dots,v_n\in\RR^n$ is an orthonormal basis
generating~$\mathcal{L}'$ and put $u_i := P(v_i)$.
The matrix $U\in\RR^{m\times n}$ with columns  $u_1,\ldots, u_n$ has orthogonal rows, that is, $UU^T = I_m$.
By assumption, the vectors $s u_i$ are in $\mathcal{L}$ and hence
can be written as integer linear combinations of the columns of~$G$.
This means that there exists an integer matrix $L\in\ZZ^{m\times n}$ such that
$G L =s U$.
The rows of this matrix are orthogonal, hence $(GL)(GL)^T = s^2 I_m$.
It remains to show that $L$ is right-invertible over $\ZZ$.
Also by assumption, since the $s u_i$ generate $\mathcal{L}$,
the columns of~$G$ can be written as integer linear combinations of $u_i$.
Thus there exists an integer matrix $R\in\ZZ^{n\times m}$ such that $(GL)R=sUR=G$.
Hence $LR=I_m$, since $\det G\neq 0$, so that $L$~is indeed right-invertible.

$(\ref{it:eutactic 2})\Rightarrow(\ref{it:eutactic 1})$:
Suppose $s^2 I_m=(GL)(GL)^T$. Thus the matrix
$X := \frac{1}{s}GL$ satisfies $XX^T=I_m$.
By \cref{lem:ort-extension} we can complete $X$ to an orthogonal matrix:
thus there exists an orthogonal matrix~$Y\in\mathrm{O}_n$
whose first $m$ rows constitute the matrix $\frac{1}{s}GL$.
Therefore $s PY=GL$. Hence, if $\mathcal{L}'$ denotes the lattice generated by the columns of $Y$,
we have $s P(\mathcal{L}')\subset\mathcal{L}$.
Equality holds since $L$ is right-invertible over $\ZZ$.
\end{proof}

\cref{prop:eutactic} indicates a strategy for proving \cref{thm:intro lattice lifting}.
We need two ingredients:

\begin{itemize}
\item A polynomial time algorithm for computing on input $G$
 an integer $n\geq m$ and a right invertible matrix
 $L\in\ZZ^{m\times n}$ such that $(GL)(GL)^T=s^2 I_m$.

\item An efficient version of \cref{lem:ort-extension}.
\end{itemize}

We prove now that both requirements can be satisfied.
For the first step we will use the following result.

\begin{theorem}[Efficient Waring decomposition for quadratic forms]\label{thm:effectivemordell}
Assume $A\in\QQ^{m\times m}$ is a positive definite, symmetric matrix of bit-length~$b$.
Then, in $\poly(m,b)$-time,
we can compute
$N$~vectors $l_1,l_2,\dots,l_N\in\QQ^m$, where $m\leq N\leq O(m(\log m + \log b))$,
such that
\begin{equation*}
A = \sum_{i=1}^N l_i\, l_i^T.
\end{equation*}
\end{theorem}

\begin{remark}
In 1932, Mordell~\cite{mordell} considered what he called the \textit{Waring's problem for quadratic forms},
the problem of writing a given positive definite quadratic form $f(x)\coloneqq x^T Ax$ as sum of squares of linear forms over $\QQ$.
Note that this equivalent to writing the matrix $A$ as in in \cref{thm:effectivemordell}. 
Mordell proved that a Waring decomposition with $N=m+3$ is always possible,
and described a method for computing the linear forms.
Unfortunately, his method relies on Lagrange's four squares theorem
that every positive integer $D$ can be written as the sum of four squares $D=a^2+b^2+c^2+d^2$.
Randomized polynomial time algorithms to compute $a,b,c,d$ are available~\cite{rabin-shallit},
but to the best of our knowledge, it is still an open problem whether a polynomial time deterministic algorithm exists.
Instead, we will use the lemma below, 
which uses more squares but can easily be shown to run in deterministic polynomial time.
\end{remark}

The proof of \cref{thm:effectivemordell} requires some preparations.

\begin{lemma}
\label{lem:sos}
Given a positive integer $D$ ,
we can in polynomial time compute
$k= O(\log \log D)$ integers $a_1,a_2,\dots,a_k$ such that
\[
 D = a_1^2 + a_2^2 + \dots + a_k^2.
\]
\end{lemma}

\begin{proof}
We first compute 
$a_1\coloneqq\lfloor\sqrt{D}\rfloor$.
Replacing $D$ with $D'\coloneqq D-a_1^2$ and repeating the process,
we compute integers $a_1,a_2,\dots,a_k$ such that
$D=a_1^2+\dots+a_k^2$. Let us bound $k$:
since $a_1^2 \leq D < (a_1+1)^2$, we have $D'=D-a_1^2 \leq (a_1+1)^2-1 -a_1^2 \leq 2 \sqrt{D}$.
This implies that if $D_j$ is computed in the $j$-th step, then
$D_j \leq 2\sqrt{D_{j-1}}\leq\dots\leq 2^{1+\frac{1}{2}+\dots+\frac{1}{2^{j-1}}}D^{2^{-j}}\leq 4 D^{2^{-j}}$.
We deduce that after $k\coloneqq\log_2\log_2 D$ steps we have $D_{k}\leq 8$.
The algorithm terminates after at most $4$ more steps.
\end{proof}

The following result can be found in~\cite[Algorithm~12.1]{schaum-lin-alg}.

\begin{lemma}[Lagrange's method for congruence diagonalization]\label{lem:lagrange-method}
Suppose $X\in\ZZ^{m\times m}$ is a symmetric matrix.
Then 
we can compute  in polynomial time
a matrix $Q\in\ZZ^{m\times m}$ with $\det Q\neq 0$ such that
\[
	Q X Q^T = \diag(d_1,d_2,\dots,d_m)\in\ZZ^{m\times m}
\] is diagonal with integer entries.
\end{lemma}

\begin{proof}[Proof of \cref{thm:effectivemordell}]
By multiplying $A$ with the square of the least common multiple of the denominators of the entries of $A$, we may assume without loss of generality that $A$ is integral. We compute a matrix $Q$ such that $QAQ^T = \diag(d_1,d_2,\dots,d_m)$
by \cref{lem:lagrange-method}.
Since $A$ is positive definite, all $d_i$ are positive.
Using \cref{lem:sos}, for each $i=1,2,\dots,m$,
we compute $e_{ij}$ such that $d_i = \sum_{j=1}^{k_i} e_{ij}^2$.
Denote by $L$ the matrix
\[
L \coloneqq \begin{bmatrix}
e_{11} & \dots & e_{1 k_1} &          &         &          &    & & &   \\
       &       &           & e_{21}   & \dots   & e_{2k_2} &    & & & \\
  	   &       &           &          &         &          &\ddots & & &\\
       &       &           &          &         &          &       & e_{m1}& \dots & e_{mk_m}
	\end{bmatrix}\quad\in\quad \ZZ^{m\times \sum k_i}.
\]
Then $LL^T=\diag(d_1,d_2,\dots,d_m)=QAQ^T$, so $A = (Q^{-1}L) (Q^{-1}L)^{T}$.
Defining $l_i\in\QQ^m$ to be the $i$-th column of $Q^{-1}L$,
we have $A=\sum_{i=1}^N l_i l_i^T$ where $N\coloneqq \sum k_i$ and this proves the claim. For the bound on $N$,
we first note that since the algorithm in \cref{lem:lagrange-method} runs in polynomial time,
we have 
$\log d_i = (bm)^{O(1)}$.
\Cref{lem:sos} then implies that~$k_i = O(\log\log d_i) = O(\log(bm))$ and~$N = O(m\log(bm))$.
\end{proof}

The efficient Waring decomposition is the first ingredient in our proof of \cref{thm:intro lattice lifting}.
The second ingredient is an efficient version of \cref{lem:ort-extension}.

\begin{lemma}\label{lem:completion-problem}
Suppose $m\leq n$ and $X\in\QQ^{m\times n}$ satisfies $X X^T = I_m$.
Then, we can compute in polynomial time an orthogonal matrix $Y\in\QQ^{n\times n}$ such that the first $m$ rows of $Y$ are the rows of $X$.
\end{lemma}

\begin{proof}
We construct $Y$ as a product of reflections at rational hyperplanes.

First note that for any $v,w\in\QQ^n$ of norm one,
there is an orthogonal matrix $Y\in\QQ^{n\times n}$
such that $Yv = w$. Indeed, w.l.o.g.\ $v\ne-w$,
and take for $Y$ the reflection at the hyperplane orthogonal to $v+w$,
which is given by the linear map
\begin{equation*}
 Y x  := 2 \, \frac{\langle x, v+w\rangle}{\Vert v+w\Vert^2} \, (v+w)  - x .
\end{equation*}
Note that $Y$ defines a rational orthogonal matrix and indeed $Yv=w$.
Denote by $e_i$ the standard basis vectors.

Given $X$ as in the lemma,
we compute an orthogonal matrix $Y_1\in\QQ^{n\times n}$
which maps~$e_1$ to the first row of~$X$.
Since $Y$ is orthogonal and $Y^2=I_n$, the first row of $Y$ equals the first row of $X$.
We construct $Y$ by iterating this process.
\end{proof}

\begin{remark}
M. Hall~\cite{hall-completion-2,hall-completion}
gave necessary conditions for an integer matrix to be completable
to a scalar multiple of an orthogonal matrix
in the sense of \cref{lem:completion-problem}.
\end{remark}

\begin{proof}[Proof of \cref{thm:intro lattice lifting}]
Given $G\in\ZZ^{m\times m}$ with $\det G\neq 0$, we compute
$s := m \|G\|_{\max} + 1 $.
By \cref{lem:singular} we have $\sigma_{\max}(G)< s$.
The minimum eigenvalue of $G^{-1}G^{-T}$ satisfies
\[
\lambda_{\min}(G^{-1}G^{-T}) = \sigma^2_{\min}(G^{-1}) = \frac{1}{\sigma^2_{\max}(G)} > \frac{1}{s^2}.
\]
Consequently, the matrix $A\coloneqq s^2 (G^{-1}G^{-T})-I_m$ is positive definite.

Using \cref{thm:effectivemordell}, we compute in polynomial time a number $N$
and a matrix $L\in\QQ^{m\times N}$ such that $A=LL^T$.
Let $f\in\ZZ_{>0}$ denote the least common multiple of the denominators of the entries of~$L$,
so that $L':=fL$ is integral. Then
$(sf)^2 G^{-1}G^{-T} - f^2 I_m= L'  (L')^T$.

Using \cref{lem:sos}, we compute integers $b_1,\ldots,b_p$ such that $f^2-1=b_1^2+b_2^2+\dots+b_p^2$.
Since $L$ is computed from $A$ in polynomial time, the bit-length of $f$ is $\poly(b,m)$ so \cref{lem:sos} implies $p=O(\log(bm))$.
Considering the integer matrix
$L'' = \begin{bmatrix}
		b_1 I_m \;\;\vline\;\; b_2 I_m \;\;\vline\;\; \dots  \;\;\vline\;\; b_p I_m  \;\;\vline\;\; L'
	\end{bmatrix}$
 of format $m\times (pm+N)$, we get
\begin{equation}\label{eq:l-prime}
 (sf)^2 G^{-1}G^{-T}  - I_m = (b_1^2+b_2^2+\dots+b_p^2)I_m+ L'  (L')^T = L'' (L'')^T.
\end{equation} 
We now consider the matrix $X\coloneqq \begin{bmatrix}
		G  \;\;\vline\;\; GL''
\end{bmatrix}$.
By \cref{eq:l-prime}, we have $XX^T = (sf)^2 I_m$.
Using \cref{lem:completion-problem} with input $(sf)^{-1} X$,
we compute a rational, orthogonal matrix $Y\in\QQ^{n\times n}$
such that the first $m$ rows of $Y$ are the rows of $(sf)^{-1} X$ where $n\coloneqq m+(pm+N)$.

The columns $v_1,\dots,v_n$ of $Y$ form an orthonormal basis. 
Moreover, the orthogonal projections~$P(v_i)$ onto the first $m$ coordinates
equal the columns of~$(sf)^{-1} X$.
Thus if we denote by $\mathcal{L}'$ the lattice generated by $v_1,\dots,v_n$, then $sf P(\mathcal{L}')$ is the lattice spanned by the columns of~$X = \begin{bmatrix}
    G  \;\;\vline\;\; GL''
\end{bmatrix}$, which is equal to $\mathcal L = G(\ZZ^m)$, because $L''$ is an integer matrix.
We conclude that~$sf P(\mathcal{L}')=\mathcal{L}$.
\end{proof}

\section{Orbit Problems and the Kempf-Ness Approach}\label{sec:kempf-ness}
In \cref{subsec:intro conclusion and outlook} we outlined a general approach for
the Orbit Equality Problem,
based on the Kempf-Ness theorem~\cite{kempfness}.
Here we analyze this approach for torus actions.
We formulate a conjecture on the complexity of computing approximate solutions to unconstrained geometric programs.
A positive answer to the conjecture leads to a numerical algorithm for the Orbit Equality Problem  
that runs in polynomial time, provided \cref{hyp:sep} is true.

\subsection{The general Kempf-Ness theorem}
This theorem holds for any reductive group~$T$
(not just a torus) with a rational action on a finite dimensional vector space~$V$.
We denote by $K$ a maximal compact subgroup of $T$ and assume that $V$
is endowed with a Hermitian inner product such that $K$ acts by isometries.
The inner product defines a norm and Euclidean metric on $V$.
In the special case~$T=(\CC^\times)^d$ of a torus we have $K=(S^1)^d$.

It is well known~\cite[Thm.~2.3.6]{derksen-kemper}
that each orbit closure $\overline{\OO_v}$ contains a unique closed orbit.
Therefore, we have
$\overline{\OO_v}\cap\overline{\OO_w}\neq\varnothing$
iff the orbit closures $\overline{\OO_v}$ and $\overline{\OO_w}$ share the same closed orbit.
We therefore focus on closed orbits:
let us call a vector $v\in V$ \textit{polystable} if its orbit $\OO_v$ is closed.
We denote by $V^{ps}$ the set of polystable vectors in~$V$.
As in \cref{sec:quotients}, we can endow the space $V^{ps}/T$
of closed orbits with the quotient topology, which is Hausdorff.
The Kempf-Ness Theorem~\cite{kempfness}
states that the space $V^{ps}/T$ is homeomorphic to a ``smaller object'',
which is defined in analytic terms, namely the space $\Crit(V)/K$
of $K$-orbits of the closed and $K$-invariant subset $\Crit(V)$ of critical points of $V$.

The critical points are defined in terms of the \textit{Kempf-Ness function}:
\begin{equation}\label{eq:kn-function}
    F_v\colon T \rightarrow \RR, \quad F_v(g) := \log\Vert g \cdot v\Vert = \frac12 \log\Vert g \cdot v\Vert^2.
\end{equation}
Let us denote the induced action of $\Lie(T)$ on $V$
by  $x\cdot v \coloneqq \frac{d}{dt}\big |_{t=0} \, (e^{tx}\cdot v)$,
for $x\in\Lie(T)$ and $v\in V$.
Note that $x\cdot v = v + M^T x$
for the action \cref{eq:def-action} of a torus.
One checks that the derivative of $F_v$ at~$I$ is given by
$D_I F_v x = \|v\|^{-2}\, \langle x\cdot v, v\rangle$.
Thus we define $v\in V$ to be \textit{critical} iff
$\langle x\cdot v, v\rangle = 0$ for all $x\in\Lie(T)$.
By definition, the set $\Crit(V)$ of critical vectors in $V$ is a closed real algebraic set in~$V$,
cut out by real quadratic polynomials.
Moreover, $\Crit(V)$ is $K$-invariant,
since we assume the Hermitian inner product to be $K$-invariant.

It is easy to see that closed orbits contain a critical point.
The celebrated Kempf-Ness theorem~\cite{kempfness} states
the converse: the orbits of critical points are closed, thus
$\Crit(V)\subset V^{ps}$.
Even more is known:

\begin{theorem}[The Kempf-Ness theorem]\label{thm:kempf-ness}
\begin{alphaenumerate}
\item The orbit $\OO_v$ is closed if and only if \\$\OO_v\cap \Crit(V)\neq\varnothing$.
\item The intersection $\KK_{v^\star}:=\overline{\OO_v}\cap\Crit(V)$ is a single $K$-orbit,
 that we call the \textit{Kempf-Ness orbit} of $v$.
\item 
 We have $\overline{\OO_v}\cap\overline{\OO_w}\neq\varnothing$ if and only if $\KK_{v^\star}=\KK_{w^\star}$.
\end{alphaenumerate}
\end{theorem}

The proof of \cref{thm:kempf-ness} relies on convexity properties of the function~$F_v$.
One observes that $t\mapsto F_v(\gamma(t))$ is convex for all all curves of the form
$\gamma(t):= \Exp(tx) \cdot v$.
A more conceptual view of the situation is as follows (see~\cite{BFGOWW2}).
One can view $F_v$ as a function on $T/K =\{ Kg \mid g \in T \}$,
since $F_v$ is constant on the cosets $Kg$.
Moreover, $T/K$ is a symmetric space with respect to a $K$-invariant
Riemannian metric, whose geodesics are exactly the curves arising from the $\gamma$.
For this reason, one calls $F_v$ a \textit{geodesically convex} function on~$T/K$.
In the special case where $T=(\CC^\times)^d$ is the torus, we have the isometry
\begin{equation}\label{eq:exp-par}
 \RR^d \stackrel{\sim}{\longrightarrow} T/K, \ x\mapsto K \Exp(x/2)
\end{equation}
given by the exponential map.
In general, the convexity implies that $v^\star\in\overline{\OO_v}$ is critical iff
\[
\Vert v^\star\Vert = \inf_{v'\in\OO_v} \, \Vert v'\Vert.
\]

\begin{remark}
The second property in \cref{thm:kempf-ness} expresses that the continuous map $\Crit(V)/K \to V^{ps}/T$,
induced by the inclusion, is a bijection.
Equivalently,  this map is a homeomorphism.\footnote{Pass to projective spaces to see this.}
\end{remark}

\subsection{Efficient approximation of the Kempf-Ness orbit}\label{sec:EA-KN0}
We return to the situation of the action of a torus $T$ on $V=\CC^n$.
Let $\omega_1,\ldots,\omega_n\in\ZZ^d$ denote the corresponding weights,
i.e., the columns of the weight matrix $M$.
We will always assume that the affine span of these weights is $\RR^d$.
We fix a vector $v\in V$ and assume for simplicity that
$q_i := |v_i|^2 > 0$ for all~$i$.
Using the parametrization~\eqref{eq:exp-par} and multiplying by $2$ to simplify the presentation,
the Kempf-Ness function~\eqref{eq:kn-function} becomes
\begin{equation}\label{eq:def-kn-func}
	f \colon \RR^d \, \rightarrow \, \RR, \, f(x) := 2\log \Vert e^{x/2}\cdot v\Vert
   = \log\,\Big(\,\sum_{i=1}^n q_i e^{\omega_i^T x} \,\Big) .
\end{equation}
The gradient $\nabla f(x)$ and the Hessian of $f$ at $x$ are given by
\begin{equation}
\label{eq:grad-of-kn-func}
  \nabla f(x) = \frac{\sum_{i=1}^n q_i   e^{\omega_i^T x}\, \omega_i }{\sum_{i=1}^n q_i  e^{\omega_i^T x}}, \quad
 \nabla^2 f(x) = \frac{\sum_{i=1}^n q_i   e^{\omega_i^T x}\, \omega_i\omega_i^T }{\sum_{i=1}^n q_i  e^{\omega_i^T x} }
          - \nabla f(x) \nabla f(x)^T .
\end{equation}
This shows that $\nabla f(x)$ lies in the interior of the \textit{weight polytope}
\[
 P \coloneqq  \mathrm{conv}\big(\omega_1,\ldots,  \omega_n\big) .
\]
Moreover, the Cauchy-Schwarz inequality implies for all nonzero $y\in\RR^d$,
\begin{equation}\label{eq:yHyT}
 y^T\nabla^2 f(x) y= \frac{\sum_{i=1}^n q_i   e^{\omega_i^T x}\, (y^T\omega_i)^2 }{\sum_{i=1}^n q_i  e^{\omega_i^T x} }
          -  (y^T\nabla f(x))^2 \ >\  0 ,
\end{equation}
which directly proves that $f$ is a strictly convex function.%
\footnote{To see this, we split $\sum_i q_i e^{\omega_i^T x} y^T \omega_i$ as $\sum_i q_i^{\frac{1}{2}}e^{\frac{1}{2}\omega_i^T x}y^T \omega_i \, \cdot \, q_i^{\frac{1}{2}}e^{\frac{1}{2}\omega_i^T x}$ and apply the Cauchy-Schwarz inequality. Since the vectors $\omega_i$ affinely span~$\RR^d$ we have $y^T\omega_i\neq y^T\omega_j$ for some $i\neq j$.
This shows that the vectors $(q_i^{\frac{1}{2}}e^{\frac{1}{2}\omega_i^T x}y^T\omega_i\mid i\in [n])$ and $(q_i^{\frac{1}{2}}e^{\frac{1}{2}\omega_i^T x}\mid i\in [n])$ cannot be proportional and hence the Cauchy-Schwarz inequality must be strict.
}

For the following equivalences, see~\cite{BLNW:20,torus}.
The function $f$~has a finite infimum iff $0\in P$.
This is equivalent to $0 \not\in\overline{\OO_v}$.
Moreover, the infimum $f_\star$ is attained iff $0$ lies in the interior of~$P$,
which is equivalent to the orbit $\OO_v$ being closed.
By the strict convexity, the minimum is attained at a unique point $x^\star\in\RR^d$.
By \cref{thm:kempf-ness},
$v^\star := e^{x^\star}\cdot v$
defines the $K$-orbit $\KK_{v^\star}$.

There is a vast amount of literature on the problem of minimizing~$f(x)$
using the ellipsoid or the interior point methods,
see~\cite{BLNW:20, geomet-tut,nemirovski-ellipsoid, vishnoi-singh, vishnoi-straszak, leake-vishnoi} and the references therein.
On input $\varepsilon>0$, these algorithms return a point $x\in\RR^d$ with $f(x)-f_\star<\varepsilon$ in polynomial time.
However, surprisingly, we are not aware of any result on computing a certified approximation to $x^\star$ in polynomial time, and currently this seems to be an open problem.
We conjecture that this can be achieved in polynomial time:

\begin{openproblem}
\label{conj:optimal}
	There exists an algorithm that on input $M\in\ZZ^{d\times n}$ such that $0\in\mathrm{int}(P)$, a vector $v\in(\QQ(\ii)^\times)^n$, and $\varepsilon\in\QQ_{>0}$, computes $x\in\QQ^d$ such that \[
	\Vert x-x^\star\Vert < \varepsilon,
	\] in time $\poly(d,n,B,b,\log\frac{1}{\varepsilon})$, where $B$ is the bit-length of $M$ and $b$ is the bit-length of $v$.
\end{openproblem}

We will now explain why computing an approximation to $x^\star$ seems to be more challenging than approximately minimizing $f(x)$.
In \cref{ex:small-hessian} below, we provide a family of Kempf-Ness functions $f$ and points~$x$ such that $f(x)-f_\star$ becomes doubly-exponentially small in the input bit-length, while the distance to the true minimizer~$x^\star$ remains constant: $\Vert x-x^\star\Vert=1$.
This implies that using known minimization algorithms in a black box will not be sufficient to solve \cref{conj:optimal} and more specific algorithms appear to be required.

\begin{example}
\label{ex:small-hessian}
We consider the action of $T\coloneqq (\CC^\times)^2$ on $V\coloneqq\CC^4$ with the following weights: $\omega_1\coloneqq (1,0),\;\omega_2\coloneqq (-2,0),\; \omega_3\coloneqq(-N,1),\; \omega_4\coloneqq(-N,-1)$ where $N>2$ is a positive integer.
See \cref{fig:kempf ness example} for illustration.
The Kempf-Ness function~\eqref{eq:def-kn-func} for $v\coloneqq (1,1,1,1)$ reads
\[
f(x_1,x_2) = \log\Big( e^{x_1} + e^{-2x_1} + e^{-N x_1 + x_2} + e^{-N x_1 - x_2}\Big).
\]
We note that $f(x)\geq 0$ for every $x\in\RR^d$ since at least one of the exponents $x_1, -2x_1$ is non-negative.

We now focus on the unique minimizer~$x^\star=(x^\star_1,x^\star_2)$ of the convex program $f_\star\coloneqq\min_x f(x)$, which is characterized by the property that the gradient vanishes: $\nabla f(x^\star)=0$.
By \cref{eq:grad-of-kn-func} this is equivalent to
\begin{equation}\label{eq:grad-for-ex}
	e^{x^\star_1} - 2e^{-2x^\star_1} - N e^{-Nx^\star_1+x^\star_2}  - N e^{-Nx^\star_1-x^\star_2} = 0
\quad\text{and}\quad
  e^{-Nx^\star_1+x^\star_2} - e^{-Nx^\star_1-x^\star_2} = 0.
\end{equation}
The second equation implies $x^\star_2=0$.
The first equation then reduces to
\begin{equation}\label{eq:grad-for-ex-simplified}
  e^{x^\star_1} - 2e^{-2x^\star_1} - 2 N e^{-Nx^\star_1} = 0,
\end{equation}
which implies that $x^\star_1$ is positive and decreases with~$N$.
In the limit $N\rightarrow\infty$, the first equation degenerates to $e^{x^\star_1}-2 e^{-2x^\star_1}=0$.
We deduce that $e^{x^{\star}_1}>\sqrt[3]{2}$ for every $N>2$.
In fact, we have:
\begin{equation}\label{eq:quantify}
  \sqrt[3]{2} < e^{x^{\star}_1} < \left( 1+\frac N{2^{N/3}} \right)\, \sqrt[3]{2}.
\end{equation}
To see this, we write $e^{x^\star_1}=(1+\varepsilon)\sqrt[3]{2}$ where $\varepsilon>0$.
Then \cref{eq:grad-for-ex-simplified} implies
\[
  (1+\varepsilon) \sqrt[3]{2} - \frac{\sqrt[3]{2}}{(1+\varepsilon)^2} = \frac{2N}{(1+\varepsilon)^N\, 2^{N/3}}.
\]
Multiplying this equation by $(1+\varepsilon)^N / \sqrt[3]{2}$, we get $(1+\varepsilon)^{N-2}\,(\varepsilon^2 + 3\varepsilon+3)\, \varepsilon = 2N \times 2^{-(N+1)/3}$.
Since~$\varepsilon>0$, we have $\varepsilon^2+3\varepsilon+3>3$ and $3\varepsilon < 2N \times 2^{-(N+1)/3}$, which implies \cref{eq:quantify}.

We now consider the point $x\coloneqq (x^\star_1, 1)$, whose first coordinate agrees with the first coordinate of $x^\star=(x^\star_1,0)$ but has a constant distance away from $x^\star$ in the second coordinate.
We will show that $f(x)-f_\star$ is exponentially small in $N$:
\[
	e^{f(x)-f_\star} = 1 + \frac{e^{-N x^\star_1}(e+e^{-1}-2)}{e^{f_\star}} \leq 1 + e^{-Nx^\star_1}(e+e^{-1}-2)
\]
where the inequality holds since $f_\star\geq 0$.
Taking logarithms we get
\[
	f(x)-f_\star \leq \log(1 + e^{-Nx^\star_1}(e+e^{-1}-2)) \leq e^{-Nx^\star_1}(e+e^{-1}-2)<2^{-N/3}(e+e^{-1}-2)
\]
since $e^{x^\star_1}>\sqrt[3]{2}$.
This shows that $f(x)-f_\star$ is exponentially small in $N$ and hence doubly-exponentially small in the input bit-length, and yet $\Vert x-x^\star\Vert = 1$.
\end{example}

\begin{figure}
	\centering
	\begin{tikzpicture}[scale = 0.6]
	\pgfdeclarelayer{nodelayer}
    \pgfdeclarelayer{edgelayer}
    \pgfsetlayers{main,nodelayer,edgelayer}
	\begin{pgfonlayer}{nodelayer}
		\node  (0) at (0, 5) {};
		\node  (1) at (0, -5) {};
		\node  (2) at (5, 0) {};
		\node  (3) at (-10, 0) {};
		\node  [fill = black, scale=0.3, circle, label={above:$\omega_1$}](4) at (2, 0) {};
		\node  [fill = black, scale=0.3, circle, label={below:$\omega_2$}](5) at (-4, 0) {};
		\node  [fill = black, scale=0.3, circle, label={above:$\omega_3$}](6) at (-10, 2) {};
		\node  [fill = black, scale=0.3, circle, label={below:$\omega_4$}](7) at (-10, -2) {};
		\node  (9) at (0, 0) {};
	\end{pgfonlayer}
	\begin{pgfonlayer}{edgelayer}
		\draw [fill =gray!50, dashed](1,0) ellipse (0.2cm and 3cm);
		\draw [dashed, gray](3.center) to (2.center);
		\draw [dashed, gray](0.center) to (1.center);
		\draw [line width = 1pt](7.center) to (4.center);
		\draw [line width = 1pt](4.center) to (6.center);
		\draw [line width = 1pt](6.center) to (7.center);
	\end{pgfonlayer}
\end{tikzpicture}
\caption{\small
An example of a torus action with weights $(1,0), (-2,0), (-N,1), (-N,-1)$. The points in the shaded region all have small $f(x)-f_\star$ but they can be far away from the optimal solution $x^\star$.
}
\end{figure}
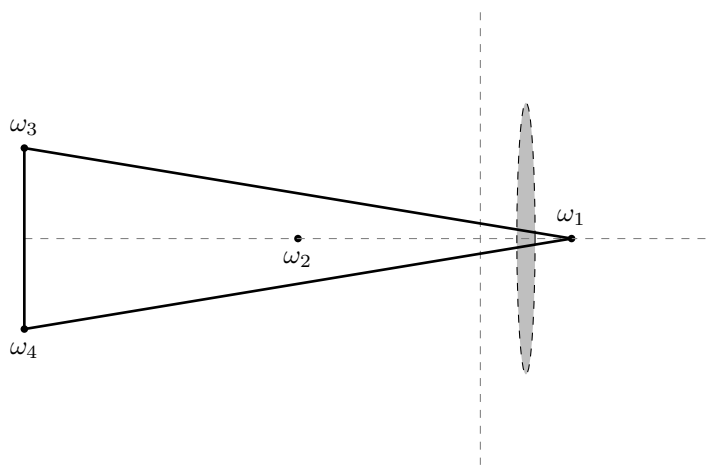

\begin{remark}
In the setting of \cref{ex:small-hessian}, consider the vectors~$v\coloneqq (1,1,1,1)$ and $w\coloneqq (1,1,2,2)$, and denote by $x^\star,y^\star$ the minimizers of the Kempf-Ness functions of~$v,w$, respectively, i.e., $e^{x^\star}=v^\star$ and $e^{y^\star}=w^\star$.
By \cref{ex:small-hessian}, we have $x^\star=(x^\star_1,0)$, where $x^\star_1$ satisfies \cref{eq:quantify}.
By a similar reasoning, we also have $y^\star=(y^\star_1,0)$, where $y^\star_1$ satisfies
\[
  \sqrt[3]{2} < e^{y^\star_1} < \left( 1 + \frac{2N}{2^{N/3}} \right)\, \sqrt[3]{2}.
\]
We will now show that the coordinates of $v^\star, w^\star$ are exponentially close in $N$.
To this end, denote $\varepsilon_N\coloneqq 2N \times 2^{-N/3}$.
We have $|(v^\star)_1-(w^\star)_1| = |e^{x^\star_1}-e^{y^\star_1}| < 2^{1/3}\,\varepsilon_N $.
For the second coordinate we have $|(v^\star)_2-(w^\star)_2|=|e^{-2x^\star_1}-e^{-2y^\star_1}|<2^{-2/3}\big((1+\varepsilon_N)^2-1\big)< 2^{4/3}\varepsilon_N$, using the crude bound $(1+\varepsilon_N)^2-1<4\varepsilon_N$, which holds for $\varepsilon_N<2$ and hence for $N$ large enough.
For the third and fourth coordinate we have $|(v^\star)_3-(w^\star)_3|=|(v^\star)_4-(w^\star)_4|=|e^{-Nx^\star_1}-2e^{-Ny^\star_1}|<|e^{-Nx^\star_1}|+2|e^{-Ny^\star_1}|<3\times 2^{-N/3}<\varepsilon_N$.
Hence $\Vert v^\star-w^\star\Vert_{\infty}<2^{4/3}\varepsilon_N$ and we conclude
\[
  \dist(\KK_{v^\star},\KK_{w^\star}) \leq \Vert v^\star-w^\star\Vert \leq 2\Vert v^\star-w^\star\Vert_{\infty} < 2^{7/3}\, \varepsilon_N
\]
where the second inequality follows from the inequality $\Vert\cdot\Vert_2\leq \sqrt{n}\Vert\cdot\Vert_\infty$ which holds in the $n$-dimensional Euclidean space.
This shows that the Euclidean distance between $\KK_{v^\star}$ and $\KK_{w^\star}$ can be exponentially small in $N$ and doubly exponentially small in the bit-length of $N$.
On the other hand, it is easy to lower bound the logarithmic distance: It is easy to verify that \[
H=\begin{bmatrix}
    2 & 1 & 0 & 0\\
    2N & 0 & 1 & 1
\end{bmatrix}
\] is the matrix of rational invariants and $H\Log v = (0,0)$ and $H\Log w = (0, 2\log 2)$. \cref{cor:delta-ld-linear-forms} then implies \[
\distLD(\KK_{v^\star},\KK_{w^\star}) \, \geq \, \distLD(\OO_v,\OO_w) \, \geq \, \frac{1}{\sigma_{\max}(H)} \, \Delta(H \Log v, H\Log w) \geq \frac{\log 2}{2N}
\] where in the last inequality we use \cref{lem:singular} to upper bound $\sigma_{\max}(H)\leq 4N$.

More generally, if \cref{hyp:sep} holds then it is true for general torus actions that $\distLD(\KK_{v^\star},\KK_{w^\star})$ is at most singly exponentially small:
If $\KK_{v^\star}\neq\KK_{w^\star}$, then, in general, $\OO_v\neq\OO_w$, hence, $\distLD(\KK_{v^\star},\KK_{w^\star})\geq \distLD(\OO_v,\OO_w)\geq \sep_T(d,n,B,b)$.
This is one of the main reasons why we work with the logarithmic distance instead of the Euclidean distance.
\end{remark}

\subsection{Deciding the equality of the Kempf-Ness orbits}\label{sec:kn-approach-in-poly-time}
We now show that, assuming \cref{hyp:sep} and \cref{conj:optimal}, that for given $v,w\in (\CC^\times)^n$,
it is possible to decide in polynomial time whether the orbits $\OO_v$ for $\OO_w$ are equal.
By \cref{thm:kempf-ness}, this is equivalent to the equality of the Kempf-Ness orbits
$\KK_{v^\star}$ and $\KK_{w^\star}$. We decide this
by computing the approximate distance of the corresponding orbits
$\KK_{v^\star}$ and $\KK_{w^\star}$
with sufficent accuracy.
This has two ingredients: first we compute $x'$ and $y'$ that are $\varepsilon$-close approximations of $x^\star$ and $y^\star$, respectively.
This is possible by  \cref{conj:optimal}.
In a second step, we compute an approximation $D$ of the distance
$\delta:=\distLD (\KK_{v'} , \KK_{w'})$ between $\KK_{v'}$ and $\KK_{w'}$
such that $\delta \leq D \leq \gamma \delta$.
For this we use \cref{thm:gap-algorithms-main}~(\ref{it:gap 2}),
see \cref{prop:delta-prime-approx} below.
When choosing $\varepsilon=2^{-\ell}$ sufficiently small, we show that
$\KK_{v^\star} = \KK_{w^\star}$ iff $D < \gamma\varepsilon$,
which allows to decide orbit equality.
\Cref{hyp:sep} is needed to make sure that it is sufficient to
choose $\ell$ as a polynomial in the input bit-length
in order to obtain a polynomial time algorithm.

We now give the technical details.
We recall that we restrict ourselves to torus actions where $0$ lies in the interior of the convex hull $P$ of the set of weights $\Omega$.
This condition guarantees that {\em all} orbits of vectors in $(\CC^\times)^n$ are closed in~$\CC^n$;
on the other hand, when $0\not\in\inter(P)$, then no such orbit is closed, see~\cite{torus}.

We first state the following straightforward consequence of \cref{thm:gap-algorithms-main}~(\ref{it:gap 2}).

\begin{corollary}\label{prop:delta-prime-approx}
We assume the \cref{hyp:sep}. 
There is an exponentially bounded approximation factor $\gamma=\gamma(d,n,B,b)$
such that,
given a $T$-action by a weight matrix $M\in\ZZ^{d\times n}$, and
given vectors  $v,w\in(\QQ(\ii)^\times)^n$ and $x,y\in\QQ^d$,
we can compute in polynomial time a number $D\in\QQ_{\geq 0}$ such that
\[
    \distLD ( \KK_{v'}, \KK_{w'} )\, \leq\, D \,\leq\, \gamma\, \distLD (\KK_{v'} , \KK_{w'}) ,
\]
where $v'\coloneqq e^x\cdot v$ and $w'\coloneqq e^y\cdot w$.
\end{corollary}

Using the approximation factor~$\gamma$ of \cref{prop:delta-prime-approx}, we define the parameter
$$
 \varepsilon :=\frac{1}{2\gamma} \sep_T(d,n,B,b) .
$$
Assuming \cref{conj:optimal}, we can compute in polynomial time approximations
\[
    \Vert x - x^\star \Vert <  \tfrac12\|M\|^{-1} \varepsilon, \qquad
    \Vert y-y^\star\Vert < \tfrac12\|M\|^{-1} \varepsilon . 
\]
Recall that  $v^\ast = e^{x^\star}\cdot v$. We set $v'\coloneqq e^x\cdot v$. Then
$\Log v^\ast = \Log v + M^Tx^\star$ and $\Log v' = \Log v + M^Tx$, hence
$$
 \distLD(v',v^\star) = \Delta(\Log v + M^T x , \Log v + M^T x^\star) =\Vert M^T(x-x^\star)\Vert
     \leq \tfrac12 \varepsilon ,
$$
which implies
$\distLD(\KK_{v'},\KK_{v^\star}) \le \tfrac12 \varepsilon$.
Analogously, we get
$\distLD(\KK_{w'},\KK_{w^\star}) \le \tfrac12 \varepsilon$
for $w'\coloneqq e^y\cdot w$.
Using the triangle inequality, we obtain
$$
 \distLD(\KK_{v^\star},\KK_{w^\star}) \ \le\ \distLD(\KK_{v^\star},\KK_{v'}) + \distLD(\KK_{v'},\KK_{w'}) + \distLD(\KK_{w'},\KK_{w^\star})
  \ \le\ \distLD(\KK_{v'},\KK_{w'}) + \varepsilon .
$$
By symmetry, this implies
\begin{equation} \label{eq:kn-orbits-approx}
     \big| \distLD(\KK_{v^\star},\KK_{w^\star}) - \distLD(\KK_{v'},\KK_{w'})\big| < \varepsilon .
\end{equation}
We now observe the following two implications:
\begin{equation}\label{eq:iff}
\begin{split}
     \distLD(\KK_{v^\star},\KK_{w^\star}) = 0\quad & \Longrightarrow \quad \distLD(\KK_{v'},\KK_{w'}) < \varepsilon\\
      \distLD(\KK_{v^\star},\KK_{w^\star}) > 0\quad & \Longrightarrow \quad \distLD(\KK_{v'},\KK_{w'}) > \gamma \varepsilon .
\end{split}
\end{equation}
The first implication is an obvious consequence of \cref{eq:kn-orbits-approx}.
For the second implication, note that
$\KK_{v^\star} \ne \KK_{w^\star}$ implies
$\distLD(\KK_{v^\star},\KK_{w^\star}) \ge \distLD(\OO_{v},\OO_{w}) \ge \sep_T$,
since $\KK_{v^\star} \subset \OO_v$ and $\KK_{w^\star} \subset \OO_w$.
Combined with \cref{eq:kn-orbits-approx}, this gives
$$
 \distLD(\KK_{v'},\KK_{w'}) \ >\  \sep_T  - \varepsilon = 2\gamma \varepsilon  - \varepsilon \ge  \gamma \varepsilon ,
$$
which proves \cref{eq:iff}.

We use \cref{prop:delta-prime-approx} to compute in polynomial time the number~$D$ such that
$\delta \leq D \leq \gamma \delta$, where $\delta:=\distLD (\KK_{v'} , \KK_{w'})$.
\Cref{eq:iff} implies now:
\begin{equation*}
\begin{split}
     \KK_{v^\star}= \KK_{w^\star}\quad & \Longrightarrow \quad   D < \gamma\varepsilon\\
     \KK_{v^\star}\ne \KK_{w^\star}\quad & \Longrightarrow \quad D > \gamma\varepsilon .
\end{split}
\end{equation*}
So the knowledge of $D$ allows to decide whether $\KK_{v^\star}= \KK_{w^\star}$
and hence whether $\OO_{v}= \OO_{w}$.

\bibliography{torus-abc-lipics}

\begin{thebibliography}{10}

\bibitem{aks-primality}
Manindra Agrawal, Neeraj Kayal, and Nitin Saxena.
\newblock {PRIMES} is in {P}.
\newblock {\em Annals of Mathematics}, 160(2):781--793, 2004.

\bibitem{aharonov-regev}
Dorit Aharonov and Oded Regev.
\newblock Lattice problems in $\textsc{NP} \cap \textsc{coNP}$.
\newblock {\em J. ACM}, 52(5):749--765, sep 2005.
\newblock \href {https://doi.org/10.1145/1089023.1089025}
  {\path{doi:10.1145/1089023.1089025}}.

\bibitem{AZGLOW}
Zeyuan Allen-Zhu, Ankit Garg, Yuanzhi Li, Rafael Oliveira, and Avi Wigderson.
\newblock Operator scaling via geodesically convex optimization, invariant
  theory and polynomial identity testing.
\newblock In {\em S{TOC}'18---{P}roceedings of the 50th {A}nnual {ACM} {SIGACT}
  {S}ymposium on {T}heory of {C}omputing}, pages 172--181. ACM, New York, 2018.
\newblock \href {https://doi.org/10.1145/3188745.3188942}
  {\path{doi:10.1145/3188745.3188942}}.

\bibitem{AKRS:21b}
Carlos Am\'{e}ndola, Kathl\'{e}n Kohn, Philipp Reichenbach, and Anna Seigal.
\newblock Invariant theory and scaling algorithms for maximum likelihood
  estimation.
\newblock {\em SIAM J. Appl. Algebra Geom.}, 5(2):304--337, 2021.
\newblock \href {https://doi.org/10.1137/20M1328932}
  {\path{doi:10.1137/20M1328932}}.

\bibitem{amendola2021toric}
Carlos Am{\'e}ndola, Kathl{\'e}n Kohn, Philipp Reichenbach, and Anna Seigal.
\newblock Toric invariant theory for maximum likelihood estimation in
  log-linear models.
\newblock {\em Algebraic Statistics}, 12(2):187--211, 2021.

\bibitem{audin2012torus}
Michele Audin.
\newblock {\em Torus actions on symplectic manifolds}, volume~93 of {\em
  Progress in Mathematics}.
\newblock Birkh{\"a}user Basel, 2012.

\bibitem{baker_explicit}
Alan Baker.
\newblock Experiments on the abc-conjecture.
\newblock {\em Publicationes Mathematicae}, 65, 11 2004.

\bibitem{baker98}
Alan Baker.
\newblock Logarithmic forms and the abc-conjecture.
\newblock In {\em Number Theory: Diophantine, Computational and Algebraic
  Aspects. Proceedings of the International Conference held in Eger, Hungary,
  July 29-August 2, 1996}, pages 37--44. De Gruyter, 2011.
\newblock URL: \url{https://doi.org/10.1515/9783110809794.37}, \href
  {https://doi.org/doi:10.1515/9783110809794.37}
  {\path{doi:doi:10.1515/9783110809794.37}}.

\bibitem{baker_wustholz_2008}
Alan Baker and Gisbert W\"{u}stholz.
\newblock {\em Logarithmic Forms and Diophantine Geometry}.
\newblock New Mathematical Monographs. Cambridge University Press, 2008.
\newblock \href {https://doi.org/10.1017/CBO9780511542862}
  {\path{doi:10.1017/CBO9780511542862}}.

\bibitem{BILPS:20}
Markus Bl{\"{a}}ser, Christian Ikenmeyer, Vladimir Lysikov, Anurag Pandey, and
  Frank-Olaf Schreyer.
\newblock Variety membership testing, algebraic natural proofs, and geometric
  complexity theory.
\newblock arXiv:1911.02534, 2020.

\bibitem{bombieri_gubler_2006}
Enrico Bombieri and Walter Gubler.
\newblock {\em Heights in Diophantine Geometry}.
\newblock New Mathematical Monographs. Cambridge University Press, 2006.
\newblock \href {https://doi.org/10.1017/CBO9780511542879}
  {\path{doi:10.1017/CBO9780511542879}}.

\bibitem{agm}
Jonathan~M. Borwein and Peter~B. Borwein.
\newblock {\em Pi and the AGM : a study in analytic number theory and
  computational complexity}.
\newblock Canadian Mathematical Society series of monographs and advanced
  texts. Wiley, New York, 1987.

\bibitem{BorweinBorwein:88}
Jonathan~M. Borwein and Peter~B. Borwein.
\newblock On the complexity of familiar functions and numbers.
\newblock {\em SIAM Rev.}, 30(4):589--601, 1988.
\newblock \href {https://doi.org/10.1137/1030134} {\path{doi:10.1137/1030134}}.

\bibitem{geomet-tut}
Stephen Boyd, Seung-Jean Kim, Lieven Vandenberghe, and Arash Hassibi.
\newblock A tutorial on geometric programming.
\newblock {\em Optimization and Engineering}, 8(1):67--127, March 2007.
\newblock \href {https://doi.org/10.1007/s11081-007-9001-7}
  {\path{doi:10.1007/s11081-007-9001-7}}.

\bibitem{BCMW:17}
Peter B{\"{u}}rgisser, Matthias Christandl, Ketan Mulmuley, and Michael Walter.
\newblock Membership in moment polytopes is in {NP} and {coNP}.
\newblock {\em {SIAM} J. Comput.}, 46(3):972--991, 2017.
\newblock \href {https://doi.org/10.1137/15M1048859}
  {\path{doi:10.1137/15M1048859}}.

\bibitem{torus}
Peter B\"{u}rgisser, Levent Do\u{g}an, Visu Makam, Michael Walter, and Avi
  Wigderson.
\newblock {Polynomial Time Algorithms in Invariant Theory for Torus Actions}.
\newblock In Valentine Kabanets, editor, {\em 36th Computational Complexity
  Conference (CCC 2021)}, volume 200 of {\em Leibniz International Proceedings
  in Informatics (LIPIcs)}, pages 32:1--32:30, Dagstuhl, Germany, 2021. Schloss
  Dagstuhl -- Leibniz-Zentrum f{\"u}r Informatik.
\newblock URL: \url{https://drops.dagstuhl.de/opus/volltexte/2021/14306}, \href
  {https://doi.org/10.4230/LIPIcs.CCC.2021.32}
  {\path{doi:10.4230/LIPIcs.CCC.2021.32}}.

\bibitem{BFGOWW2}
Peter B{\"{u}}rgisser, Cole Franks, Ankit Garg, Rafael~Mendes de~Oliveira,
  Michael Walter, and Avi Wigderson.
\newblock Towards a theory of non-commutative optimization: geodesic first and
  second order methods for moment maps and polytopes.
\newblock In {\em 60th {A}nnual {IEEE} {S}ymposium on {F}oundations of
  {C}omputer {S}cience---{FOCS} 2019}, pages 845--861. IEEE Computer Soc., Los
  Alamitos, CA, 2019.
\newblock URL: \url{http://arxiv.org/abs/1910.12375}.

\bibitem{BFGOWW}
Peter B\"{u}rgisser, Cole Franks, Ankit Garg, Rafael Oliveira, Michael Walter,
  and Avi Wigderson.
\newblock Efficient algorithms for tensor scaling, quantum marginals, and
  moment polytopes.
\newblock In {\em 59th {A}nnual {IEEE} {S}ymposium on {F}oundations of
  {C}omputer {S}cience---{FOCS} 2018}, pages 883--897. IEEE Computer Soc., Los
  Alamitos, CA, 2018.
\newblock \href {https://doi.org/10.1109/FOCS.2018.00088}
  {\path{doi:10.1109/FOCS.2018.00088}}.

\bibitem{BGOWW}
Peter B\"{u}rgisser, Ankit Garg, Rafael Oliveira, Michael Walter, and Avi
  Wigderson.
\newblock Alternating minimization, scaling algorithms, and the null-cone
  problem from invariant theory.
\newblock In {\em 9th {I}nnovations in {T}heoretical {C}omputer {S}cience},
  volume~94 of {\em LIPIcs. Leibniz Int. Proc. Inform.}, pages Art. No. 24, 20.
  Schloss Dagstuhl. Leibniz-Zent. Inform., Wadern, 2018.

\bibitem{BLNW:20}
Peter B\"{u}rgisser, Yinan Li, Harold Nieuwboer, and Michael Walter.
\newblock Interior-point methods for unconstrained geometric programming and
  scaling problems.
\newblock arXiv:2008.12110, 2020.

\bibitem{ml-lse-1}
Giuseppe~C. Calafiore, Stephane Gaubert, and Corrado Possieri.
\newblock Log-sum-exp neural networks and posynomial models for convex and
  log-log-convex data.
\newblock {\em IEEE Transactions on Neural Networks and Learning Systems},
  31(3):827--838, 2020.
\newblock \href {https://doi.org/10.1109/TNNLS.2019.2910417}
  {\path{doi:10.1109/TNNLS.2019.2910417}}.

\bibitem{ml-lse-2}
Giuseppe~C. Calafiore, Stephane Gaubert, and Corrado Possieri.
\newblock A universal approximation result for difference of log-sum-exp neural
  networks.
\newblock {\em IEEE Transactions on Neural Networks and Learning Systems},
  31(12):5603--5612, 2020.
\newblock \href {https://doi.org/10.1109/TNNLS.2020.2975051}
  {\path{doi:10.1109/TNNLS.2020.2975051}}.

\bibitem{conway-sloane}
John~H. Conway and Neil J.~A. Sloane.
\newblock Low-dimensional lattices v. integral coordinates for integral
  lattices.
\newblock {\em Proceedings of the Royal Society of London. Series A,
  Mathematical and Physical Sciences}, 426(1871):211--232, 1989.
\newblock URL: \url{http://www.jstor.org/stable/2398341}.

\bibitem{derksen-kemper}
Harm Derksen and Gregor Kemper.
\newblock {\em Computational Invariant Theory}.
\newblock BV035421342 Encyclopaedia of Mathematical Sciences volume 130.
  Springer, Heidelberg ; New York ; Dordrecht ; London, second enlarged edition
  with two appendices by vladimir l. popov, and an addendum by norbert a. campo
  and vladimir l. popov edition, 2015.

\bibitem{DM-poly}
Harm Derksen and Visu Makam.
\newblock Polynomial degree bounds for matrix semi-invariants.
\newblock {\em Adv. Math.}, 310:44--63, 2017.
\newblock \href {https://doi.org/10.1016/j.aim.2017.01.018}
  {\path{doi:10.1016/j.aim.2017.01.018}}.

\bibitem{DM-oc}
Harm Derksen and Visu Makam.
\newblock Algorithms for orbit closure separation for invariants and
  semi-invariants of matrices.
\newblock {\em Algebra Number Theory}, 14(10):2791--2813, 2020.
\newblock \href {https://doi.org/10.2140/ant.2020.14.2791}
  {\path{doi:10.2140/ant.2020.14.2791}}.

\bibitem{derksen2021maximum}
Harm Derksen and Visu Makam.
\newblock Maximum likelihood estimation for matrix normal models via quiver
  representations.
\newblock {\em SIAM Journal on Applied Algebra and Geometry}, 5(2):338--365,
  2021.

\bibitem{derksen2022maximum}
Harm Derksen, Visu Makam, and Michael Walter.
\newblock Maximum likelihood estimation for tensor normal models via castling
  transforms.
\newblock In {\em Forum of Mathematics, Sigma}, volume~10, page e50. Cambridge
  University Press, 2022.

\bibitem{dinurcvp}
Irit Dinur, Guy Kindler, Ran Raz, and Shmuel Safra.
\newblock Approximating {CVP} to within-almost polynomial factors is {NP}-hard.
\newblock {\em Combinatorica}, 23(2):205--243, April 2003.
\newblock \href {https://doi.org/10.1007/s00493-003-0019-y}
  {\path{doi:10.1007/s00493-003-0019-y}}.

\bibitem{mono-testing-abc}
Kousha Etessami, Alistair Stewart, and Mihalis Yannakakis.
\newblock A note on the complexity of comparing succinctly represented
  integers, with an application to maximum probability parsing.
\newblock {\em ACM Trans. Comput. Theory}, 6(2), may 2014.
\newblock \href {https://doi.org/10.1145/2601327} {\path{doi:10.1145/2601327}}.

\bibitem{FS}
Michael~A. Forbes and Amir Shpilka.
\newblock Explicit {N}oether normalization for simultaneous conjugation via
  polynomial identity testing.
\newblock In {\em Approximation, randomization, and combinatorial
  optimization}, volume 8096 of {\em Lecture Notes in Comput. Sci.}, pages
  527--542. Springer, Heidelberg, 2013.
\newblock \href {https://doi.org/10.1007/978-3-642-40328-6_37}
  {\path{doi:10.1007/978-3-642-40328-6_37}}.

\bibitem{franks2021near}
Cole Franks, Rafael Oliveira, Akshay Ramachandran, and Michael Walter.
\newblock Near optimal sample complexity for matrix and tensor normal models
  via geodesic convexity.
\newblock {\em arXiv preprint arXiv:2110.07583}, 2021.

\bibitem{GGOW16}
Ankit Garg, Leonid Gurvits, Rafael Oliveira, and Avi Wigderson.
\newblock A deterministic polynomial time algorithm for non-commutative
  rational identity testing.
\newblock In {\em 57th {A}nnual {IEEE} {S}ymposium on {F}oundations of
  {C}omputer {S}cience---{FOCS} 2016}, pages 109--117. IEEE Computer Soc., Los
  Alamitos, CA, 2016.
\newblock \href {https://doi.org/10.1109/FOCS.2016.95}
  {\path{doi:10.1109/FOCS.2016.95}}.

\bibitem{GGOW:20}
Ankit Garg, Leonid Gurvits, Rafael Oliveira, and Avi Wigderson.
\newblock Operator scaling: theory and applications.
\newblock {\em Found. Comput. Math.}, 20(2):223--290, 2020.
\newblock \href {https://doi.org/10.1007/s10208-019-09417-z}
  {\path{doi:10.1007/s10208-019-09417-z}}.

\bibitem{homomorphic-crypto}
Craig Gentry.
\newblock Fully homomorphic encryption using ideal lattices.
\newblock In {\em Proceedings of the Forty-First Annual ACM Symposium on Theory
  of Computing}, STOC '09, pages 169--178, New York, NY, USA, 2009. Association
  for Computing Machinery.
\newblock \href {https://doi.org/10.1145/1536414.1536440}
  {\path{doi:10.1145/1536414.1536440}}.

\bibitem{goldfeld}
Dorian Goldfeld.
\newblock Beyond the last theorem.
\newblock {\em Math Horizons}, 4(1):26--34, 1996.
\newblock \href
  {http://arxiv.org/abs/https://doi.org/10.1080/10724117.1996.11974985}
  {\path{arXiv:https://doi.org/10.1080/10724117.1996.11974985}}, \href
  {https://doi.org/10.1080/10724117.1996.11974985}
  {\path{doi:10.1080/10724117.1996.11974985}}.

\bibitem{ggh-crypto}
Oded Goldreich, Shafi Goldwasser, and Shai Halevi.
\newblock Public-key cryptosystems from lattice reduction problems.
\newblock In Burton~S. Kaliski, editor, {\em Advances in Cryptology --- CRYPTO
  '97}, pages 112--131, Berlin, Heidelberg, 1997. Springer Berlin Heidelberg.

\bibitem{granvilletucker}
Andrew Granville and Thomas Tucker.
\newblock It's as easy as abc.
\newblock {\em Notices of the American Mathematical Society}, 49, 01 2002.

\bibitem{symplectic-in-physics}
Victor Guillemin and Shlomo Sternberg.
\newblock {\em Symplectic techniques in physics}.
\newblock Cambridge Univ. Press, Cambridge u.a., 1. publ., reprint. edition,
  1986.

\bibitem{Gur04}
Leonid Gurvits.
\newblock Classical complexity and quantum entanglement.
\newblock {\em J. Comput. Syst. Sci.}, 69(3):448--484, 2004.
\newblock \href {https://doi.org/10.1016/j.jcss.2004.06.003}
  {\path{doi:10.1016/j.jcss.2004.06.003}}.

\bibitem{hall-completion-2}
Marshall Hall.
\newblock Integral matrices $a$ for which $aa^t = mi$.
\newblock {\em Number Theory and Algebra}, pages 119 -- 134, 1977.
\newblock URL:
  \url{https://www.scopus.com/inward/record.uri?eid=2-s2.0-84887022329&partnerID=40&md5=5f91330ec7e4ddf22586cd2717a71b3a}.

\bibitem{hall-completion}
Marshall Hall.
\newblock Combinatorial completions.
\newblock In {\em Advances in Graph Theory}, volume~3 of {\em Annals of
  Discrete Mathematics}, pages 111--123. Elsevier, 1978.
\newblock URL:
  \url{https://www.sciencedirect.com/science/article/pii/S0167506008705016},
  \href {https://doi.org/https://doi.org/10.1016/S0167-5060(08)70501-6}
  {\path{doi:https://doi.org/10.1016/S0167-5060(08)70501-6}}.

\bibitem{hamada-hirai}
Masaki Hamada and Hiroshi Hirai.
\newblock Computing the nc-rank via discrete convex optimization on cat(0)
  spaces.
\newblock {\em SIAM Journal on Applied Algebra and Geometry}, 5(3):455--478,
  2021.
\newblock \href {http://arxiv.org/abs/https://doi.org/10.1137/20M138836X}
  {\path{arXiv:https://doi.org/10.1137/20M138836X}}, \href
  {https://doi.org/10.1137/20M138836X} {\path{doi:10.1137/20M138836X}}.

\bibitem{hardy-wright:08}
Godfrey~H. Hardy and Edward~M. Wright.
\newblock {\em An introduction to the theory of numbers}.
\newblock Oxford University Press, Oxford, sixth edition, 2008.
\newblock Revised by D. R. Heath-Brown and J. H. Silverman, With a foreword by
  Andrew Wiles.

\bibitem{himmelberg}
Charles~J. Himmelberg.
\newblock Pseudo-metrizability of quotient spaces.
\newblock {\em Fund. Math.}, pages 1--6, 1968.

\bibitem{ml-lse-3}
Warren Hoburg, Philippe Kirschen, and Pieter Abbeel.
\newblock Data fitting with geometric-programming-compatible softmax functions.
\newblock {\em Optimization and Engineering}, 17(4):897--918, December 2016.
\newblock \href {https://doi.org/10.1007/s11081-016-9332-3}
  {\path{doi:10.1007/s11081-016-9332-3}}.

\bibitem{ntru-crytpo-1}
Jeffrey Hoffstein, Jill Pipher, and Joseph~H. Silverman.
\newblock Ntru: A ring-based public key cryptosystem.
\newblock In Joe~P. Buhler, editor, {\em Algorithmic Number Theory}, pages
  267--288, Berlin, Heidelberg, 1998. Springer Berlin Heidelberg.

\bibitem{ntru-crypto-2}
Jeffrey Hoffstein and Joseph Silverman.
\newblock Optimizations for ntru.
\newblock In Kazimierz Alster, Jerzy Urbanowicz, and Hugh~C. Williams, editors,
  {\em Proceedings of the International Conference organized by the Stefan
  Banach International Mathematical Center Warsaw, Poland, September 11-15,
  2000}, pages 77--88, Berlin, New York, 2001. De Gruyter.
\newblock URL: \url{https://doi.org/10.1515/9783110881035.77} [cited
  2023-06-20], \href {https://doi.org/doi:10.1515/9783110881035.77}
  {\path{doi:doi:10.1515/9783110881035.77}}.

\bibitem{ivanyos-qiao}
G\'{a}bor Ivanyos and Youming Qiao.
\newblock {\em On the orbit closure intersection problems for matrix tuples
  under conjugation and left-right actions}, pages 4115--4126.
\newblock Society for Industrial and Applied Mathematics, 2023.
\newblock URL:
  \url{https://epubs.siam.org/doi/abs/10.1137/1.9781611977554.ch158}, \href
  {http://arxiv.org/abs/https://epubs.siam.org/doi/pdf/10.1137/1.9781611977554.ch158}
  {\path{arXiv:https://epubs.siam.org/doi/pdf/10.1137/1.9781611977554.ch158}},
  \href {https://doi.org/10.1137/1.9781611977554.ch158}
  {\path{doi:10.1137/1.9781611977554.ch158}}.

\bibitem{IQS}
G\'{a}bor Ivanyos, Youming Qiao, and K.~V. Subrahmanyam.
\newblock Non-commutative {E}dmonds' problem and matrix semi-invariants.
\newblock {\em Comput. Complexity}, 26(3):717--763, 2017.
\newblock \href {https://doi.org/10.1007/s00037-016-0143-x}
  {\path{doi:10.1007/s00037-016-0143-x}}.

\bibitem{IQS2}
G\'{a}bor Ivanyos, Youming Qiao, and K.~V. Subrahmanyam.
\newblock Constructive non-commutative rank computation is in deterministic
  polynomial time.
\newblock {\em Comput. Complexity}, 27(4):561--593, 2018.
\newblock \href {https://doi.org/10.1007/s00037-018-0165-7}
  {\path{doi:10.1007/s00037-018-0165-7}}.

\bibitem{KI04}
Valentine Kabanets and Russell Impagliazzo.
\newblock Derandomizing polynomial identity tests means proving circuit lower
  bounds.
\newblock {\em Computational Complexity}, 13(1-2):1--46, 2004.
\newblock \href {https://doi.org/10.1007/s00037-004-0182-6}
  {\path{doi:10.1007/s00037-004-0182-6}}.

\bibitem{kannan1987}
Ravi Kannan.
\newblock Minkowski's convex body theorem and integer programming.
\newblock {\em Mathematics of Operations Research}, 12(3):415--440, 1987.
\newblock URL: \url{http://www.jstor.org/stable/3689974}.

\bibitem{snf}
Ravindran Kannan and Achim Bachem.
\newblock Polynomial algorithms for computing the {S}mith and {H}ermite normal
  forms of an integer matrix.
\newblock {\em SIAM J. Comput.}, 8(4):499--507, 1979.
\newblock \href {https://doi.org/10.1137/0208040} {\path{doi:10.1137/0208040}}.

\bibitem{kempfness}
George Kempf and Linda Ness.
\newblock The length of vectors in representation spaces.
\newblock In Knud L{\o}nsted, editor, {\em Algebraic Geometry}, pages 233--243,
  Berlin, Heidelberg, 1979. Springer Berlin Heidelberg.

\bibitem{koiran-hn}
Pascal Koiran.
\newblock Hilbert's {N}ullstellensatz is in the {P}olynomial {H}ierarchy.
\newblock {\em Journal of Complexity}, 12(4):273--286, 1996.
\newblock \href {https://doi.org/https://doi.org/10.1006/jcom.1996.0019}
  {\path{doi:https://doi.org/10.1006/jcom.1996.0019}}.

\bibitem{kung-rota:84}
Joseph P.~S. Kung and Gian-Carlo Rota.
\newblock The invariant theory of binary forms.
\newblock {\em Bull. Amer. Math. Soc. (N.S.)}, 10(1):27--85, 1984.
\newblock \href {https://doi.org/10.1090/S0273-0979-1984-15188-7}
  {\path{doi:10.1090/S0273-0979-1984-15188-7}}.

\bibitem{kuperberg-knot}
Greg Kuperberg.
\newblock Knottedness is in {NP}, modulo {GRH}.
\newblock {\em Advances in Mathematics}, 256:493--506, 2014.
\newblock URL:
  \url{https://www.sciencedirect.com/science/article/pii/S0001870814000188},
  \href {https://doi.org/https://doi.org/10.1016/j.aim.2014.01.007}
  {\path{doi:https://doi.org/10.1016/j.aim.2014.01.007}}.

\bibitem{lang-elliptic-curves}
Serge Lang.
\newblock {\em Elliptic Curves: Diophantine Analysis}, volume 231 of {\em
  Grundlehren der mathematischen Wissenschaften}.
\newblock Springer, 2013.

\bibitem{leake-vishnoi}
Jonathan Leake and Nisheeth~K. Vishnoi.
\newblock On the computability of continuous maximum entropy distributions with
  applications.
\newblock In {\em Proceedings of the 52nd Annual ACM SIGACT Symposium on Theory
  of Computing}, STOC 2020, pages 930--943, New York, NY, USA, 2020.
  Association for Computing Machinery.
\newblock \href {https://doi.org/10.1145/3357713.3384302}
  {\path{doi:10.1145/3357713.3384302}}.

\bibitem{lll}
Arjen~K. Lenstra, Hendrik.~W. Lenstra, and L\'{a}szl\'{o} Lov\'{a}sz.
\newblock Factoring polynomials with rational coefficients.
\newblock {\em Mathematische Annalen}, 261(4):515--534, December 1982.
\newblock \href {https://doi.org/10.1007/BF01457454}
  {\path{doi:10.1007/BF01457454}}.

\bibitem{schaum-lin-alg}
Seymour Lipschutz.
\newblock {\em Schaum's outline of theory and problems of linear algebra :
  [including 600 solved problems ; completely solved in detail]}.
\newblock Schaum's outline series. McGraw-Hill, New York u.a., 6th ed. edition,
  1974.

\bibitem{MW19}
Visu Makam and Avi Wigderson.
\newblock Singular tuples of matrices is not a null cone (and the symmetries of
  algebraic varieties).
\newblock {\em J. Reine Angew. Math.}, 780:79--131, 2021.
\newblock \href {https://doi.org/10.1515/crelle-2021-0044}
  {\path{doi:10.1515/crelle-2021-0044}}.

\bibitem{ml-tropical}
Petros Maragos, Vasileios Charisopoulos, and Emmanouil Theodosis.
\newblock Tropical geometry and machine learning.
\newblock {\em Proceedings of the IEEE}, 109(5):728--755, 2021.
\newblock \href {https://doi.org/10.1109/JPROC.2021.3065238}
  {\path{doi:10.1109/JPROC.2021.3065238}}.

\bibitem{marsdenweinstein}
Jerrold Marsden and Alan Weinstein.
\newblock Reduction of symplectic manifolds with symmetry.
\newblock {\em Reports on Mathematical Physics}, 5(1):121--130, 1974.
\newblock URL:
  \url{https://www.sciencedirect.com/science/article/pii/0034487774900214},
  \href {https://doi.org/https://doi.org/10.1016/0034-4877(74)90021-4}
  {\path{doi:https://doi.org/10.1016/0034-4877(74)90021-4}}.

\bibitem{matveev}
E.~M. Matveev.
\newblock An explicit lower bound for a homogeneous rational linear form in
  logarithms of algebraic numbers.
\newblock {\em Izvestiya: Mathematics}, 62(4):723, aug 1998.
\newblock URL: \url{https://dx.doi.org/10.1070/IM1998v062n04ABEH000190}, \href
  {https://doi.org/10.1070/IM1998v062n04ABEH000190}
  {\path{doi:10.1070/IM1998v062n04ABEH000190}}.

\bibitem{micciancio}
Daniele Micciancio and Shafi Goldwasser.
\newblock {\em Complexity of lattice problems : a cryptographic perspective}.
\newblock The Kluwer international series in engineering and computer science
  BV000632170 671. Kluwer Academic, Boston, 2002.

\bibitem{miller-primality}
Gary~L. Miller.
\newblock Riemann's hypothesis and tests for primality.
\newblock {\em Journal of Computer and System Sciences}, 13(3):300--317, 1976.
\newblock URL:
  \url{https://www.sciencedirect.com/science/article/pii/S0022000076800438},
  \href {https://doi.org/https://doi.org/10.1016/S0022-0000(76)80043-8}
  {\path{doi:https://doi.org/10.1016/S0022-0000(76)80043-8}}.

\bibitem{mordell}
Louis~J. Mordell.
\newblock On the representation of a binary quadratic form as a sum of squares
  of linear forms.
\newblock {\em Mathematische Zeitschrift}, 35(1-15):1432--1823, 1932.
\newblock URL: \url{https://doi.org/10.1007/BF01186544}.

\bibitem{GCTV}
Ketan Mulmuley.
\newblock Geometric complexity theory {V}: {E}fficient algorithms for {N}oether
  normalization.
\newblock {\em J. Amer. Math. Soc.}, 30(1):225--309, 2017.
\newblock \href {https://doi.org/10.1090/jams/864}
  {\path{doi:10.1090/jams/864}}.

\bibitem{mulmuley2001geometric}
Ketan Mulmuley and Milind Sohoni.
\newblock Geometric complexity theory {I}: An approach to the {P} vs. {NP} and
  related problems.
\newblock {\em SIAM Journal on Computing}, 31(2):496--526, 2001.

\bibitem{mumford:red}
David Mumford.
\newblock {\em The red book of varieties and schemes}, volume 1358 of {\em
  Lecture Notes in Mathematics}.
\newblock Springer-Verlag, Berlin, 1999.
\newblock \href {https://doi.org/10.1007/b62130} {\path{doi:10.1007/b62130}}.

\bibitem{MFK:94}
David Mumford, John Fogarty, and Frances Kirwan.
\newblock {\em Geometric invariant theory}, volume~34 of {\em Ergebnisse der
  Mathematik und ihrer Grenzgebiete (2) [Results in Mathematics and Related
  Areas (2)]}.
\newblock Springer-Verlag, Berlin, third edition, 1994.
\newblock \href {https://doi.org/10.1007/978-3-642-57916-5}
  {\path{doi:10.1007/978-3-642-57916-5}}.

\bibitem{nemirovski-ellipsoid}
Arkadi Nemirovski and Uriel Rothblum.
\newblock On complexity of matrix scaling.
\newblock {\em Linear Algebra and its Applications}, 302-303:435--460, 1999.
\newblock URL:
  \url{https://www.sciencedirect.com/science/article/pii/S0024379599002128},
  \href {https://doi.org/https://doi.org/10.1016/S0024-3795(99)00212-8}
  {\path{doi:https://doi.org/10.1016/S0024-3795(99)00212-8}}.

\bibitem{oesterle}
Joseph Oesterl\'e.
\newblock Nouvelles approches du th\'eor\`eme de {Fermat}.
\newblock In {\em S\'eminaire Bourbaki : volume 1987/88, expos\'es 686-699},
  number 161-162 in Ast\'erisque. Soci\'et\'e math\'ematique de France, 1988.
\newblock talk:694.
\newblock URL: \url{http://www.numdam.org/item/SB_1987-1988__30__165_0/}.

\bibitem{pan-chen:99}
Victor~Y. Pan and Zhao~Q. Chen.
\newblock The complexity of the matrix eigenproblem.
\newblock In {\em Annual {ACM} {S}ymposium on {T}heory of {C}omputing
  ({A}tlanta, {GA}, 1999)}, pages 507--516. ACM, New York, 1999.
\newblock \href {https://doi.org/10.1145/301250.301389}
  {\path{doi:10.1145/301250.301389}}.

\bibitem{rabin-shallit}
Michael~O. Rabin and Jeffery~O. Shallit.
\newblock Randomized algorithms in number theory.
\newblock {\em Communications on Pure and Applied Mathematics},
  39(S1):S239--S256, 1986.
\newblock \href {https://doi.org/https://doi.org/10.1002/cpa.3160390713}
  {\path{doi:https://doi.org/10.1002/cpa.3160390713}}.

\bibitem{integer-farkas}
Alexander Schrijver.
\newblock {\em Theory of linear and integer programming}.
\newblock Wiley-Interscience series in discrete mathematics and optimization.
  Wiley, Chichester u.a., reprinted edition, 1999.

\bibitem{vishnoi-singh}
Mohit Singh and Nisheeth~K. Vishnoi.
\newblock Entropy, optimization and counting.
\newblock In {\em Proceedings of the Forty-Sixth Annual ACM Symposium on Theory
  of Computing}, STOC '14, pages 50--59, New York, NY, USA, 2014. Association
  for Computing Machinery.
\newblock \href {https://doi.org/10.1145/2591796.2591803}
  {\path{doi:10.1145/2591796.2591803}}.

\bibitem{stewart_oesterle-masser_1986}
Cameron~L. Stewart and R.~Tijdeman.
\newblock On the {Oesterl\'{e}}-{Masser} conjecture.
\newblock {\em Monatshefte f\"{u}r Mathematik}, 102(3):251--257, September
  1986.
\newblock \href {https://doi.org/10.1007/BF01294603}
  {\path{doi:10.1007/BF01294603}}.

\bibitem{yu-stewart}
Cameron~L. Stewart and Kunrui Yu.
\newblock On the abc conjecture.
\newblock {\em Mathematische Annalen}, 291(2):225--230, 1991.
\newblock URL: \url{http://eudml.org/doc/164860}.

\bibitem{vishnoi-straszak}
Damian Straszak and Nisheeth~K. Vishnoi.
\newblock Maximum entropy distributions: Bit complexity and stability.
\newblock In Alina Beygelzimer and Daniel~Hsu 0001, editors, {\em Conference on
  Learning Theory, COLT 2019, 25-28 June 2019, Phoenix, AZ, USA}, volume~99 of
  {\em Proceedings of Machine Learning Research}, pages 2861--2891. PMLR, 2019.
\newblock URL: \url{http://proceedings.mlr.press/v99/straszak19a.html}.

\bibitem{Sturmfels}
Bernd Sturmfels.
\newblock {\em Algorithms in Invariant Theory}.
\newblock Texts {\&} Monographs in Symbolic Computation. Springer, 2008.
\newblock \href {https://doi.org/10.1007/978-3-211-77417-5}
  {\path{doi:10.1007/978-3-211-77417-5}}.

\bibitem{vdpoorten}
Alfred~J. van~der Poorten.
\newblock On {B}aker's inequality for linear forms in logarithms.
\newblock {\em Mathematical Proceedings of the Cambdridge Philosophical
  Society}, 80(2):233--248, 1976.
\newblock \href {https://doi.org/10.1017/S0305004100052877}
  {\path{doi:10.1017/S0305004100052877}}.

\bibitem{waldschmidt-2}
Michel Waldschmidt.
\newblock {\em Diophantine approximation on linear algebraic groups :
  transcendence properties of the exponential function in several variables}.
\newblock Grundlehren der mathematischen Wissenschaften BV000000395 326.
  Springer, Berlin u.a., 2000.

\bibitem{waldschmidt-1}
Michel Waldschmidt.
\newblock Open diophantine problems.
\newblock {\em Mosc. Math.~J.}, 4:245--305, 2004.

\bibitem{abc-consequences}
Michel Waldschmidt.
\newblock Lecture on the abc conjecture and some of its consequences.
\newblock In Pierre Cartier, A.D.R. Choudary, and Michel Waldschmidt, editors,
  {\em Mathematics in the 21st Century}, pages 211--230, Basel, 2015. Springer
  Basel.

\bibitem{wustholz1993}
Gisbert W\"{u}stholz and Alan Baker.
\newblock Logarithmic forms and group varieties.
\newblock {\em Journal f\"{u}r die reine und angewandte Mathematik},
  442:19--62, 1993.
\newblock URL: \url{http://eudml.org/doc/153550}.

\end{thebibliography}

\end{document}